\documentclass[11pt]{article}
\usepackage{amsmath}
\usepackage{amssymb}
\usepackage{amsthm}
\usepackage[backend=biber,style=alphabetic,doi=false]{biblatex}
\usepackage{url}
\usepackage{hyperref}
\usepackage{xcolor}
\usepackage{dsfont}
\usepackage{float}
\usepackage{mathrsfs}
\usepackage{yfonts}
\newtheorem{definition}{Definition}
\newtheorem*{example}{Example}

\newtheorem{remark}[definition]{Remark}
\newtheorem{theorem}[definition]{Theorem}
\newtheorem{lemma}[definition]{Lemma}
\newtheorem{corollary}[definition]{Corollary}

\usepackage{graphicx}
\usepackage{caption}
 \usepackage{subcaption}
 \usepackage{verbatim}
\usepackage[title]{appendix}
\usepackage{upgreek}
\usepackage{tikz}
\usepackage{quantikz}
\usetikzlibrary{shapes.geometric,plotmarks,backgrounds,fit,calc,circuits.ee.IEC}
\usetikzlibrary {arrows.meta}
\usetikzlibrary{patterns}
\usetikzlibrary{decorations.pathreplacing}

\usepackage{geometry}
\bibliography{references}
\newcommand{\bL}{\mathbf{L}}

\newcommand{\bF}{\mathbb{F}}
\newcommand{\bC}{\mathbb{C}}

\newcommand{\cC}{{\cal C}}
\newcommand{\cD}{{\cal D}}
\newcommand{\cE}{{\cal E}}

\newcommand{\cG}{{\cal G}}
\newcommand{\cH}{{\cal H}}
\newcommand{\cI}{{\cal I}}

\newcommand{\cN}{{\cal N}}

\newcommand{\cT}{{\cal T}}
\newcommand{\cU}{{\cal U}}
\newcommand{\cV}{{\cal V}}
\newcommand{\cW}{{\cal W}}

\newcommand{\cZ}{{\cal Z}}

\newcommand{\bA}{\mathbf{A}}

\newcommand{\be}{\mathbf{e}}

\newcommand{\bm}{\mathbf{m}}

\newcommand{\bu}{\mathbf{u}}
\newcommand{\bv}{\mathbf{v}}

\newcommand{\by}{\mathbf{y}}

\newcommand{\ovr}{\overline{r}}
\newcommand{\on}{\overline{n}}
\newcommand{\om}{\overline{m}}

\newcommand{\cnot}{\mathrm{CNOT}}

\newcommand{\ident}{\mathds{1}}
\newcommand{\dnorm}[1]
{{\ensuremath{\|#1\|_\diamond}}}

\newcommand{\ot}{\overline{T}}

\newcommand{\odelta}{\overline{\delta}}

\usepackage{braket}

\newtheorem{notation}{Notation}

\DeclareMathOperator{\tr}{Tr}
\DeclareMathOperator{\poly}{poly}
\DeclareMathOperator{\pr}{Pr}
\DeclareMathOperator{\loc}{Loc}

\DeclareMathOperator{\wt}{wt}

 \tikzset{
operator/.style = {draw,fill=white,minimum size=1.5em},
operator2/.style = {draw,fill=white,minimum height=3cm},
phase/.style = {draw,fill,shape=circle,minimum size=5pt,inner sep=0pt},
surround/.style = {fill=blue!10,thick,draw=black,rounded corners=2mm},
cross/.style={path picture={ 
\draw[thick,black](path picture bounding box.north) -- (path picture bounding box.south) (path picture bounding box.west) -- (path picture bounding box.east);
}},
crossx/.style={path picture={ 
\draw[thick,black,inner sep=0pt]
(path picture bounding box.south east) -- (path picture bounding box.north west) (path picture bounding box.south west) -- (path picture bounding box.north east);
}},
circlewc/.style={draw,circle,cross,minimum width=0.3 cm},
}
\tikzset{meter/.append style={draw, inner sep=8, rectangle, font=\vphantom{A}, minimum width=25, line width=.8,
 path picture={\draw[black] ([shift={(.1,.3)}]path picture bounding box.south west) to[bend left=50] ([shift={(-.1,.3)}]path picture bounding box.south east);\draw[black,-latex] ([shift={(0,.1)}]path picture bounding box.south) -- ([shift={(.3,-.1)}]path picture bounding box.north);}}}

\title{Fault-tolerant interfaces for quantum LDPC codes}

\author{Matthias Christandl$^1$, Omar Fawzi$^2$, and Ashutosh Goswami$^{1}$ \\[2mm]
    {\small $^1$Department of Mathematical Sciences, University of Copenhagen, Denmark} \\
    {\small $^2$Universit\'e de Lyon, Inria, ENS de Lyon, UCBL, LIP, France}\\
  {\small christandl@math.ku.dk, omar.fawzi@ens-lyon.fr, akg@math.ku.dk}}

\date{}
\begin{document}

\maketitle

\begin{abstract}
The preparation of a quantum state using a noisy quantum computer (gate noise strength $\delta$) will necessarily affect an O($\delta$)-fraction of the qubits, no matter which protocol is used. Here, we show that fault-tolerant quantum state preparation can be achieved with constant qubit overhead improving on previous constructions requiring polylogarithmic overhead.

To achieve this, we add to the toolbox of fault-tolerant schemes for circuits with quantum input and output. More specifically, we construct fault-tolerant interfaces that decrease the level of protection for quantum low-density parity-check (LDPC) codes. When information is encoded in multiple code blocks, our interfaces have constant qubit overhead. 

In our decoder construction that changes the level of protection by an arbitrary amount, we circumvent bottlenecks to error pileup and overhead by gradual lowering of the level of encoding at the same time as we increase the number of blocks on which decoding is carried out simultaneously. 
\end{abstract}


\section{Introduction}
Fault-tolerant quantum computation is typically motivated from the perspective of realizing quantum computations that have classical input and output, e.g. in the computation of a function. One of the main results in the field of fault-tolerant quantum computation is the \emph{threshold theorem}, which states that any quantum circuit with classical input and output can be performed reliably, using a fault-tolerant circuit, provided that the noise rate is below a constant, i.e., an error rate not depending on the number of locations (gates) in the (original) circuit~\cite{aharonov1997fault, kitaev, preskill1998fault, Aliferis2006quantum}. 

\smallskip While the threshold theorem addresses quantum computations with classical input and output, many important quantum information tasks involve using circuits that have quantum input/output. Examples include quantum communication protocols, quantum learning tasks, and, more generally, information processing tasks that involve interactions with black boxes. Such black boxes may have inputs and outputs that are quantum and encoded in a way that we do not control. In such settings, the notion of fault-tolerance must be adapted. In particular, a fault-tolerant implementation must preserve the original circuit’s input/output relation in the encoding fixed by the black boxes. As the quantum inputs and outputs are affected by noise, the natural guarantee to aim for is that the fault-tolerant circuit realizes the original one up to a weak noise applied on input and output systems.

A framework for fault-tolerant circuits with quantum input/output was presented in Ref.~\cite{christandl2024fault} based on the fault-tolerance formalism of~\cite{kitaev}. This framework builds on the fault-tolerance literature and unifies and improves on previous specific results such as state preparation circuits~\cite{gottesman2013fault} and preparation and measurement interfaces~\cite{CChMH-FT2022,belzig2023fault}. The framework was used to construct a constant qubit overhead fault-tolerant scheme for general noise~\cite{christandl2025fault} and to construct a fault-tolerant QRAM~\cite{dalzell2025distillation}.

An important primitive in fault-tolerant quantum input/output is the fault-tolerant preparation of an $n$-qubit quantum state described by a circuit. This primitive was achieved in~\cite{gottesman2013fault, CChMH-FT2022, belzig2023fault, christandl2024fault} using concatenated codes and thus incurs a qubit overhead that grows polylogarithmically with $n$.

\smallskip In this work, we ask  whether it is possible to achieve fault-tolerant realization of state preparation circuits, with constant qubit overhead. It is worth emphasizing that the fault-tolerant circuit must preserve the same input-output relation as the original circuit, i.e., it should prepare the target state up to a weak noise acting on it, rather than an encoded version of that state in a quantum error-correcting code. 

\smallskip We answer this question positively in the model of circuit-level stochastic noise, where the noise acts independently on each gate. The noisy version  of a  gate $g$, under circuit-level stochastic noise with error rate $\delta \in [0, 1]$, realizes the gate $g$ with probability $(1 - \delta)$, while with probability $\delta$, it realizes an arbitrary quantum gate $\tilde{g}$ with the same input and output dimensions as $g$.

\smallskip In our first main result, we show that any state preparation circuit can be realized fault-tolerantly up to a local stochastic noise, while incurring only a constant qubit overhead. Our state preparation is according to the following theorem. 

\begin{theorem}[Informal version of Theorem \ref{thm:main-stprep}] \label{thm:main-stprep-informal} 
There exists a threshold value $\delta_{th} > 0$ and a constant $\kappa>0$ such that the following holds. Consider a state preparation circuit $\Phi$, with $x$ qubit output and operating on $O(x)$ qubits and  having size $|\Phi| = \mathrm{poly}(x)$. Then, there exists a quantum circuit $\overline{\Phi}$, with the same input and output systems as $\Phi$ and working on $O(x)$ qubits, such that its noisy realization $\tilde{\cT}_{\overline{\Phi}}$ under circuit level stochastic noise with parameter $\delta < \delta_{th}$ approximately simulates $\Phi$ up to a layer of local stochastic noise  with strength $\delta'=O(\delta^{\kappa})$ applied at the output.
\end{theorem}

\smallskip In order to establish this result, we consider quantum low-density parity-check (QLDPC) codes. The reason for this is as follows: first, they encode logical qubits at constant rate, \emph{i.e.}, with a constant qubit overhead, and  have  good error correction performance and efficient decoders~\cite{tillich2013quantum, leverrier2015quantum, fawzi2018efficient,  breuckmann2021balanced, breuckmann2021quantum}. Moreover, there are families of QLDPC codes with a minimum distance scaling linearly with the code length~\cite{leverrier2022quantum, panteleev2022asymptotically}. Second, a constant qubit overhead scheme for fault-tolerant computation has been constructed using QLDPC codes~\cite{gottesman2013fault, fawzi2020constant, christandl2025fault}. Moreover, progress in recent years has made QLDPC codes a leading path towards fault-tolerant quantum computing~\cite{yoder2025tour, he2025extractors, cross2024improved, xu2025batched, bonilla2025constant, webster2025explicit, webster2026thepinnacle}.

\smallskip We construct a decoding interface circuit, which fault-tolerantly maps logical information encoded in multiple blocks of a QLDPC code into bare physical qubits, up to a local stochastic noise. Moreover, if the number of blocks is sufficiently large, the qubit overhead of the decoding interface circuit is constant. Our decoding interface, which is crucial for fault-tolerant state preparation in Theorem~\ref{thm:main-stprep-informal}, is as follows.

\begin{theorem}[Informal version of Theorem \ref{thm:ft-cons-int-main}] \label{thm:ft-cons-int-main-inf}
For certain families of QLDPC codes $C_r, r= 1, 2, \dots$, where $\cC_r$ encodes $m_r$ logical qubits in $n_r$ physical qubits, the following holds. There exists a polynomial $p(\cdot)$, constants $\mu, \kappa > 0$ and a threshold value $\delta_{th} > 0$, such that for all encoding levels $r$, quantum information encoded in $h$ blocks of $C_r$ can be decoded into bare physical qubits with a decoding interface circuit affected by circuit-level stochastic noise of strength $\delta < \delta_{th}$. The error on the output is a local stochastic channel with strength $\delta'=O(\delta^{\kappa})$ if the input has at most $\mu n_r$ errors in each block. The qubit overhead of the decoding interface is constant for $h \geq p(m_r)$.
\end{theorem}

\smallskip We emphasize that our methods to obtain Theorem~\ref{thm:main-stprep-informal} and Theorem~\ref{thm:ft-cons-int-main-inf} are significantly different from methods used in~\cite{gottesman2013fault} to obtain constant qubit overhead fault-tolerant quantum computing.  In fact, a straightforward generalization of the protocols in~\cite{gottesman2013fault} will not lead to a fault-tolerant decoding interface. To highlight this difference, we below first describe the construction due to~\cite{gottesman2013fault} and outline the challenges in generalizing it to state preparation circuits. We then present our own construction and highlight the main innovations.

\paragraph{Description of Gottesman's constant qubit overhead scheme for quantum computation:}  The fault-tolerant construction in~\cite{gottesman2013fault} considers sequential circuits, such that at any layer only one or a few non-trivial gates are applied on a selected set of logical qubits, while the idle gate is applied on the remaining logical qubits. Any quantum circuit $\Phi$, operating on $x'$ qubits, can be turned into a sequential circuit $\Phi_{\mathrm{seq}}$ by incurring only a $O(x')$ overhead in the size (i.e., total number of locations) and depth of the circuit. We note that this sequentialization is key to ensuring constant qubit overhead.

\smallskip Consider the sequential circuit $\Phi_{\mathrm{seq}}$, operating on $x'$ logical qubits. Consider a QLDPC code family $\cC_r, r= 1, \dots$, encoding $m_r$ logical qubits in $n_r$ physical qubits such that the encoding rate $\frac{m_r}{n_r}$ is constant. We first divide $x'$ logical qubits into $h = \frac{x'}{m_r}$ blocks, each containing $m_r$ logical qubits. For the fault-tolerant circuit $\Phi_{\mathrm{FT}}$, we encode each block of $m_r$ logical qubits in a code block of QLDPC code $\cC_r$, containing $n_r$ physical qubits. We replace a layer of gates in $\Phi_{\mathrm{seq}}$ by a logical layer in $\Phi_{\mathrm{FT}}$  as follows:

\smallskip If the idle gate is applied on a block in $\Phi_{\mathrm{seq}}$, we apply the error correction on the corresponding block in $\Phi_{\mathrm{FT}}$. If a non-trivial gate is applied on some of the qubits in a block in $\Phi_{\mathrm{seq}}$, we apply the corresponding logic gate on the corresponding block in $\Phi_{\mathrm{FT}}$. A logic gate on a QLDPC code block is applied using \emph{gate teleportation} technique on the corresponding logical block~\cite{gottesman2013fault, fawzi2020constant, grospellier2019constant}. 

\smallskip The error correction circuit for QLDPC codes operates on $O(m_r)$ qubits; therefore, the qubit overhead of $\Phi_{\mathrm{EC}^r}$ is constant for the blocks on which error correction is applied. However, for the logic gate implementation using gate teleportation, the overhead is not constant with respect to the block on which it is applied. But since we only need one non-trivial gate, by choosing $m_r$ to be sub-linearly growing with $x'$, for example $m_r = \sqrt{x'}$, the total overhead can be made constant with respect to the total number of qubits $x'$.

\paragraph{Challenges in generalizing constant qubit overhead scheme to state preparation:} 
As mentioned before, we need a decoding interface to extend constant qubit overhead scheme to state preparation. We note that in principle, a teleportation based decoding interface can be obtained using teleportation as depicted in Fig~\ref{fig:dec-tele}.  Let  $V_r$ be the encoding isometry of $\cC_r$ mapping $m_r$ qubits into $n_r$ qubits. One can teleport any code state of $\cC_r$ out of the code space by using the entangled state 
\begin{equation} \label{eq:ent-state-int}
  \ket{\Psi}_{r} := (V_r \otimes \ident_{m_r}) \left(\frac{\ket{0} \ket{0} + \ket{1} \ket{1}}{\sqrt{2}}\right)^{\otimes m_r},   
\end{equation}
where $V_r$ acts on $m_r$ qubits corresponding to the first half of $m_r$-EPR pairs.

\begin{figure}
\centering
\begin{tikzpicture}[thick,>=stealth]

  \tikzstyle{qubit}=[line width=0.8pt];
  \tikzstyle{cbit}=[double, line width=0.6pt];

  \def\ypsi{1.0}
  \def\yEPRa{0.2}
  \def\yEPRb{-0.6}

  \draw[qubit] (-1,\ypsi) node[left] {$\ket{\psi_r}$} -- (0,\ypsi);

  \draw[qubit] (-1,\yEPRa) -- (0,\yEPRa);
  \draw[qubit] (-1,\yEPRb) -- (0,\yEPRb);

  \draw (-1,\yEPRa) -- (-1.4,-0.2);
  \draw (-1,\yEPRb) -- (-1.4,-0.2);
  \node[left] at (-1.4,-0.2) {$\ket{\Psi_r}$};

  \draw (0,-0.1) rectangle (2,1.3);
  \node at (1,0.6) {Bell meas.};

  \draw[qubit] (0,\ypsi) -- (0,\ypsi);
  \draw[qubit] (0,\yEPRa) -- (0,\yEPRa);

  \draw[qubit] (0,\yEPRb) -- (5.0,\yEPRb);

  \draw[cbit] (2,\ypsi)  -- (3.0,\ypsi);
  \draw[cbit] (2,\yEPRa) -- (3.0,\yEPRa);

  \draw (3.0,\yEPRa-0.2) rectangle (4.3,\ypsi+0.2);
  \node at (3.65,0.6) {\scriptsize classical};
  \node at (3.65,0.35) {\scriptsize proc.};

  \draw[cbit] (4.3,\ypsi)  -- (5.4,\ypsi);
  \draw[cbit] (4.3,\yEPRa) -- (5.4,\yEPRa);

  \fill (5.4,\ypsi)  circle (1.5pt);
  \fill (5.4,\yEPRa) circle (1.5pt);

  \draw[cbit] (5.4,\ypsi) -- (5.4,\yEPRa);
  \draw[cbit] (5.4,\yEPRa) -- (5.4,\yEPRb+0.4);

  \draw (5.0,\yEPRb-0.4) rectangle (5.8,\yEPRb+0.4);
  \node at (5.4,\yEPRb) {$P$};

  \draw[qubit] (5.8,\yEPRb) -- (7.0,\yEPRb);

\end{tikzpicture}

\caption{
Teleportation-based decoding interface. A code state is teleported out of the code space by performing a logical Bell measurement between the input state $\ket{\psi_r}$ and one half of the entangled  resource state $\ket{\Psi_r}$ as defined in Eq.~\eqref{eq:ent-state-int}. Based on the two classical outcomes of the Bell measurement a Pauli correction $P$ is
applied to the remaining half of $\ket{\Psi_r}$.
}
\label{fig:dec-tele}
\end{figure}
Two main challenges in generalizing constant qubit overhead scheme to decoding interface are: 
\begin{itemize}
    \item[(i)]  the teleportation-based interface is not fault-tolerant itself,

    \item[(ii)] the sequentialization procedure is not robust for the decoding interface.
\end{itemize}
We first describe challenge (i).

We note that teleportation involves logical Bell measurement on two code blocks of the code $\cC_r$ (see Fig.~\ref{fig:dec-tele}). Logical Bell measurements correspond to transversal Bell measurements on physical qubits of the two code blocks~\cite[Chapter 6]{grospellier2019constant}. The outcome of the physical Bell measurements are further classically processed to get the outcome of logical Bell measurement, based on which a Pauli correction is applied on the teleported, unencoded state. The main obstacle here is that the time complexity of this classical processing depends on the code length $n_r$ of $\cC_r$. For QLDPC codes, the algorithms for classical processing scale $O(\mathrm{polylog}(n_r))$~\cite{fawzi2018efficient, grospellier2019constant, gu2024single, leverrier2023decoding}. This implies that teleported unencoded qubits need to wait for $O(\mathrm{polylog}(n_r))$ time steps. Since there is no active error correction on the unencoded qubits, the quantum information will eventually be garbled up as $r \to \infty$. We emphasize that for a fixed $r$, the waiting time is not a problem as long as $\delta$ is sufficiently small; the problem arises only when considering a sequence of codes $\cC_r, r=1, 2, \dots$, and a fixed error rate $\delta$.

\smallskip We note that the teleportation based interface does not have a constant qubit overhead due to fault-tolerant preparation of the state $\ket{\Psi_r}$, which is done using concatenated codes as in gate teleportation.  
\begin{figure}[!t]
    \centering
\begin{tikzpicture}

\draw
(-0.2, 0) node[isosceles triangle, draw, minimum size = 1cm](a){$\Gamma_r$}
(0, -2) node[draw, minimum size = 1cm](b){$\Phi_{\mathrm{EC}^{r, s}}$}

(0, -5) node[draw, minimum size = 1cm](c){$\Phi_{\mathrm{EC}^{r, s}}$}
;

\draw
(a.west) to ++(-0.8, 0)
($(a.west) +(-0.6, 0.15)$) to node[above](){$n_r$} ($(a.west) +(-0.4, -0.15)$)
(b.west) to ++(-0.8, 0)
($(b.west) +(-0.6, 0.15)$) to node[above](){$n_r$} ($(b.west) +(-0.4, -0.15)$)
(c.west) to ++(-0.8, 0)
($(c.west) +(-0.6, 0.15)$) to node[above](){$n_r$} ($(c.west) +(-0.4, -0.15)$)
($0.5*(b) + 0.5*(c)$) node[](){$\vdots$} 
;

\draw
(2.0, 0) node[draw, minimum height = 1.2cm] (d){{\scriptsize $\cI$}}

(3, 0) node[draw, minimum height = 1.2cm] (e){{\scriptsize $\cI$}}

(2.2, -2) node[isosceles triangle, draw, minimum size = 1cm](f){$\Gamma_r$}

(2.4, -5) node[draw, minimum size = 1cm](g){$\Phi_{\mathrm{EC}^{r, s}}$}
;

\draw
(a.east) to (d.west)
($(a.east) +(0.2, 0.15)$) to node[above](){$m_r$} ($(a.east) +(0.4, -0.15)$)
(b.east) to (f.west)
($(b.east) +(0.4, 0.15)$) to node[above](){$n_r$} ($(b.east) +(0.6, -0.15)$)
(c.east) to (g.west)
($(c.east) +(0.4, 0.15)$) to node[above](){$n_r$} ($(c.east) +(0.6, -0.15)$)

($0.5*(d) + 0.5*(e)$) node[](){{\scriptsize$\dots$}}
($0.5*(f) + 0.5*(g)$) node[](){$\vdots$}

(e.east) to ++(0.8, 0)
($(e.east) +(0.4, 0.15)$) to node[above](){$m_r$} ($(e.east) +(0.6, -0.15)$)
(f.east) to ++(0.8, 0)
($(f.east) +(0.4, 0.15)$) to node[above](){$m_r$} ($(f.east) +(0.6, -0.15)$)
(g.east) to ++(0.8, 0)
($(g.east) +(0.4, 0.15)$) to node[above](){$n_r$} ($(g.east) +(0.6, -0.15)$)
;

\draw
(6, 0) node[draw, minimum height = 1.2cm] (h){{\scriptsize $\cI$}}

(7, 0) node[draw, minimum height = 1.2cm] (i){{\scriptsize $\cI$}}

(6, -2) node[draw, minimum height = 1.2cm] (j){{\scriptsize $\cI$}}

(7, -2) node[draw, minimum height = 1.2cm] (k){{\scriptsize $\cI$}}

(6.2, -5) node[isosceles triangle, draw, minimum size = 1cm](l){$\Gamma_r$}
;

\draw
(h.west) to ++(-0.8, 0)
($(h.west) +(-0.6, 0.15)$) to node[above](){$m_r$} ($(h.west) +(-0.4, -0.15)$)
(j.west) to ++(-0.8, 0)
($(j.west) +(-0.6, 0.15)$) to node[above](){$m_r$} ($(j.west) +(-0.4, -0.15)$)
(l.west) to ++(-0.8, 0)
($(l.west) +(-0.6, 0.15)$) to node[above](){$n_r$} ($(l.west) +(-0.4, -0.15)$)

($0.5*(h) + 0.5*(i)$) node[](){{\scriptsize$\dots$}}
($0.5*(j) + 0.5*(k)$) node[](m){{\scriptsize$\dots$}}
($0.5*(m) + 0.5*(l)$) node[](m){$\vdots$}

($0.5*(e) + 0.5*(h)$) node[](){$\dots$}
($0.5*(f) + 0.5*(j) + (0.5, 0)$) node[](){$\dots$}
($0.5*(g) + 0.5*(l) + (0.2, 0)$) node[](){$\dots$}

(i.east) to ++(0.8, 0)
($(i.east) +(0.4, 0.15)$) to node[above](){$m_r$} ($(i.east) +(0.6, -0.15)$)
(k.east) to ++(0.8, 0)
($(k.east) +(0.4, 0.15)$) to node[above](){$m_r$} ($(k.east) +(0.6, -0.15)$)
(l.east) to ++(0.8, 0)
($(l.east) +(0.4, 0.15)$) to node[above](){$m_r$} ($(l.east) +(0.6, -0.15)$)

;

\end{tikzpicture}
\caption{The figure represents the construction of a constant qubit overhead interface using sequential implementation of a non-constant qubit overhead interface $\Gamma_r$. Each horizontal wire represents either a block of $n_r$ qubits or a block of $m_r$ qubits as written on top of the wire. In the first layer, $\Gamma_r$ is applied on the first wire mapping $n_r$ qubits to $m_r$ qubits. While $\Gamma_r$ is being applied on the first wire, error correction steps are applied on the remaining wires. On the second layer, $\Gamma_r$ is applied on the second wire, while on the first wire multiple layers of idle gate is applied, and on the wires $i = 3, 4, \dots$ correction steps are applied. Similarly, on $j^{th}$ layer for $j \leq h$, where $h$ is the total number of wires,  $\Gamma_r$ is applied on the $j^{th}$ wire, and on all the wires before it $i = 1, \dots, j-1$, idle gates are applied and on all the wires after it $i = j+1, \dots, h$ error correction steps are applied.}
\label{fig:con-ovr-intf}
\end{figure}

We now consider challenge $\mathrm{(ii)}$.

\smallskip \noindent We suppose that there is an interface circuit $\Gamma_{r}$ (not necessarily using teleportation) that maps logical information encoded in $\cC_r$ onto bare physical qubits fault-tolerantly (i.e., up to a weak noise on the output), but requires a qubit overhead that grows with $n_r$, i.e., the number of physical qubits in $\cC_r$. Suppose further that we want to upgrade the interface to a constant qubit overhead interface using sequentialization technique used for logic gate implementation in Gottesman's scheme.

\smallskip As depicted in Fig.~\ref{fig:con-ovr-intf}, we consider $h$ code blocks of $\cC_r$, and apply the interface $\Gamma_r$ on only one (or a few) of the chosen blocks in a layer, and continue in this way until the interface has been applied on all $h$ blocks. We choose the code parameter $r$ such that $m_r$ grows sub-linearly with respect to $x'$, i.e., the total number of logical qubits. This in turn implies that $h = \tfrac{x'}{m_r}$ grows (i.e., not constant) with respect to $x'$. This sequentialization procedure does not lead to a fault-tolerant decoding interface, as we now explain:

\smallskip Once the decoding interface is applied on a block, the logical qubits therein are not anymore protected by the code. However, such an unprotected block needs to wait until the interface has been applied on the remaining blocks. Since $h$ is growing with $x'$, the unprotected block of $m_r$ qubits needs to wait for $t$ time steps that grows with $x'$ for the remaining blocks to complete. As $t$ is not constant, the logical information encoded in the unprotected block will be completely lost as $x' \to \infty$.

\begin{figure}[t]
  \centering
  \begin{tikzpicture}[x=0.8cm,y=0.9cm,every node/.style={font=\small}]

    \node at (-1.8,0) {$\cC_r$};

    \foreach \i/\x in {0/0,1/2.4,2/4.8,3/7.2} {
      \draw (\x,0) rectangle ++(2,0.6);
    }
    \fill[blue!40] (0,0) rectangle ++(2,0.6);

    \node at (4.0,1.0) {$y=0$};

    \draw[->] (4.0,-0.1) -- (4.0,-0.7)
      node[midway,right,xshift=2pt] {$\Gamma_{r,r-1}$};

    \node at (-1.8,-1.8) {$\cC_{r-1}$};

    \foreach \x in {0,1.1} {
      \draw (\x,-1.8) rectangle ++(0.9,0.5);
    }
    \foreach \x in {2.4,3.5} {
      \draw (\x,-1.8) rectangle ++(0.9,0.5);
    }
    \foreach \x in {4.8,5.9} {
      \draw (\x,-1.8) rectangle ++(0.9,0.5);
    }
    \foreach \x in {7.2,8.3} {
      \draw (\x,-1.8) rectangle ++(0.9,0.5);
    }

    \foreach \x in {0,1.1,2.4,3.5} {
      \fill[blue!40] (\x,-1.8) rectangle ++(0.9,0.5);
    }

    \node at (4.0,-0.9) {$y=1$};

    \draw[->] (4.0,-1.9) -- (4.0,-2.5)
      node[midway,right,xshift=2pt] {$\Gamma_{r-1,r-2}$};

    \node at (-1.8,-3.6) {$\cC_{r-2}$};

    \foreach \x in {0,0.55,1.10,1.65} {
      \draw (\x,-3.6) rectangle ++(0.35,0.5);
    }
    \foreach \x in {2.4,2.95,3.50,4.05} {
      \draw (\x,-3.6) rectangle ++(0.35,0.5);
    }
    \foreach \x in {4.8,5.35,5.90,6.45} {
      \draw (\x,-3.6) rectangle ++(0.35,0.5);
    }
    \foreach \x in {7.2,7.75,8.30,8.85} {
      \draw (\x,-3.6) rectangle ++(0.35,0.5);
    }

    \foreach \x in {0,0.55,1.10,1.65,
                    2.4,2.95,3.50,4.05,
                    4.8,5.35,5.90,6.45} {
      \fill[blue!40] (\x,-3.6) rectangle ++(0.35,0.5);
    }

    \node at (4.0,-2.7) {$y=2$};

    \draw (0,-5.3) rectangle +(0.8,0.5);
    \node[anchor=west] at (1.1,-5.1) {idle / only EC};

    \fill[blue!40] (0,-6.2) rectangle +(0.8,0.5);
    \node[anchor=west] at (1.1,-6.0) {interface $\Gamma_{r-y,r-y-1}$ applied};
  \end{tikzpicture}
  \caption{constant qubit overhead interface: sequential application of the partial interfaces $\Gamma_{r-y,r-y-1}$ for $y = 0, 1, 2$ and $h = 4$. As the level decreases from $r$ to $r-1$ to $r-2$,  the number of blocks increases $(4, 8, 16)$, and the fraction of blocks on which the interface is applied in parallel (shaded) also grows $(1/4, 1/2, 3/4)$.}
\label{fig:sequential-interface}
\end{figure}

\smallskip We now present our construction.
\paragraph{Fault-tolerant decoding interface with constant qubit overhead:}
Consider $h$ code blocks of $\cC_r$. Instead of mapping the information completely out of a given code block, we consider a partial decoding interface $\Gamma_{r, r-1}$, which maps logical qubits encoded in one block of $\cC_r$ into two blocks of $\cC_{r-1}$. The interface $\Gamma_{r, r-1}$ is based on teleportation; therefore, requires fault-tolerant preparation of an entangled resource state. Hence, $\Gamma_{r, r-1}$ does not have a constant qubit overhead. We note that the interface $\Gamma_{r, r-1}$ avoids challenge (i) from above since the teleported information is still encoded in the code space of $\cC_{r-1}$ on which error correction is applied while the classical processing is done on physical Bell measurements. 
\smallskip To achieve constant qubit overhead, we consider a different sequentialization procedure as explained below (see also Fig.~\ref{fig:sequential-interface}).

\smallskip\noindent Consider $h$ blocks of $\cC_r$. To keep the overhead constant, we choose a small fraction of blocks of $\cC_r$ on which the circuit $\Gamma_{r, r-1}$ is applied in parallel, while performing error correction on the remaining blocks. This procedure is repeated until $\Gamma_{r, r-1}$ has been applied on all $h$ blocks, after which we obtain $2h$ blocks of $\cC_{r-1}$. Since $\Gamma_{r-1, r-2}$ has a smaller qubit overhead than $\Gamma_{r, r-1}$, we may choose a larger fraction of blocks on which to apply $\Gamma_{r-1, r-2}$ in parallel, and repeat it until $\Gamma_{r-1, r-2}$ has been applied on all the blocks.

\smallskip Iterating the above procedure, we decrease the level of QLDPC code through sequence $r, r-1, r-2, \dots$. The fraction of blocks on which $\Gamma_{r - y, r-y-1}$ is applied in parallel increases with $y$; hence the processing time at level $r-y$, i.e., time it takes to apply $\Gamma_{r - y, r-y-1}$ on all $2^{y}h$ blocks, decreases with $y = 0,1, 2, \dots$. This property is crucial for fault-tolerance since the logical error rate of $\cC_{r-y}$ increases with $y$ and the decreasing processing time prevents errors from accumulating excessively.  We stop the iteration $r, r-1, r-2, \dots, $ at some fixed level $r'$. Since $r'$ is fixed,  we can map from level $r'$ to $1$ by applying $\Gamma_{r'}$ on all the blocks in parallel, and only blowing up the overhead by a constant factor.

\smallskip It is worth emphasizing that fault-tolerant state preparation and decoding interface with constant qubit overhead have several potential applications. We remark on these applications in fault-tolerant computation and communication in Section~\ref{sec:apps}.

\paragraph{Structure of manuscript}
The manuscript is structured as follows. After some preliminary notation, which is presented in Section~\ref{sec:prelim}, we discuss properties of QLDPC codes in Section~\ref{sec:QLDPC}. In Section~\ref{sec:mainsection}, we present and discuss our main results on fault-tolerant interfaces. The error analysis of the interfaces will be presented in a separate section,  Section~\ref{sec:err-analysis}. In Section \ref{sec:FT-state-prep}, we derive constant qubit overhead fault-tolerant state preparation and discuss its applications.

\paragraph{Related work}
We note that in independent concurrent work,  
constructions in the fault-tolerant quantum input/output paradigm have been developed using constant rate concatenated codes \cite{belzig2026constant}.

\section{Preliminaries}\label{sec:prelim}
\paragraph{Notation.} We denote the single-qubit Hilbert space by $\bC^2$. The set of linear operators on $\bC^2$ is denoted by $\bL(\bC^2)$. Quantum channels are denoted by calligraphic letters such as $\cN, \cT, \cV, \cW, \cU$, etc. For any positive integer $n$, we write $[n] := \{1, \dots, n\}$. We shall denote the identity operator as $\ident$ and identity channel as $\cI$.

In the following, we recall necessary background on quantum error correcting codes and noise models in the context of fault-tolerant quantum computing. In the last subsection, we recall a recent statement on fault-tolerant state preparation, with a growing qubit overhead from~\cite{christandl2024fault}.

\subsection{Basic definitions}

\medskip Let $[n] := \{1, \dots, n\}$ be a set of labels for $n$-qubits. We define weight of operators as follows.
\begin{definition}[Weight of an operator] \label{def:weight-z}
An operator $E \in \bL((\bC^2)^{\otimes n } \otimes R)$, with $R$ being a reference system, is said to have weight $z \leq n$ with respect to the $n$ qubit system if it can be written as linear combination of the operators of following type for a subset $A \subseteq [n]$, $|A| \leq z$
\begin{equation}
    (\otimes_{i \in [n]\setminus A} \ident_i) \otimes E'[A] \otimes Z[R],
\end{equation}
where $E'[A]$ and $Z[R]$ are arbitrary operators acting on $A$ and $R$, respectively. The subset $A \subseteq [n]$ can be different for each term in the linear combination. 
\end{definition}
We note that for two copies of an $n$-qubit system, Def.~\ref{def:weight-z} allows to define weight with respect to each copy by considering the other system as part of the reference.
\begin{definition}[Weight of a quantum channel] \label{def:weight}
A quantum channel $\cN$ acting on an $n$-qubit system and a reference system $R$ is said to have weight $z \leq n$ with respect to the $n$-qubit system if there exists a Kraus representation $ \cN = \sum_i E_i (\cdot) E_i^\dagger$, where all the operators $E_i \in \bL((\bC^2)^{\otimes n } \otimes R)$ have weight $z$ with respect to the $n$-qubit system.  
\end{definition}
 
\smallskip We will use the diamond norm as a measure of distance between two superoperators~\cite{kitaev, watrous2009semidefinite, kretschmann2008information}.
\begin{definition}[Diamond norm]
 For any superoperator $\cT: \bL(M) \to \bL(N)$, the diamond norm $\dnorm{T}$ is defined as,
\begin{equation}
    \dnorm{\cT} := \sup_{G} \| \cI_G \otimes \cT  \|_1,
\end{equation}
where $G$ is a reference system and the one-norm is given by,
\begin{equation}
 \| \cT  \|_1   := \sup_{\substack{\| \rho\|_1 \leq 1 \\ \rho \in \bL(M)}}  \| \cT(\rho) \|_1. 
\end{equation}    
\end{definition}

\subsection{Quantum Codes}

\paragraph{Quantum codes:}
 A quantum code $\cC$ of type $(n, m)$, that is, encoding $m$ logical qubits in $n$ physical qubits is defined by an isometry $V : (\bC^2)^{\otimes m} \to (\bC^2)^{\otimes n}$, embedding $m$ qubit Hilbert space into $n$ qubit Hilbert space. The $\mathrm{Im}(V) \subseteq (\bC^2)^{\otimes n}$ is referred to as the corresponding code space. In this work, we will consider a particular family of quantum codes called \emph{Calderbank-Shor-Steane (CSS)} codes, which are a subset of a larger family of codes known as  \emph{stabilizer codes}~\cite{gottesman-thesis, calderbank1996good, steane1996multiple}.

\paragraph{Pauli group:} Pauli operators on a single qubit correspond to the set of matrices $P_1 := \{ \pm 1, \pm i \} \times \{\ident, X, Y, Z\}$, where
\[
\ident = \begin{pmatrix} 1 & 0 \\ 0 & 1 \end{pmatrix},\quad
X = \begin{pmatrix} 0 & 1 \\ 1 & 0 \end{pmatrix},\quad
Y = \begin{pmatrix} 0 & -i \\ i & 0 \end{pmatrix},\quad
Z = \begin{pmatrix} 1 & 0 \\ 0 & -1 \end{pmatrix}.
\]
The set $P_1$ is a group under matrix multiplication.  Pauli operators on $n$-qubits are obtained by taking $n$-fold tensor product of single qubit Pauli operators. The Pauli group on $n$-qubits corresponds to the set $P_n := \{ g_1 \otimes g_2 \otimes \cdots \otimes g_n \mid g_i \in P_1, \forall i \in [n] \}$. Using the identity $Y = i X Z$, any $n$-qubit Pauli operator $g$ can be written as follows,
\begin{equation} \label{eq:bin-Pauli}
    g = \phi \: X^{u_1} Z^{v_1} \otimes  X^{u_2} Z^{v_2} \otimes \cdots \otimes X^{u_n} Z^{v_n}, \text{ where } \phi \in \{\pm 1, \pm i\}.
\end{equation}
where $u_1, v_1, u_2, v_2, \dots, u_n, v_n \in \{0, 1\}$. Therefore, ignoring multiplicative phase factor in $\{\pm 1, \pm i\}$, we can uniquely identify Pauli operators by their $X$ and $Z$ binary vectors $\bu = (u_1, u_2, \dots, u_n)$  and $\bv = (v_1, v_2, \dots, v_n)$. If $\bv$ is a zero vector, we refer to $g$ as a Pauli operator of $X$ type. Similarly, if $\bu$ is a zero vector, we refer to $g$ as Pauli operator of $Z$ type.

 \paragraph{Stabilizer codes.}   A stabilizer code of type $(n, m)$ is defined by a subgroup $S$ of $n$ qubit Pauli group $P_n$, such that elements in $S$ pairwise commute with each other and $- \ident_n \not \in S$, where $\ident_n = \ident \otimes \cdots \otimes \ident$, with $\ident$ being the identity operator on $\bC^2$. The subgroup $S$ is described by an independent generating set $\cG := \{ g_1, \dots, g_{n-m} \} \in S$, containing $n-m$ pairwise commuting Pauli operators.

 The non-trivial logical operators of the stabilizer code correspond to $N(S) \setminus S$, where $N(S)$ is the normalizer of the stabilizer group $S$. The minimum distance of the stabilizer code is defined as
\begin{equation}
    d_{\min} = \min_{s \in N(S) \setminus S} \wt(s),
\end{equation}
where the weight $\wt(s)$ of an $n$-qubit Pauli operator $s$ according to Def.~\ref{def:weight-z}.

\paragraph{CSS codes.} A stabilizer code $\cC$ is referred to as CSS code if there exists a generating set $\cG$ such that it can be partitioned into sets $\cG_X$ and $\cG_Z$, containing only $X$ and $Z$ type Pauli operators, respectively.

\smallskip We associate with $\cG_X$ a binary matrix $H_X$ of dimension $|\cG_X| \times n$, such that rows of $H_X$ correspond to binary representations of the Pauli operator in $\cG_X$. Similarly, we associate with $\cG_Z$ a binary matrix $H_Z$ of dimension $|\cG_Z| \times n$. The pairwise commutativity of elements in $\cG = \cG_X \cup \cG_Z$ implies that $H_X H_Z^\top = [0]_{|\cG_X| \times |\cG_Z|}$.

\smallskip Note that the CSS code $\cC$ corresponding to $H_X$ and $H_Z$ encode $m := n - |\cG_X| - |\cG_Z|$ 
logical qubits into $n$ physical qubits. One can define classical codes $\cC_X$ and $\cC_Z$ by taking binary matrices $H_X$ and $H_Z$ as their parity check matrices, respectively. Let $\cC_X^\perp$ and $\cC_Z^\perp$ be dual codes of  $\cC_X$ and $\cC_Z$, respectively.\footnote{ The dual code $\cC^\perp_X$ of a binary code $\cC_X \subseteq \{0, 1\}^n$ is defined as $\cC^\perp_X := \{ \bu \in \{0, 1\}^n \mid \bu \cdot \bv = 0, \forall \bv \in \cC_X\}$.} The condition $H_X H_Z^\top = 0$ implies that $\cC_X^\perp \subseteq \cC_Z$ and $\cC_Z^\perp \subseteq \cC_X$. Therefore, one can define quotient groups $C_Z / C_X^\perp$ and $\cC_X / \cC_Z^\perp$. Note that dimension of the space $\dim(C_Z / C_X^\perp)$, i.e., the number of independent element in it,  is equal to $m$.

\smallskip Consider a $m \times n$ matrix $G_q$ defined by taking a set of $m$ independent vectors in $C_Z / C_X^\perp$. Then, the logical basis state $\ket{\bu_L}$ of the code $\cC$ corresponding to $\bu \in \{0, 1\}^m$ is given by, 
\begin{equation} \label{eq:basis-state-log}
    \ket{\bu_L} := \frac{1}{\sqrt{|\cC_X^\perp|}}\sum_{\by \in \cC_X^\perp} \ket{ \overline{\bu}\oplus\by },  \text{ for  } \overline{\bu} = G_q^\top \bu \in \cC_Z / \cC_X^\perp.
\end{equation}
The encoding isometry $V$ of $\cC$ is given by,
\begin{equation}
    V \ket{\bu} = \ket{\bu_L}, \: \forall \bu \in \{0, 1\}^m.
\end{equation}

\paragraph{Minimum Distance:} Let $d^X_{\min} := \min_{x \in (\cC_X \setminus \cC_Z^\perp) \setminus \{(0, 0, \dots 0)\} } \wt(x)$, where $\wt(x)$ denotes the Hamming weight of the bit string $x$, and  let $d^Z_{\min} := \min_{x \in (\cC_Z \setminus \cC_X^\perp) \setminus \{(0, 0, \dots 0)\} } \wt(x)$. Then the minimum distance of the CSS code $\cC$, corresponding to classical codes $\cC_X$ and $\cC_Z$ is defined as,
\begin{equation}
    d_{\min} = \min(d^X_{\min}, d^Z_{\min})
\end{equation}

\paragraph{Error correction procedure for CSS codes.} Given a Pauli error $ E = X^{\be_X} Z^{\be_Z}$, where $\be_X, \be_Z \in \{0, 1\}^n$, the corresponding error syndromes are given by $(H_X \be_Z, H_Z \be_X)$. The error syndromes can be obtained by performing Pauli measurements corresponding to the stabilizer generators of the CSS code. Then, using classical decoders for the codes $\cC_X$ and  $\cC_Z$, the $Z$ and $X$ errors are estimated to be $\hat{\be}_Z$ and $\hat{\be}_X$, respectively. The errors $\be_X, \be_Z$  are \emph{correctable} if the residual errors $\hat{\be}_X + \be_X$ and $\hat{\be}_Z + \be_Z$ are equivalent to identity up to a stabilizer, that is,  $\hat{\be}_X + \be_X \in \cC_X^\perp$ and $\hat{\be}_Z + \be_Z \in \cC_Z^\perp$.

\smallskip The notion of correctability extends by linearity beyond Pauli operators. Any linear operator $E \in \bL\left((\bC^2)^{\otimes n}\right)$ is said to be correctable if it can be written as linear combination of correctable errors, i.e., $E = \sum_{\be_X, \be_Z  \text{ correctable} } \alpha_{\be_X, \be_Z} X^{\be_X} Z^{\be_Z}$.   In the same spirit, a superoperator $\cV: \bL\left((\bC^2)^{\otimes n}\right) \to  \bL\left((\bC^2)^{\otimes n}\right)$ is correctable if it admits a decomposition as $\cV(\cdot) = \sum_i E_i (\cdot) E'_i$, where $E_i, E'_i$ are correctable errors for all $i$.

\smallskip Let $\cE := V(\cdot)V^\dagger$ be the encoding channel of $\cC$. Let $\cD$ be the ideal error correction procedure of $\cC$, including syndrome extraction, classical decoding and applying the error estimates. Then for any correctable superoperator $\cV$, we have 
\begin{equation}
    \cD \circ \cV \circ \cE = \kappa(\cV) \cE, 
\end{equation}
where $\kappa(\cV)$ is a constant depending only on $\cV$. In particular, $\kappa(\cV) = 1$ when $\cV$ is a quantum channel.

\paragraph{Subsystem codes.} A subsystem code divides the encoded degrees of freedom into logical qubits, which store the protected quantum information, and gauge qubits, whose state is irrelevant~\cite{bacon2006operator, kribs2006operator}. More precisely, the code space admits a decomposition $\cH_{\mathrm{L}}\otimes \cH_{\mathrm{G}}$, where $\cH_{\mathrm{L}}$ and $\cH_{\mathrm{G}}$ correspond to the logical and gauge subsystems, respectively.

 \smallskip A subsystem stabilizer code can be obtained from a stabilizer code encoding $m$ logical qubits by declaring $l < m$ of these qubits to be the gauge qubits. We call these code \emph{stabilizer subsystem codes}~\cite{poulin2005stabilizer}.

Note that the logical operators that act non-trivially on the logical subsystem, possibly also acting on the gauge subsystems, form a subset of $N(S) \setminus S$, where $S$ is the stabilizer group of the original code and $N(S)$ is its normalizer. Therefore, the minimum distance of the subsystem code is at least the minimum distance of the stabilizer code from which it is obtained.
 
\smallskip On the contrary a subsystem code can be turned into a stabilizer code by gauge fixing. In particular, we consider a set of $m-l$ commuting and independent logical operators corresponding to the gauge qubits and include it in the stabilizer generating set of the subsystem code~\cite{paetznick2013universal}. 

\paragraph{Minimum distance and error correction:}
The minimum distance of a quantum code is connected with its ability to correct errors. In order to state this connection, we need a notion of reduced stabilizer weight as given in Def.~\ref{def:stab-red-weight}, Def.~\ref{def:reduced-op} and Def.~\ref{def:reduced-supop}.

\begin{definition}[stabilizer-reduced error weight of a Pauli operator~\cite{gu2024single}] \label{def:stab-red-weight}
Consider a CSS code of type $(n, m)$, and let $X^{\be_x}, \be_x  \in \{0, 1\}^{n}$ be a Pauli-$X$ error. Let $|\be_x|$ denote the Hamming weight of $\be_x$. The stabilizer-reduced weight $|\be_x|_\mathrm{red}$ of $\be_x$ is defined as, 
\begin{equation*} 
   |\be_x|_\mathrm{red} := \min \{ |\be'_x| : X^{\be'_x} \text{ is  equivalent to $X^{\be_x}$ up to stabilizer multiplication} \}. 
\end{equation*}
 Similarly, we can define stabilizer-reduced weight for Pauli-$Z$ errors.  The stabilizer-reduced weight of the total error $\be = (\be_x, \be_z)$ is defined as $|\be|_\mathrm{red} := \max \{|\be_x|_\mathrm{red}, |\be_z|_\mathrm{red}\}$. 
\end{definition}
 In Def.~\ref{def:reduced-op} and Def.~\ref{def:reduced-supop}, we extend the definition of the stabilizer-reduced error weight to linear error operators and channels, respectively.
To this end, we consider error operators that act jointly on the $n$ physical qubits of the code and on an $\overline{n}$-qubit auxiliary system. This is done keeping in mind that we will later consider error correction of multiple code blocks in parallel; therefore, we would need a notion of stabilizer-reduced error weight with respect to each code block.

\begin{definition}[stabilizer-reduced error weight of a linear operator~\cite{christandl2025fault}] \label{def:reduced-op} 
    Consider a CSS code of type $(n, m)$. Let $E \in \bL((\mathbb{C}^2)^{\otimes \overline{n}} \otimes (\mathbb{C}^2)^{\otimes n})$ be an operator, jointly acting on a code block containing $n$ physical qubits and on an auxiliary system containing $\overline{n}$-qubits.  We expand $E$ as,
\begin{equation} \label{eq:expansion-k}
    E = \sum_{\be_x, \be_z \in \{0, 1\}^n}    E_{\be_x, \be_z} \otimes X^{\be_x} Z^{\be_z},
\end{equation}
where $E_{\be_x, \be_z} := \frac{1}{2^n} \tr_{[n]}(E (\ident_{\overline{n}} \otimes Z^{\be_z} X^{\be_x}) ) \in  \bL((\mathbb{C}^2)^{\otimes \overline{n}}) $, where $\tr_{[n]}$ refers to tracing out the $n$-qubits in the code block. The stabilizer-reduced weight of $E$ with respect to the code block is defined as 
\begin{equation*} 
    |E|_\mathrm{red} := \max \{ |\be|_\mathrm{red}, \be = (\be_x, \be_z) \in \{0, 1\}^{2n} : E_{\be_x, \be_z} \neq 0\}.
\end{equation*}
\end{definition}
\medskip
\begin{definition} [stabilizer-reduced error weight of a quantum channel~\cite{christandl2025fault}] \label{def:reduced-supop}
For a quantum channel $\cV: \bL((\mathbb{C}^2)^{\otimes \overline{n}} \otimes (\mathbb{C}^2)^{\otimes n}) \to \bL((\mathbb{C}^2)^{\otimes \overline{n}} \otimes (\mathbb{C}^2)^{\otimes n})$, we say $|\cV|_\mathrm{red} \leq t$ if there exists a decomposition,
\begin{equation}
    \cV = \sum_{i} E_i (\cdot) E_i^{ \dagger},
\end{equation}
such that $\max \{ |E_i|_\mathrm{red}\}_i \leq t$.
\end{definition}
We note that the Pauli errors of reduced weight strictly less than $d_{\min}$ are detectable. Moreover, ideal error correction is guaranteed to correct all Pauli errors of stabilizer-reduced weight strictly less than $\frac{d_{\min}}{2}$. By linearity, the same guarantee extends from Pauli errors to linear operators and channels whose stabilizer-reduced weight is strictly less than $\frac{d_{\min}}{2}$.

\begin{remark} \label{rem:com-red-weight}
Fix a stabilizer code and consider stabilizer-reduced weight with respect to the stabilizer code. The stabilizer-reduced weight of an operator is at most its ordinary weight, as defined in Def.~\ref{def:weight-z}. Moreover, for any two operators $E_1$ and $E_2$ acting partially on the same code block,  we have  $|E_1E_2|_{\mathrm{red}} \leq |E_1|_{\mathrm{red}}+|E_2|_{\mathrm{red}}$. Similarly, for any two quantum channels $\cV_1$ and $\cV_2$ acting partially on the same code block, we have
$|\cV_1\circ\cV_2|_{\mathrm{red}}  \leq |\cV_1|_{\mathrm{red}}+|\cV_2|_{\mathrm{red}},$ whenever their composition is well defined.
\end{remark}

\subsection{Noise model}
In this section, we describe the models of quantum circuits and noise used in this paper. We assume a stochastic circuit noise model, where each gate is independently affected by noise. For the purpose of error correction, however, we adopt the more general framework of local stochastic noise channels, which may introduce correlated errors across multiple qubits. This choice reflects the fact that stochastic circuit noise can induce correlated errors on qubits depending on the structure of the circuit. 
\subsubsection*{Quantum circuit}
We work in the circuit model of quantum computation, with a finite gate set $\bA$ consisting of an idle gate, initialization and measurement in the computational (Pauli-$Z$) basis, a universal set of unitary gates (e.g., $\cnot$, Hadamard $H$, and $T$), and their corresponding classically controlled versions. For any gate $g \in \bA$, we denote by $\cT_g$ the quantum channel realized by $g$. 

\smallskip We shall consider the following definition of quantum circuit~\cite{kitaev, christandl2024fault}.

\begin{definition}[Quantum circuit] \label{def:quantum-circuit}
A quantum circuit $\Phi$ of depth $d$ is a collection of the following objects
\begin{itemize}
 \item[$(1)$]   A sequence of finite sets $\Delta_0, \dots, \Delta_d$, called \emph{layers}. The zeroth layer $\Delta_0$ is called the \emph{input} of the circuit, and the final layer $\Delta_d$ is the \emph{output}. Each layer $\Delta_i$ (for $i = 0, \dots, d$) decomposes as $\Delta_i = \Delta_i^c \cup \Delta_i^q$, where $\Delta_i^c$ contains \emph{classical wires} and $\Delta_i^q$ contains \emph{quantum wires}. If the circuit has only classical input, then $\Delta_0^q = \emptyset$; if the circuit has no input, then $\Delta_0 = \emptyset$.
We say $\Phi$ operates on $x_q$ qubits and $x_c$ bits, where
\begin{equation}
    x_q := \max \{|\Delta^q_i|\}_{i = 0}^d \text{ and } x_c := \max \{|\Delta^c_i|\}_{i = 0}^d.
\end{equation}
  
\item[$(2)$] Partitions of each set $\Delta_{i-1}$, $i = 1, \dots, d$ into ordered subsets (registers) $A_{i1}, \dots, A_{is_i}$, and a corresponding partition of a superset $\Delta'_i \supseteq \Delta_i$ into registers $A'_{i1}, \dots, A'_{is_i}$.

\item[$(3)$] A collection of gates $g_{ij}$, where $j \leq s_i$, drawn from the gate set $\mathbf{A}$, such that register $A_{ij}$ and $A'_{ij}$ are input and outputs of $g_{ij}$, respectively. The gate $g_{ij}$ or their numbers $(i,j)$ are called the locations of the quantum circuit. Note that for a fixed $i, 1 \leq i \leq d$, the set of locations $\{g_{ij}: 1 \leq j \leq s_j \}$ corresponds to the locations in a single layer of the quantum circuit.
\end{itemize}

\end{definition}

\smallskip A quantum circuit $\Phi$ according to Definition~\ref{def:quantum-circuit} realizes a quantum channel 
\begin{equation}
    \cT = \cT_d \circ \cdots \circ \cT_1,
\end{equation}
 where $\cT_i$ acts as follows on any  $\rho \in \bL({\mathbb{C}^2}^{\otimes |\Delta_{i-1}|})$,
\begin{equation}
    \cT_i (\rho) = \tr_{\Delta'_i \setminus \Delta_i} \Big[ \cT_{g_{i,1}}[A'_{i1}; A_{i1}] \otimes \dots \otimes \cT_{g_{i, s_i}}[A'_{is_i}; A_{is_i}]  \Big],
\end{equation}
where $\cT_{g_{i, j}}[A'_{ij}; A_{ij}]$ denotes the quantum channel $\cT_{g_{i,j}}$, applied to the register $A_{ij}$, whose output is stored in the register $A'_{ij}$.

\begin{notation}
    For a quantum circuit $\Phi$, we denote the set of its locations by $\loc(\Phi)$. The size of the circuit $\Phi$, that is, the number of locations in it, is denoted by $|\Phi| := |\loc(\Phi)|$. 
\end{notation}

\subsubsection*{Circuit noise}
We consider the following noise model for quantum circuits.
\begin{definition}[Stochastic circuit-level noise] \label{def:stoc-circ}
Consider a quantum gate $g$, realizing quantum channel $\cT_g$. The noisy version of the gate $g$ under the stochastic circuit-level noise with parameter $\delta > 0$ realizes the following channel,
\begin{equation} \label{eq:noisy-real-gate}
    \tilde{\cT}_g = (1- \delta) \cT_g + \delta \: \cZ,
\end{equation} 
where $\cZ$ is an arbitrary quantum channel, with the same input and output systems as $\cT_g$. For a quantum circuit $\Phi$, the corresponding noisy circuit, under the stochastic circuit-level noise with parameter $\delta$ refers to the quantum circuit obtained by replacing each location in $\Phi$ by the corresponding noisy location. Note that $\cZ$ can be different in each location.  
\end{definition}
The noisy quantum circuit corresponding to $\Phi$ under the stochastic circuit-level noise with parameter $\delta$ realizes the following quantum channel,
\begin{equation} \label{eq:noisy-real-circuit-1}
    \tilde{\cT} = \tilde{\cT}_d \circ \cdots \circ \tilde{\cT}_1,
\end{equation}
 where $\cT_i$ acts as follows on any  $\rho \in \bL({\mathbb{C}^2}^{\otimes |\Delta_{i-1}|})$,
\begin{equation} \label{eq:noisy-real-circuit-2}
    \tilde{\cT}_i  = \tr_{\Delta'_i \setminus \Delta_i} \Big[\tilde{\cT}_{g_{i,1}}[A'_{i1}; A_{i1}] \otimes \dots \otimes \tilde{\cT}_{g_{i, s_i}}[A'_{is_i}; A_{is_i}] \Big],
\end{equation}
where $\tilde{\cT}_{g_{i, j}}$ is according to Eq.~\eqref{eq:noisy-real-gate}.
\begin{remark} \label{rem:noisy-gate}
We note when $g$ is either unitary or state-preparation gate, we can push the noise in their noisy channel $\tilde{\cT}_g = (1- \delta) \cT_g + \delta \cZ$ from Eq.~\eqref{eq:noisy-real-gate} to the left as below
\begin{equation}
    \tilde{\cT}_g = \bigl((1- \delta)\cI + \delta  \cZ'\bigr) \circ \cT_g,
\end{equation}
where \(\cI\) denotes the identity channel on the output of \(\cT_g\). For a unitary gate \(g\), we set
\(\cZ' := \cZ \circ \cT_g^\dagger\), and for a state-preparation gate we set
\(\cZ' := \cZ \circ \tr\), where \(\tr\) denotes the trace over the output system of \(\cT_g\). 

\smallskip For the gate $g$ corresponding to the computational basis measurement, we can push the noise to right as below
\begin{equation} \label{eq:noisy-measurement}
\tilde{\cT}_g = \cT_g \circ \bigl((1- \delta)\cI + \delta  \cZ'\bigr).
\end{equation}
We now show that Eq.~\eqref{eq:noisy-measurement} holds. For computational basis measurement, we have
\begin{equation}
     \cT_g(\cdot) = \sum_{a \in \{0,1\}} \tr(\cdot \proj{a}) \proj{a}. 
\end{equation}
For the noisy measurement $\tilde{\cT}_g = (1- \delta) \cT_g + \delta \cZ$, the quantum channel $\cZ$ is an arbitrary two-outcome quantum measurement. It can be described by positive operators $P_a \geq 0$, $a \in \{0,1\}$, with $\sum_{a} P_a = \ident$, as
\begin{equation} 
    \cZ(\cdot) = \sum_{a \in \{0, 1\}} \tr(\cdot P_a) \proj{a}.
\end{equation}
We can write the channel $\cZ$ as first applying a quantum channel $\cZ'$ and then applying the perfect measurement $\cT_g$, that is,
\begin{equation} \label{eq:measur-error}
    \cZ = \cT_g \circ \cZ',
\end{equation}
where the quantum channel \(\cZ'\) is given by
\begin{equation}
    \cZ'(\cdot) := \tr_A\!\bigl( U_{A \to AA'} (\cdot)\, U^\dagger_{A \to AA'} \bigr),  
\end{equation}
with \(A, A'\) being qubit systems, and \(U_{A \to AA'}\) an isometry defined as
\begin{equation*}
    U_{A \to AA'} = \sum_{a \in \{0, 1\}} \sqrt{P_a} \otimes \ket{a}.
\end{equation*}
Substituting the expression of $\cZ$ from Eq.~\eqref{eq:measur-error} in $\tilde{\cT}_g = (1- \delta) \cT_g + \delta \cZ$, we have
\begin{equation*}
    \tilde{\cT}_g = \cT_g \circ \bigl((1- \delta)\cI + \delta  \cZ'\bigr).
\end{equation*}
\end{remark}

\subsubsection*{Faulty circuit}
We consider the following notion of fault pattern and faulty circuit as given in Def.~\ref{def:ft-circuit}. As given in Remark~\ref{rem:prob-dist}, noisy version of a circuit according to Def.~\ref{def:stoc-circ} can be described as a probability distribution on faulty circuits. 
\begin{definition}[fault pattern and faulty circuit] \label{def:ft-circuit}
We say a gate $g$ is faulty if it realizes an arbitrary channel instead of the intended channel $\cT_g$. Consider a quantum circuit $\Psi$ and let $\loc(\Psi)$ be its set of locations. A fault pattern refers to a subset $\Delta \subseteq \loc(\Psi)$, and the corresponding faulty circuit $(\Psi, \Delta)$ refers to the quantum circuit obtained from $\Psi$ by replacing gates corresponding to locations in $\Delta$ by faulty gates.
\end{definition}

\begin{remark} \label{rem:prob-dist}
Note that the noisy version of a circuit $\Psi$, under the stochastic circuit-level noise with parameter $\delta$ can be seen as a probability distribution over faulty circuits $(\Psi, \Delta)$, where the fault pattern $\Delta \in \loc(\Psi)$ happens with probability 
\begin{equation}
\pr(\Delta) = (1 - \delta)^{|\Psi| - |\Delta|} \delta^{|\Delta|}.    
\end{equation}
Let $\cT_\Delta$ denote the channel realized by the faulty circuit $(\Psi, \Delta)$. Then, for the channel $\tilde{\cT}_\Psi$ realized by the noisy circuit corresponding to $\Psi$ under circuit-level stochastic channel with parameter $\delta$, according to Eq.~\eqref{eq:noisy-real-circuit-1} and Eq.~\eqref{eq:noisy-real-circuit-2}, we have:
\begin{equation}
    \tilde{\cT}_\Psi =  \sum_{\Delta \subseteq \loc(\Psi)} \pr(\Delta) \: \cT_\Delta.
\end{equation}
\end{remark}

\subsection*{Local stochastic noise}
Even though we have considered circuit-level stochastic noise, which acts independently on each location, it may lead to a correlated noise on qubits, namely local stochastic noise (see also Remark~\ref{rem:noisy-g-ls} and Theorem~\ref{thm:state-prep} below). In this section, we provide definition and some useful properties of local stochastic noise.

\begin{definition}[Local stochastic channel] \label{def:loc-stoc-noise}
Let $B$ be an $n$-qubit quantum system. A quantum channel $\cW[B, R]$, acting jointly on $B$ and a reference system $R$, is said to be local stochastic with parameter $\delta$ with respect to $B$ if
\begin{equation} \label{eq:W-exp-1}
    \cW[B, R] = \sum_{A \subseteq B} \pr_{\cW}(A)  \: \cI[ B \setminus A] \otimes \cN[A,R], 
\end{equation}
where $ \{ \pr_{\cW}(A) \}_{A \subseteq B}$ is a probability distribution, i.e.,  $\pr_{\cW}(A) \geq 0$ and $\sum_{A \subseteq B} \pr_{\cW}(A) = 1$, and  $\cN[A,R]$ is an arbitrary  quantum channel, such that for any $T \subseteq B$,
the probability that $T$ is included in $A$ is given by
\begin{equation} \label{eq:local-stoc-prop}
    \pr_{\cW}(T \subseteq A) \leq \delta^{|T|}.
\end{equation} 
We shall drop the subscript from $\pr_{\: \cW}$, when no confusion is possible.\footnote{Here $\pr_{\cW}(T \subseteq A) = \sum_{A: T\subseteq A} \pr_{\cW}(A)$.}
\end{definition}
We also consider a definition of local stochastic channel with respect to multiple quantum systems as given in Def.~\ref{def:loc-stoc-mult}. This will be useful when considering multiple code blocks of an error correcting code.
\begin{definition} \label{def:loc-stoc-mult}
Consider quantum systems $B_1, B_2, \dots, B_h$, containing $n_1, n_2, \dots, n_h$ qubits, respectively. We say a quantum channel $\cW[B_1, B_2, \dots, B_h, R]$, where $R$ is a reference system, is a local stochastic channel with parameter $\delta_1, \delta_2, \dots, \delta_h$,  with respect to quantum systems $B_1, B_2, \dots, B_h$, respectively if 
\begin{multline} \label{eq:W-exp}
   \cW[B_1, B_2, \dots, B_h, R] = \\ \sum_{A_i \subseteq B_i, i \in [h]} \pr_{\: \cW}(A_1, A_2, \dots, A_h)  \: \cI[ B_1 \setminus A_1, B_2 \setminus A_2, \dots, B_h \setminus A_h] \otimes \cN[A_1, \dots, A_h, R],      
\end{multline}
where $\cN$ is an arbitrary channel and for $T_i \subseteq B_i, i \in [h]$, the probability that $T_i$ is included in $A_i$ is given by
\begin{equation} \label{eq:W-prob}
    \pr_{\: \cW}(T_i \subseteq A_i) \leq \delta_i^{|T_i|}, \: \forall i \in [h]. 
\end{equation}
\end{definition}

\smallskip It is straightforward to see that a quantum channel $\cW$, which is local stochastic with respect to $B$ is also local stochastic with respect to any $B' \subseteq B$. Conversely, a local stochastic channel with respect to quantum systems $B_1, \dots, B_h$ is not necessarily local stochastic with respect to the joint system $\cup_{i \in [h]} B_i$. In what follows, we will mainly care about local stochastic property with respect to a code block, as error correction is done with respect to a code block.

\smallskip In the following remark, we note that the noisy version of a transversal gate under circuit-level stochastic noise is equal to the ideal gate up to a specific local stochastic noise.
\begin{remark} \label{rem:noisy-g-ls}
    Consider a  unitary gate $g$ acting on $t \geq 1$ qubits. Let $S_i, i \in [t]$ denote quantum systems, each containing $n$ qubits, and let $S_{i, j}, j \in [n]$ denote a qubit in $i$th block. Consider the transversal implementation of $g$ on them, that is,
    \begin{equation}
        g_S := \otimes_{j \in [n]} g_j,
     \end{equation}
where $g_j$ denotes a copy of $g$, acting on $t$ qubits corresponding to $S_{1,j}, S_{2,j}, \dots, S_{t,j}$. The noisy realization of $g_S$ corresponds to the perfect realization of $g_S$ followed by a local stochastic channel as follows,   
\begin{equation}
    \tilde{\cT}_{g_S} = \cW \circ \cT_{g_S},
\end{equation}
where $\cW$ is a local stochastic channel with parameter $\delta$ with respect to each block $S_i$. This can be seen as follows:

from Remark~\ref{rem:noisy-gate}, we have
\begin{equation}
    \tilde{\cT}_{g_S} = \left(\otimes_{j \in [n]}  \left( (1- \delta) \cI_j +  \delta \cZ_j  \right) \right) \circ \cT_{g_S},
\end{equation}
where $\cI_j$ and $\cZ_j$ are identity and an arbitrary channel, respectively, on $t$ qubits corresponding to $S_{1,j}, S_{2,j}, \dots, S_{t,j}$.  We define
\begin{align}
    \cW &:=  \otimes_{j \in [n]}  \left( (1- \delta) \cI_j +  \delta \cZ_j  \right) \nonumber \\ 
    & = \sum_{A \subseteq [n]} (1 - \delta)^{n-|A|} \delta^{|A|} (\otimes_{j \in [n]\setminus A} \cI_j) \otimes (\otimes_{j \in A}\cZ_j) \label{eq:ls-g1}
\end{align}
For any subset $T \subseteq [n]$, we have
\begin{align}
    \sum_{A: T \subseteq A} (1 - \delta)^{n-|A|} \delta^{|A|} & = \sum_{i = 0}^{n- |T|- i} \binom{n-|T|}{i} (1- \delta)^{n-|T|-i} \delta^{|T| + i} \nonumber \\
    & = \delta^{|T|} \label{eq:ls-g2}
\end{align}
From Eq.~\eqref{eq:ls-g1} and Eq.~\eqref{eq:ls-g2}, it follows that $\cW$ is a local stochastic channel with parameter $\delta$, with respect to each system $S_i, i \in [t]$.  
However, for $t > 1$, we note that since $\cZ_j$ acts on $t$ qubits, $\cW$ is not a local stochastic channel with parameter $\delta$ with respect to the joint system $\cup_{i = 1}^t S_i$. 

\smallskip Additionally, for $t = 1$ and $S = S_1$, and $g$ being a single qubit measurement, it can be seen using Remark~\ref{rem:noisy-gate} that 
\begin{equation}
     \tilde{\cT}_{g_S} = \cT_{g_S} \circ \cW, 
\end{equation}
where $\cW$ is a local stochastic channel with parameter $\delta$ with respect to $S$.
\end{remark}

In the following lemma, we consider a local stochastic channel $\cW$ with parameter $\delta$ with respect to a $n$ qubit system. We show that probability that the corresponding error acts non-trivially on a set of size greater than $\mu n$ decreases exponentially in $n$ as long as $\delta$ is below a threshold depending only on $\mu$.

\begin{lemma} \label{lem:ls-lw}
Consider constant $0< \mu < 1$ and $\delta < 2^{-\frac{h_2(\mu)}{\mu}}$, where $h_2(x) := -x \log_2 x - (1 - x) \log_2 (1-x)$ is the binary entropy.
Let $\cW[B_1, \dots, B_h, R]$, with $B_i, i \in [h]$ being an $n_i$-qubit system and $R$ being a reference, be a local stochastic channel with parameter $\delta$ with respect to system $B_i, i \in [h]$. Then, the following holds,
\begin{equation} \label{eq:ls-lw}
    \cW[B_1, \dots, B_h, R] = (1 - \tau) \:\cN[B_1, \dots, B_h, R] + \tau \: \cZ[B_1, \dots, B_h, R],
\end{equation}
where $\cN[B_1, \dots, B_h, R]$ is a channel of weight $\mu n_i$ with respect to $B_i$ for all $i \in [h]$ and $\tau \leq h  (2^{\frac{h_2(\mu)}{\mu}} \delta)^{\mu n}$, with $n := \min\{n_i \mid i \in [h]\}$ and $\cZ$ is an arbitrary channel.

In particular, the probability that $\cW$ acts non-trivially on a subset of $B_i$ with size greater than $\mu n_i$ goes to zero as $n \to \infty$.
\end{lemma}

\begin{proof}
For the local stochastic channel $\cW$, we have
\begin{align} 
   &\cW[B_1, B_2, \dots, B_h, R]  \nonumber\\ & = \sum_{A_i \subseteq B_i, i \in [h]} \pr_{\: \cW}(A_1, A_2, \dots, A_h)  \: \cI[ B_1 \setminus A_1, B_2 \setminus A_2, \dots, B_h \setminus A_h] \otimes \cN[A_1, \dots, A_h, R],  \nonumber \\
   & = (1 - \tau) \:\cN[B_1, \dots, B_h, R] + \tau \: \cZ[B_1, \dots, B_h, R], \label{eq:w-parts}
\end{align}
where 
\begin{align}
    \cN[B_1, B_2, \dots, B_h, R]= \frac{1}{1 - \tau} \sum_{i \in [h]} \sum_{A_i : |A_i| \leq \mu n} \pr_{\: \cW}(A_1, A_2, \dots, A_h) \: \cI[ B \setminus A] \otimes \cN[A_1, \dots, A_h, R],
\end{align}
with $\tau = \pr_W(|A_i| > \mu n \text{ for some } i \in [h])$.

\smallskip We will now give an upper bound on $\pr_\cW(|A_i| \ge \mu n)$ for any $i \in [h]$ and use this bound together with union bound to obtain an upper bound on $\tau$. For any positive integer $0 \leq t \leq n$, we have 
by a union bound,
\begin{align}
\pr_\cW(|A_i| \ge t)
& = \sum_{A_i \subseteq B_i, |A_i| \geq t } \pr_\cW(A_i) \nonumber\\
& \leq \sum_{T_i \subseteq B_i:\,|T|= t}\pr_\cW(T_i \subseteq A_i) \nonumber\\
&\leq \sum_{T_i \subseteq B_i:\,|T|= t} \delta^t \nonumber\\ 
&= \binom{n_i}{t}\,\delta^{t}, \label{eq:binom_delta}
\end{align}
where the first inequality uses a union bound and the fact that any $A_i \subseteq B_i, |A| \geq t$ must contain a subset $T_i \subseteq B_i, |T| = t$, and the second inequality uses local stochastic property of $\cW$.

\smallskip Now using the entropy bound $\binom{n_i}{t}\le 2^{n_i h_2(t/n_i)}$ and taking $t = \mu n_i$, we have
\begin{equation}\label{eq:entropy_step}
\pr_\cW(|A_i|\ge \mu n_i) \leq 2^{n_i h_2(\mu)}\,\delta^{\mu n_i} = \big(2^{h_2(\mu)/\mu}\delta\big)^{\mu n_i}.
\end{equation}
Using Eq.~\eqref{eq:entropy_step} and a union bound, we have
\begin{equation} \label{eq:tau-up}
    \tau \leq \pr(|A_i| \geq \mu n_i, \text{ for some } i \in [h])\leq \sum_{i \in [h]} \pr_\cW(|A_i| > \mu n) \leq h \big(2^{h_2(\mu)/\mu}\delta\big)^{\mu n}, 
\end{equation}
where $n = \min\{n_i \mid i \in [h]\}$.
From Eq.~\eqref{eq:w-parts} and Eq.~\eqref{eq:tau-up}, it follows that Eq.~\eqref{eq:ls-lw} holds.

For the last claim, if $\delta<2^{-h_2(\mu)/\mu}$ then $2^{h_2(\mu)/\mu}\delta<1$, and therefore
$\big(2^{h_2(\mu)/\mu}\delta\big)^{\mu n} \to 0$ as $n\to\infty$.
\end{proof}

\subsection{Fault-tolerant quantum state preparation}
We  recall below the state preparation theorem from~\cite[Corollary 61]{christandl2024fault}.
\begin{theorem}[Fault-tolerant state preparation~\cite{christandl2024fault}] \label{thm:state-prep}
There exists a fixed threshold value $\delta_{th} > 0$ such that the following holds:

Let $\Phi$ be a quantum circuit with no input and quantum output, and working on $x'$ qubits and depth $d$. Then, for any positive integer $k$, there exists another quantum circuit $\overline{\Phi}$, with the same input and output systems as $\Phi$, working on $O(x' \: \mathrm{poly}(k))$ qubits and of depth $O(d \:\mathrm{poly}(k))$, such that for the channel $\tilde{\cT}_{\overline{\Phi}}$ corresponding to the noisy realization of   $\overline{\Phi}$ under stochastic circuit-level noise with parameter $\delta < \delta_{th}$, we have
\begin{equation} \label{eq:prep-ls}
  \tilde{\cT}_{\overline{\Phi}} = (1 - \tau) \:  \cW \circ \cT_{\Phi} + \tau \:\cZ,
\end{equation}
where $\cW$ is a local stochastic channel with parameter $ 2 c\delta$, $\tau \leq O(|\Phi|(c\delta)^k)$, and $\cZ$ is an arbitrary channel.
\end{theorem}
\section{Quantum LDPC codes}\label{sec:QLDPC}
For our constant qubit overhead protocols, we consider quantum low-density parity-check (quantum LDPC or QLDPC) codes of CSS type~\cite{mackay2004sparse, breuckmann2021quantum, }. We require QLDPC codes to satisfy some additional properties, including constant encoding rate, linear minimum distance, and having efficient error correction protocols. We show that such QLDPC codes can be obtained using quantum Tanner codes~\cite{leverrier2022quantum}, and decoder analyses of~\cite{leverrier2023decoding, gu2024single}.

\subsection{Rates of QLDPC codes} \label{sec:rates-qldpc}
Below, we first provide definition of QLDPC codes.
\begin{definition}[QLDPC codes] \label{def:qldpc-rate-constant}
Let $\cC_r,  r = 1, 2, \dots $  be a family of stabilizer codes, where $\cC_r$ is a code. It is said to be QLDPC if there exists a generating set $\cG$, possibly redundant in the sense that it's not a linearly independent set, such that weight of each $g \in \cG$ is bounded by a constant  independent of $r$ and the total number of generators non-trivially acting on any given qubit is upper bounded by a constant independent of the parameter $r$.   
\end{definition}
For a CSS code family $\cC_r,  r = 1, 2, \dots $, QLDPC property can be stated in terms of the check matrices of the code. Let $H_X^r$ and $H_Z^r$ be the $X$ and $Z$ type parity check matrices of $\cC_r$, respectively. Then, the code family is said to be QLDPC if the Hamming weight of any row or column of the parity check matrices $H^r_X$ and $H_Z^r$ is constant, i.e., independent of $r$.

\smallskip Several constructions of QLDPC codes with constant rate and minimum distance growing sub-linearly have been proposed, for instance~\cite{tillich2013quantum, couvreur2013construction, evra2022decodable,hastings2021fiber}. The existence of QLDPC codes with linear minimum distance remained an open problem for many years, which was settled in Refs.~\cite{panteleev2022asymptotically, leverrier2022quantum}, where  the existence of such codes is explicitly demonstrated.

\medskip  In what follows, we fix a family of CSS QLDPC code  $\{ \cC_r \mid r = 1, 2, \dots \}$, where $\cC_r$ is a stabilizer code of type $(n_r, \overline{m}_r)$. We suppose that $\cC_1$ is the trivial code encoding one logical qubit into one physical qubit. For each $\cC_r$, we may turn it into a subsystem code by considering some of the logical qubits of $\cC_r$ as gauge qubits. We will not do gauge fixing, i.e., the parity check matrices of the subsystem code are the same as those of $\cC_r$. Moreover, we will use the same error correction procedure as $\cC_r$ for the subsystem code, i.e., the same stabilizer generators and classical decoding procedure as the original code. Thus, the subsystem code inherits the QLDPC stabilizer checks
and the error-correction procedure of $\cC_r$. The logical Pauli operators associated with the gauge qubits are not used as generators, hence, they are not required to have bounded weight.

\smallskip Let $n_r$ be the number of physical qubits of $\cC_r$ and $m_r \leq \overline{m}_r$ be the number of protected logical qubits corresponding to the subsystem code of $\cC_r$. Let $d^{(r)}_{\min}$ be its minimum distance. We suppose that the following properties hold
\begin{align}
m_r &\geq \alpha n_r \label{eq:n-r}  \\
m_r &= 2 m_{r - 1} \label{eq:m-r-r-1} \\
d^{(r)}_{\min} &\geq \beta n_r \label{eq:min-distance},
\end{align}
for some constant $\alpha, \beta \in (0, 1]$. Note that Eq.~\eqref{eq:n-r} means that the encoding rate of the subsystem code is linear. Moreover, Eq.~\eqref{eq:m-r-r-1} means that the subsystem code corresponding to $\cC_r$ encodes exactly the twice as many qubits compare to the subsystem code corresponding to $\cC_{r-1}$.

Using Eq.~\eqref{eq:m-r-r-1}-\eqref{eq:min-distance}, we have that
\begin{equation} \label{eq:n-r-1-bound}
    n_r \leq \frac{m_r}{\alpha} =  \frac{2 m_{r-1}}{\alpha} \leq \frac{2 n_{r-1}}{\alpha}.
\end{equation}
Moreover, for our QLDPC code $\cC_r, r= 1, 2, \dots$, we assume an \emph{efficient classical decoder} that can correct a constant fraction of errors, that is, it outputs a valid correction for any Pauli error of stabilizer-reduced weight (see also Def.~\ref{def:stab-red-weight}) bounded by $\frac{\beta}{2} n_r - 1$, with $\beta$ in Eq.~\eqref{eq:min-distance}.

 Let $\overline{V}_r: (\bC^2)^{\otimes \overline{m}_r} \to (\bC^2)^{\otimes n_r}$ be the encoding isometry of $\cC_r$. We will use the isometry $V_r:(\bC^2)^{\otimes m_r} \to (\bC^2)^{\otimes n_r}$ for the corresponding subsystem code as defined below
\begin{equation} \label{eq:iso-subsytem}
    V_r\ket{\psi_r} := \overline{V}_r(\ket{\psi_r} \otimes \ket{0}^{\otimes (\overline{m}_r-m_r)}), \: \forall \ket{\psi_r} \in (\bC^2)^{\otimes m_r}.
\end{equation}
Thus, the gauge qubits are fixed to all zero state in the definition of $V_r\ket{\psi_r}$. We will use $V_r$ for the preparation of code states of the subsystem code for example in Section~\ref{sec:partial-interfaces-tele}.   An encoding channel $\cE_{r,\eta}: \bL\left((\bC^2)^{\otimes m_r}\right) \to \bL\left((\bC^2)^{\otimes n_r}\right)$ for the subsystem code is as follows
\begin{equation}
    \cE_{r, \eta}(\cdot) = \overline{V}_r( \cdot \otimes \eta) \overline{V}^\dagger_r,
\end{equation}
where $\eta$ is an arbitrary state on $\overline{m}_r-m_r$ qubits. In the following, we will not explicitly show the dependence on $\eta$ and simply write $\cE_r$ as the encoding channel of the subsystem code corresponding to $\cC_r$.

\smallskip In Lemma~\ref{lem:qldpc-new} below, we construct a subsystem QLDPC code using quantum tanner codes~\cite{leverrier2022quantum} satisfying Eq.~\eqref{eq:m-r-r-1}-\eqref{eq:n-r-1-bound}. The proof of Lemma~\ref{lem:qldpc-new} relies on the following theorem. 
\begin{theorem} \label{thm:qldpc-tanner}
    There exists an integer $C > 0$ and $\beta > 0$ and a family of QLDPC codes $\{\cC_r\}_{r \geq 1}$ satisfying for all $r \geq 1$
    \begin{align*}
        n_{r} &\leq C n_{r-1} \\
        m_{r} &\geq \frac{1}{4} n_r \\
        m_{r} &\geq 6 m_{r-1} \\
        d_{\min}^{(r)} &\geq \beta n_{r}.
    \end{align*}
Moreover, there exists a parallel decoder running in time $O(\log n_r)$ that can decode a linear number of errors.
\end{theorem}
\begin{proof}
    We can for instance take the construction of left-right Cayley complex for the group $\mathrm{PSL}_2(q^r)$ for some odd prime $q$. This construction is described after~\cite[Theorem 1]{leverrier2022quantum} and comes from~\cite{dinur2022locally}. By choosing the parameter $\rho = 1/4$, we obtain: $n_{r} = (q+1)^2 (q^{3r} - q^{r}) / 2$ and $m_r \geq n_r/4$ and $d_{\min}^{(r)} \geq 1/(1000(q+1)^2) n_{r})$. It is simple to check that for this choice $(q^{3} - 1) n_{r-1} \leq n_{r} \leq q^3 n_{r-1}$. As a result, $m_{r} \geq n_r/4 \geq (q^3-1)n_{r-1}/4 \geq 6m_{r-1}$. The decoding algorithm follows from~\cite{leverrier2023decoding}.
\end{proof}

\begin{lemma} \label{lem:qldpc-new}
Consider a family of QLDPC codes $\cC_r, r= 1, 2, \dots, $ of type $(n_r, m_r)$ such that the following properties hold,
  \begin{align*}
        m_{r} &\geq \alpha n_r \\
        m_r &\geq C_2 m_{r-1}   \\
        d_{\min}^{(r)} &\geq \beta n_{r} \\
        n_{r} &\leq C_1 n_{r-1} 
    \end{align*}
for some constant $\alpha, \beta \in (0, 1]$ and $C_1 \geq C_2 > 1$. Then, there exists a subsystem QLDPC code family $\overline{\cC}_{s}, s = 1, 2, \dots, $ where $\overline{\cC}_{s}$ has $\overline{n}_{s}$ physical qubits and $\overline{m}_{s}$ protected logical qubits with the following properties,
  \begin{align*}
\om_{s} &=  2^s\\
\overline{m}_{s} &\geq \frac{\alpha}{C_1} \on_{s} \\ 
\overline{d}_{\min}^{(s)} &\geq \beta \on_{s} \\
\on_{s}   &\leq  \frac{2C_1}{\alpha} \on_{s - 1} 
\end{align*}  
Therefore, the code $\overline{\cC}_s, s = 1, 2, \dots$ satisfies Eq.~\eqref{eq:m-r-r-1}-\eqref{eq:n-r-1-bound}. 
\end{lemma}

\begin{proof}
We will obtain the code family $\overline{\cC}_{s}, s= 1, 2, \dots$, using the given code family $\cC_r, r = 1, 2, \dots $. In particular, we will be using codes from $\cC_r, r = 1, 2, \dots $, with a smaller encoding rate than they can support. It is worth noting that the decreasing encoding rate can only increase minimum distance.

\smallskip The code $\overline{\cC}_s$ for any $s > 0$ is obtained as follows:
consider $r$ such that $m_{r-1} < 2^s \leq m_r$. Since $m_r \geq C_2^{r-1} m_{1}$ is strictly increasing (as $C_2 > 1$), there exists an $r$ for any $s$ satisfying $m_{r-1} < 2^s \leq m_r$. Note that for any $s' \geq s$, the corresponding $r'$, such that $m_{r'-1} < 2^{s'} \leq m_{r'}$, satisfies $r' \geq r$. Let $\overline{\cC}_s$ be the code obtained by encoding $2^s$ logical qubits in $n_r$ physical qubits corresponding to $\cC_r$.  By construction, we have that $ \overline{d}_{\min}^{(s)} \geq d_{\min}^{r} \geq \beta \on_{s}$. The encoding rate of $\overline{\cC}_s$ is given by
\begin{equation} \label{eq:rate-s}
    \frac{\om_{s}}{\on_{s}} > \frac{m_{r-1}}{n_{r}} \geq \alpha \frac{n_{r-1}}{n_{r}} \geq \frac{\alpha}{C_1}.
\end{equation}
Moreover,
\begin{align}
    \on_s \leq \frac{C_1}{\alpha} \om_s = \frac{2C_1}{\alpha} \om_{s-1} \leq \frac{2C_1}{\alpha} \on_{s-1}.
\end{align}
\end{proof}

\subsection{Fault-tolerant error correction}
For our QLDPC code $\cC_r, r= 1, 2, \dots$, satisfying Eq.~\eqref{eq:n-r}-\eqref{eq:min-distance}, we need a fault-tolerant error correction circuit. In the fault-tolerant setting, syndrome extraction is noisy; therefore, faults in the syndrome extraction circuit can introduce additional errors on the physical qubits of the code block and, crucially, can also corrupt the  syndrome bits. The goal of fault-tolerant error correction, however, is not to eliminate all errors, but to ensure that, despite noisy syndrome extraction, the error on the code state remains controlled. In particular, since our QLDPC codes can correct a constant fraction of errors, it suffices that fault-tolerant error correction keeps the physical qubit error below this constant fraction.

\smallskip While many constructions of fault-tolerant error correction rely on multiple rounds of syndrome measurement, some code families admit \emph{single-shot} fault-tolerant error correction, requiring only one round of noisy syndrome extraction. For quantum Tanner codes, our main candidate family in this paper (see Section~\ref{sec:rates-qldpc}), single-shot fault-tolerant error correction was established in Ref.~\cite{gu2024single}. Provided that (i) the input error is within a  constant fraction and (ii) the syndrome-bit errors are themselves bounded by a constant fraction, the construction of Ref.~\cite{gu2024single} guarantees that the post-correction error on the code state remains within (slightly smaller) constant fraction.

\smallskip A key property of QLDPC codes is that they admit constant-depth syndrome-extraction circuits, since their parity-check matrices have constant-weight rows and columns. Therefore, in the syndrome extraction circuit, a single fault can affect only $O(1)$ data qubits and corrupt only $O(1)$ syndrome bits. Consequently, under circuit-level stochastic noise with a fixed constant parameter, the total error introduced during a round of noisy syndrome extraction has weight with high probability at most a constant fraction of the data qubits and a constant fraction of the syndrome bits (see also~\cite[Lemma~2]{christandl2025fault}). Combining this property of syndrome extraction with the single-shot fault-tolerant error correction of Ref.~\cite{gu2024single} motivates the following definition of fault-tolerant error correction (see also~\cite[Theorem~5]{christandl2025fault}), which we will use throughout this work.
\begin{definition}[Fault-tolerant error correction] \label{def:ft-err-cor}
A quantum circuit $\Phi_{\mathrm{EC}^r}$, with $n_r$ qubit input and output, implementing error correction on the code $\cC_r$ is said to be fault-tolerant if it is of constant depth $O(1)$ and requires $O(m_r)$ ancilla qubits, and the following holds for a fixed threshold value $\delta_{th} > 0$ and some constants $\beta_1, \beta_2, \kappa, c, c' > 0$:

\smallskip Let $\cN: \bL((\bC^2)^{\otimes h n_r} \otimes R) \to  \bL((\bC^2)^{\otimes h n_r} \otimes R)$, with $R$ being a reference system, be an arbitrary quantum channel with stabilizer-reduced weight $|\cN|_{\mathrm{red}} \leq \beta_1 n_r +  \beta_2 n_r$, with respect to each $n_r$ qubit system in its input.\footnote{stabilizer-reduced weight of a quantum channel is defined in Def.~\ref{def:reduced-supop}.}  Then, the noisy realization $\tilde{\cT}_{\Phi_{\mathrm{EC}^r}}$ of $\Phi_{\mathrm{EC}^r}$ under circuit-level stochastic noise with parameter $\delta < \delta_{th}$ applied in parallel to $h$ blocks of $\cC_r$ gives 
\begin{equation} \label{eq:T-del}
(\tilde{\cT}^{\otimes h}_{\Phi_{\mathrm{EC}^r}} \otimes \cI[R]) \circ \cN \circ (\cE^{\otimes h}_r \otimes \cI[R])  =  (1 - \tau) \: \cN' \circ (\cE^{\otimes h}_r \otimes \cI[R]) + \tau \:  \cZ \circ \cN \circ (\cE^{\otimes h}_r \otimes \cI[R]), 
\end{equation}
where $\cN'$ is a quantum channel with stabilizer-reduced weight $\beta_2 n_r +  \kappa n_r\delta$ with respect to each $n_r$ qubit system and $\tau \leq h \varepsilon(r, \delta)$, with $\varepsilon(r, \delta) =  (c \delta)^{ c' n_r}$. $\cZ$ is an arbitrary error channel. The parameter $\varepsilon(r, \delta)$ is referred to as the ``fault-tolerant parameter". 
\end{definition}
\paragraph{Fault-tolerant memory:} Consider fault-tolerant error correction according to Def.~\ref{def:ft-err-cor}. We note that by choosing $\delta < \frac{\beta_1}{\kappa}$, we can make sure that the fault-tolerant error correction suppresses the stabilizer-reduced weight from $\beta_1 n_r + \beta_2 n_r$ to $\beta'_1 n_r + \beta_2 n_r$, where $\beta'_1 := \kappa \delta < \beta_1$. The error suppression allows to create a fault-tolerant memory, where logical information can be preserved for arbitrarily long time by choosing $r$ sufficiently big, as a function of storage duration. 

\smallskip Using the fault-tolerant error correction circuit $\Phi_{\mathrm{EC}^r}$, we define $s > 0$ error correction steps as follows:
\begin{equation} \label{eq:s-cor-step}
     \Phi_{\mathrm{EC}^{r, s}} := \Phi_{\mathrm{EC}^{r}} \circ \stackrel{ s \text{ times}}{\cdots} \circ \: \Phi_{\mathrm{EC}^{r}}.
\end{equation}
The following result on $\Phi_{\mathrm{EC}^{r, s}}$ follows using a union bound on multiple instances of the fault-tolerant error correction from Def.~\ref{def:ft-err-cor}.

\begin{lemma}[Multiple rounds of error correction] \label{lem:mult-ec}
Assuming the fault-tolerant error correction according to Def.~\ref{def:ft-err-cor} and the threshold value $\delta_{th} > 0$ and constants $\beta_1, \beta_2, \kappa, c, c' > 0$ therein, the following holds:
   
\smallskip Let $\cN: \bL((\bC^2)^{\otimes h n_r} \otimes R) \to  \bL((\bC^2)^{\otimes h n_r} \otimes R)$ be an arbitrary quantum channel with stabilizer-reduced weight $|\cN|_{\mathrm{red}} \leq \beta_1n_r + \beta_2n_r$. Then, for the noisy realization $\tilde{\cT}_{\Phi_{\mathrm{EC}^{r, s}}}$ of $\Phi_{\mathrm{EC}^{r, s}}$ defined in Eq.~\eqref{eq:s-cor-step}, under circuit-level stochastic noise with parameter $\delta < \min(\delta_{th}, \frac{\beta_1}{\kappa})$, we have
\begin{equation} \label{eq:T-del-1}
(\tilde{\cT}^{\otimes h}_{\Phi_{\mathrm{EC}^{r,s}}} \otimes \cI[R]) \circ \cN \circ (\cE^{\otimes h}_r \otimes \cI[R])  =  (1 - \tau) \: \cN' \circ (\cE^{\otimes h}_r \otimes \cI[R]) + \tau \:  \cZ \circ \cN \circ (\cE^{\otimes h}_r \otimes \cI[R]), 
\end{equation}
where $\cN'$ is a quantum channel with stabilizer-reduced weight $\beta_2 n_r + \kappa n_r\delta$ with respect to each $n_r$ qubit system, and $\tau \leq sh \varepsilon(r, \delta)$, with $\varepsilon(r, \delta) =  (c \delta)^{ c' n_r}$. $\cZ$ is an arbitrary error channel.
\end{lemma}

\section{Fault-tolerant interfaces}
\label{sec:mainsection}
In this section, we construct fault-tolerant decoding interfaces with constant qubit overhead for QLDPC codes. Our construction assumes QLDPC code $\cC_r, r= 1, 2, \dots$, satisfying Eq.~\eqref{eq:m-r-r-1}-\eqref{eq:min-distance}. Moreover, efficient classical decoding and fault-tolerant error correction according to Def.~\ref{def:ft-err-cor} are assumed. As shown in Section~\ref{sec:QLDPC}, quantum Tanner codes combined with the decoder analysis of Refs.~\cite{leverrier2023decoding, gu2024single} 
satisfy these assumptions.

\smallskip This section is organized as follows: In Section~\ref{sec:main-result}, we state our main result on the fault-tolerant decoding interface $\Xi^{[h]}_{r}$, which relies on the fault-tolerance of partial decoding interfaces $\Gamma_{r, r'}$. The remaining sections are dedicated to proving this result: In Section~\ref{sec:ft-interface-non-constant}, we discuss the partial decoding interface and in Section~\ref{sec:proof-succ-int} we provide its error analysis. Section~\ref{sec:many-blocks} provides an extension to multiple code blocks. In Section~\ref{sec:const-ovr-interface} we present the main argument regarding the constant qubit overhead. The error analysis is more involved and will be provided in a separate Section~\ref{sec:err-analysis}. 
\subsection{Main result} \label{sec:main-result}
Our goal is to construct a decoding interface circuit $\Xi^{[h]}_{r}$ that maps $m_r h$ logical qubits encoded in $h$ blocks of $\cC_r$ to bare physical qubits for any $r > 0$, while having a constant qubit overhead, i.e. operating on $O(m_r h)$ qubits. Moreover, for fault-tolerance, we require that the noisy version of $\Xi^{[h]}_{r}$ preserves the logical information encoded in $h$ blocks of $\cC_r$ up to a weak local stochastic noise applied on the bare physical qubits.

\smallskip Our main result is the following Theorem~\ref{thm:ft-cons-int-main}, where we provide our fault-tolerant decoding interfaces for QLDPC codes.

\begin{theorem} \label{thm:ft-cons-int-main}
There exists a polynomial $p(\cdot)$, a constant $\mu > 0$, a threshold value $\delta_{th} > 0$, qubit overhead constants $\theta, \theta' > 0$, and error rate constants $\kappa_1, \kappa_2 > 0$, such that the following holds:

\smallskip For any $r, h \in \mathbb{N}$, any $\delta < \delta_{th}$, and any quantum channel 
$$\cN: \bL((\bC^2)^{\otimes n_r h} \otimes R) \to \bL((\bC^2)^{\otimes n_r h} \otimes R),$$ with stabilizer-reduced weight relative to $\cC_r$ bounded as $|\cN|_{\mathrm{red}} \leq \mu n_r$ with respect to the $n_r$ qubits in each block of the input, there exists a quantum circuit $\Xi^{[h]}_r$, with $n_r h$ qubit input and $m_r h$ qubit output and operating on fewer than $ \theta p(m_r) m_r+ \theta' m_r h$ qubits, such that its noisy realization under circuit level stochastic noise with parameter $\delta$ satisfies:
 \begin{equation} \label{eq:interface-xi-r}
    (\tilde{\cT}_{\Xi^{[h]}_r} \otimes \cI[R]) \circ \cN \circ (\cE_r^{\otimes h} \otimes \cI[R])  =  \cV', 
 \end{equation}
where $\cV':\bL((\bC^2)^{\otimes m_r h} \otimes R) \to \bL((\bC^2)^{\otimes m_r h} \otimes R)$ is a local stochastic channel with parameter $(\kappa_1 \delta)^{\kappa_2}$ with respect to $m_r h$ qubits in its input.

In particular, if $h \geq p(m_r)$, $\Xi^{[h]}_{r}$ operates on fewer than $(\theta + \theta') m_rh$ qubits; therefore, the qubit overhead is given by the constant $\theta + \theta'$.
\end{theorem}
To construct our decoding interface $\Xi^{[h]}_{r}$, we proceed as follows:  we first construct a partial decoding interface $\Gamma_{r, r'}$ in Lemma~\ref{lem:succ-int} for $r > r' > 1$ such that $r - r' = O(1)$ (i.e., the interface reduces the encoding by a constant number of levels), mapping a block of $\cC_r$ into $2^{r-r'}$ blocks of $\cC_{r'}$. The interface $\Gamma_{r, r'}$ does not have a constant qubit overhead. We also construct an interface $\Gamma_{r, 1}$ in Lemma~\ref{lem:succ-int-1} that works for any constant $r$, that is, $r = O(1)$. 

\smallskip We then combine $\Gamma_{\hat{r}, \hat{r}-1}$ for different values of $\hat{r}$ to construct an interface $\Xi^{[h]}_{r, r'}$ for arbitrary $r, r'$ (Lemma~\ref{lem:cons-over-interface}), which maps $h$ blocks of $\cC_r$ into $2^{r-r'}h$ blocks of $\cC_{r'}$.  Finally, we obtain the interface $\Xi^{[h]}_{r}$ for Theorem~\ref{thm:ft-cons-int-main} (here $r$ is arbitrary) by composing interfaces $\Xi^{[h]}_{r, r'}$ and $\Gamma_{r', 1}$ for some fixed $r'$.

\subsection{Partial decoding interfaces for QLDPC codes} \label{sec:ft-interface-non-constant}
In this section, we provide a partial decoding interface based on quantum teleportation, which will be used to construct the decoding interface $\Xi^{[h]}_r$ for Theorem~\ref{thm:ft-cons-int-main} in Section~\ref{sec:const-ovr-interface}. 

\smallskip We consider the following scenario: Let $r > r' \geq 1$, such that $r - r' = O(1)$. We want to map $m_r$ logical qubits encoded in one block of $\cC_r$ consisting of $n_r$ physical qubits to $m_r$ logical qubits encoded in $2^{r-r'}$ blocks of $\cC_{r'}$ each consisting of $n_{r'}$ physical qubits. In order to do so, in Lemma~\ref{lem:succ-int}, we present a fault-tolerant interface circuit $\Gamma_{r, r'}$ for $r > r' > 1$, with $n_r$ qubit input and $ 2^{r-r'} n_{r'}$ qubit output. For the special case of $r = O(1)$ and $r' = 1$, we present a fault-tolerant interface $\Gamma_{r, 1}$ in Lemma~\ref{lem:succ-int-1}.

\smallskip We emphasize that $r - r' = O(1)$ is important for the fault-tolerance of the partial interface $\Gamma_{r, r'}$. Moreover,  the overhead of $\Gamma_{r, r'}$ grows with $r$. 
\begin{lemma}\label{lem:succ-int}
There exist polynomials $p_1(\cdot), p_2(\cdot)$, and $p_3(\cdot)$, constants $\mu, \theta,  \lambda, d, c, c',  > 0$, and a fixed threshold value $\delta_{th} > 0$, such that the following holds:
 
\smallskip Consider codes $\cC_r$ and $\cC_{r'}$, where $r > r' > 1$ and $r - r' = O(1)$. Let $\cN: \bL((\bC^2)^{\otimes n_r} \otimes R) \to \bL((\bC^2)^{\otimes n_r}  \otimes R)$, with $R$ being a reference system, be a quantum channel with stabilizer-reduced weight relative to $\cC_r$ bounded as $|\cN|_{\mathrm{red}} \leq \mu n_r$ with respect to $n_r$ qubits in its input.

\smallskip Then, for any positive integer $k$, there exists a quantum circuit $\Gamma_{r, r'}$, with $n_r$ qubit input and $2^{r-r'} n_{r'}$ qubit output, working on fewer than $\theta m_r p_1(m_r)$ physical qubits and of depth smaller than $d p_2(m_r)$, such that for the quantum channel $\tilde{\cT}_{\Gamma_{r, r'}}$ corresponding to the noisy realization of $\Gamma_{r, r'}$ under circuit-level stochastic noise with parameter $\delta < \delta_{th}$, we have\footnote{Here $\cE^i_{r'}, i \in [2^{r-r'}]$ denote $2^{r-r'}$ copies of the encoding $\cE_{r'}$ of $\cC_{r'}$.}
\begin{multline} \label{eq:gamma-r-r'}
    (\tilde{\cT}_{\Gamma_{r, r'}} \otimes \cI[R]) \circ \cN \circ (\cE_r \otimes \cI[R])    = (1 - \tau_{r})  \cN' \circ \left((\otimes_{i \in [2^{r-r'}]}\cE^i_{r'}) \otimes  \cI[R]\right)  \\ + \tau_{r} \: (\cZ_{r \to r'} \otimes \cI[R]) \circ \cN \circ (\cE_r \otimes \cI[R]),
\end{multline}
where $\cN': \bL((\bC^2)^{\otimes n_{r'} \: 2^{r-r'}} \otimes R) \to \bL((\bC^2)^{\otimes n_{r'} \: 2^{r-r'}}  \otimes R)$ is a quantum channel with stabilizer-reduced weight $\mu n_{r'}$ with respect to each $n_{r'}$ qubit system in its input. The parameter $\tau_{r} \in [0, 1]$ is given by 
\begin{equation} \label{eq:bound-tau-rk}
    \tau_{r} \le \lambda p_3(m_{r'}) (c \delta)^{c' m_{r'}}.
\end{equation}
 The map $\cZ_{r  \to r'}$ is an arbitrary channel with $n_{r}$ qubit input and $2^{r-r'}n_{r'}$ qubit output.
\end{lemma}
Note that the qubit overhead of the interface $\Gamma_{r,r'}$ with respect to the number of logical qubits on which it operates scales as $p_1(m_r)$. Moreover, the noisy implementation of the interface $\Gamma_{r, r'}$ when applied on a code state of $\cC_r$ with up to a constant fraction of errors has two possible outcomes:
\begin{itemize}
\item[$(i)$] With probability $\tau_{r}$, the interface fails entirely, in the sense that the logical information is corrupted by an arbitrary noise channel $\cZ_{r \to r'}$ acting on the whole block of $\cC_r$. In this case, we do not control the error.

\item[$(ii)$] With the remaining probability $1 - \tau_{r}$, the interface works as intended: the encoded information in $\cC_r$ is mapped into $2^{r-r'}$ smaller code blocks, each encoded in $\cC_{r'}$ up to a constant fraction of errors, where the constant fraction is at most the input fraction. 
\end{itemize}

In the following lemma, we present the partial interface $\Gamma_{r, r'}$ for the special case $r=O(1), r' = 1$. Compared to Lemma~\ref{lem:succ-int}, we obtain a more explicit characterization of the output noise in Lemma~\ref{lem:succ-int-1}, in particular, it is a local stochastic channel. Note from Lemma~\ref{lem:ls-lw} that a local stochastic channel can be treated as low-weight noise with high probability.
\begin{lemma}\label{lem:succ-int-1}
There exists a fixed threshold value $\delta_{th} > 0$ and constants $\mu, c, c' > 0$ such that the following holds:

\smallskip Consider the code $\cC_r$, where $r = O(1)$. Let $\cN: \bL((\bC^2)^{\otimes n_r} \otimes R) \to \bL((\bC^2)^{\otimes n_r}  \otimes R)$, with $R$ being a reference system, be a quantum channel with stabilizer-reduced weight (relative to $\cC_r$) bounded as $|\cN|_{\mathrm{red}} \leq \mu n_r$ with respect to $n_r$ qubits in its input.

\smallskip Then, there exists a quantum circuit $\Gamma_{r, 1}$ with $n_r$ qubit input and $m_r$ qubit output and with a qubit count and depth depending only on $r$, such that for the quantum channel $\tilde{\cT}_{\Gamma_{r, 1}}$ corresponding to the noisy realization of $\Gamma_{r, 1}$ under circuit-level stochastic noise with parameter $\delta < \delta_{th}$, we have
\begin{multline} \label{eq:gamma-r-1}
    (\tilde{\cT}_{\Gamma_{r, r'}} \otimes \cI[R]) \circ \cN \circ (\cE_r \otimes \cI[R])    = (1 - \tau_r)  \cV \circ \left(\cI_{m_r} \otimes  \cI[R]\right)  \\ + \tau_r \: (\cZ_{r \to 1} \otimes \cI[R]) \circ \cN \circ (\cE_r \otimes \cI[R]),
\end{multline}
where the parameter $\tau_r \in [0, 1]$ is bounded as $\tau_r \le \lambda_r (c \delta)^{c' m_r}$, with $\lambda_r$ only depending on $r$. The map $\cV: \bL((\bC^2)^{\otimes m_{r}} \otimes R) \to \bL((\bC^2)^{\otimes m_{r} } \otimes R)$ is a local stochastic channel with parameter $\delta' = \lambda'_{r} \delta$ with respect to $m_r$ qubit system in its input, where $\lambda'_{r} > 0$ only depends on $r$. $\cZ_{r  \to 1}$ is an arbitrary channel with $n_{r}$ qubit input and $m_r$ qubit output.
\end{lemma}

We now describe the construction of  $\Gamma_{r, r'}$, and prove Lemma~\ref{lem:succ-int} and Lemma~\ref{lem:succ-int-1}.

\subsection{Description of the partial decoding interfaces}  \label{sec:partial-interfaces-tele}
We use quantum teleportation to construct the  interface $\Gamma_{r, r'}$. To this end, we consider  $m_r$ pairs of maximally entangled states on systems $A' = \{a'_1, \dots, a'_{m_r}\}$ and $B' = \{b'_1, \dots, b'_{m_r}\}$ as follows,
\begin{equation} \label{eq:gamma-noisy-tau}
  \ket{\psi}_{A', B'}  := \left(\otimes_{j \in [m_r]} \frac{\ket{0}_{a'_j} \otimes \ket{0}_{b'_j} + \ket{1}_{a'_j} \otimes \ket{1}_{b'_j}}{\sqrt{2}} \right).
\end{equation}
 Consider the encoding isometry $V_r : (\mathbb{C}^2)^{\otimes m_r} \to (\mathbb{C}^2)^{\otimes n_r}$ as defined in Eq.~\eqref{eq:iso-subsytem}, i.e., all the gauge qubits corresponding to $\cC_r$ are all fixed to the $\ket{0}$ state. We divide $B'$ into subsystems $B'_i$ for $i \in [2^{r-r'}]$, each containing $m_{r'}$ qubits, as follows,
\[ B'_i = \{ b_{(i-1) m_{r'} + 1}, \dots, b_{i m_{r'}}  \} \]
With the system $A'$, we associate an $n_r$ qubit system $A$ and with each $B'_i$, we associate an $n_{r'}$ qubit system $B_i$ and let $B := \cup_{i \in [2^{r-r'}]} B_i$. We now define the following joint state on $A$ and $B$,
\begin{equation}\label{eq:ent-code-state}
 \ket{\psi_{r,r'}}_{A, B} :=    \Big( V_r[A' \to A] \otimes (\otimes_{i  \in [2^{r-r'}]} V_{r'}[B'_i \to B_i]) \Big)  \ket{\psi}_{A', B'}.
\end{equation} 
Note that in Eq.~(\ref{eq:ent-code-state}), the subsystem $A'$ of $\ket{\psi}_{A', B'}$ is encoded in $\cC_r$ while $B'$ is encoded in $2^{r-r'}$ code blocks of $\cC_{r'}$. 

\begin{figure}[!t]
\centering
\resizebox{\linewidth}{!}{
 \begin{tikzpicture}
 \tikzset{
  half circle/.style={
      semicircle,
      shape border rotate=90,
      anchor=chord center
      }
}
\draw
(-1, -3.5) node[draw, half circle, minimum size = 2.5 cm](e){$\overline{\Psi}_{r, r'}$}
;
\draw
(-4.2, 0) node[left](a1){$Q$}
(2.0, 0) node[phase](b1){}
(3, 0) node[draw, inner sep = 8] (c1){$H$}
(4.5, 0) node[meter] (d1) {}
(7.5, 0) node[draw, inner sep = 8] (e1) {Class. Proc.}
;
\draw
(-1.8, 0) node[draw, minimum size = 10mm, minimum width = 4cm](z1){$\Phi_{\mathrm{EC}^{r, s_1}}$}
;
\draw
(-1, -1.5) node[](a2){}
(2.0, -1.5) node[circlewc](b2){}
(3, -1.5) node[] (c2){}
(4.5, -1.5) node[meter] (d2) {}
(7.5, -1.5) node[draw, inner sep = 8] (e2) {Class. Proc.}
;
\draw
(-1, -3) node[](a3){}
(-1, -5.5) node[](a4){}
($(e1.east) + (1, -4.2)$) node[draw, minimum size = 3.0cm, , minimum width = 1.5cm] (e3) {$X^{\bu_1} Z^{\bu_2}$}
;
\draw
(a1) to (z1.west)
(z1.east) to (b1)
(b1) to (c1)
(c1) to (d1)

(a2 -| e.east) to (b2)
       (b2) to (d2)

      (a2 -| e.east)  ++ (0.2, 0)  node[above](){$A$}
       (a3 -| e.east)  ++ (0.3, 0)  node[above](){$B_1$}
        (a4 -| e.east)  ++ (0.6, 0)  node[above](){$B_{2^{r-r'}}$}
      ;

 \draw
 (a3 -| e.east)  ++ (6.0, 0) node[draw, minimum size = 10mm, minimum width = 7cm](z2){$\Phi_{\mathrm{EC}^{r', s_2}}$}
  (a4 -| e.east)  ++ (6.0, 0) node[draw, minimum size = 10mm, minimum width = 7cm](z3){$\Phi_{\mathrm{EC}^{r', s_2}}$}

  (a3 -| e.east)  to (z2.west) 
  (z2.east)  to (a3 -| e3.west)

  (a4 -| e.east) to (z3.west) 
  (z3.east)  to (a4 -| e3.west)
 ;

\draw
(b1) to (b2)
(a3 -| e3.east) to ++(0.5, 0)
(a4  -| e3.east) to ++(0.5, 0)
;
      \draw[double]
      (d1.east) to node[above](){ $\bm_1$} (e1.west)
      (d2.east) to node[above](){$\bm_2$} (e2.west)
      (e1.east) to node[above](){$\bu_1$} ++(1, 0) node[phase](){} to (e3.north)
      (e2.east) to node[above](){$\bu_2$} ++(1, 0) node[phase](){} to (e3.north)
      ;

      \draw[dashed]

       ($(a1) + (5.5, 0.9)$) to ++(0, -7.5)

       ($(a1) + (3.0, 1.0)$) node[left](){$(1)$}
       ($(a1) + (9.5, 1.0)$) node[left](){$(2)$}
      ;

      \draw
      ($0.5*(a3) + 0.5*(a4)$) ++ (1, 0) node[rotate = 90](){$\cdots$}
    ($0.5*(a3) + 0.5*(a4)$) ++ (6, 0) node[rotate = 90](){$\cdots$}
      ;
 \end{tikzpicture}
 }
\caption{Construction of the partial decoding interface $\Gamma_{r,r'}$. The input code block $Q$  undergoes $s_1$ rounds of error correction during which the entangled state $\ket{\psi_{r,r'}}_{A, B}$ in Eq.~\ref{eq:ent-code-state} is prepared using the quantum circuit $\overline{\Psi}_{r,r'}$. A transversal logical Bell measurement is then performed between $Q$ and subsystem $A$ of the $\ket{\psi_{r,r'}}_{A, B}$. The measurement outcomes $m_1$ and $m_2$ are classically processed to obtain the logical corrections $u_1$ and $u_2$. The output blocks $B_1,\ldots,B_{2^{r-r'}}$ undergo $s_2$ rounds of error correction, followed by the Pauli correction $X^{u_1}Z^{u_2}$. Here $s_1,s_2=O(\poly(m_r))$.}
\label{fig:tele-circ}
\end{figure}

\smallskip The quantum state $\ket{\psi}_{A', B'}$ can be prepared using a  circuit of constant depth (i.e., not depending on the parameter $r$). This corresponds to  first initializing qubits in $A'$ and $B'$ in $\ket{0}$ state, then applying the Hadamard gate in parallel on qubits in $A'$ and then performing the $\cnot$ gates in parallel from $a'_j \to b'_j, \forall j \in [m_r]$. Moreover, the isometry $V_r$ can be realized by a circuit, working on $O(m_r)$ qubits and of depth $O(m_r)$~\cite{maslov2018shorter, tamiya2024polylog}. Therefore, $\ket{\psi_{r,r'}}_{A, B}$ can be prepared by a quantum circuit $\Psi_{r,r'}$ of depth $O(m_r)$, operating on $O(m_r)$ qubits.

\smallskip We consider the fault-tolerant circuit $\overline{\Psi}_{r,r'}$ corresponding to the state preparation circuit $\Psi_{r,r'}$ from Theorem~\ref{thm:state-prep}, where we chose the positive integer $k$ as
\begin{equation}
    k = O(m_r).
\end{equation}
The quantum circuit $\overline{\Psi}_{r,r'}$  operates on $O(m_r \mathrm{poly}(m_r))$ qubits and has depth $O(\mathrm{poly}(m_r))$.

\smallskip We now define the interface $\Gamma_{r,r'}$, with $n_r$ qubit input and $n_r = 2^{r-r'} n_{r'}$ qubit output, using the state preparation circuit $\overline{\Psi}_{r, r'}$ and the error correction steps $\Phi_{\mathrm{EC}^{r,s}}$ as follows (see also Fig.~\ref{fig:tele-circ}),
\begin{enumerate}
\item[$(i)$] Firstly, we run the state preparation circuit $\overline{\Psi}_{r, r'}$ while in parallel applying $s_1$ error correction steps $\Phi_{\mathrm{EC}^{r,s_1}}$ on the input qubits $Q := \{q_1, \dots, q_{n_r}\}$, where $s_1 = O(\mathrm{poly}(m_r))$ is given by the depth of $\overline{\Psi}_{r, r'}$. 

\item[$(ii)$] We then apply the logical Bell measurements on systems $Q, A$, which consist of transversal $\cnot$ gates between $Q$ and $A$, transversal Hadamard gates on $Q$, and then transversal computational basis measurements on $Q$ and $A$. These measurements output classical bit strings  $\bm_1, \bm_2 \in \{0, 1\}^{n_r}$, on which a processing is applied to get logical measurement outcomes $\bu_1, \bu_2 \in \{0, 1\}^{m_r}$. While Bell measurements and the classical processing are done, one applies in parallel error correction steps on $B_i, i \in [2^{r-r'}]$. 

The classical processing needs to be efficient; its time complexity is determined by the efficiency of the classical decoder associated with the QLDPC code. We refer to~\cite[Chapter 6]{grospellier2019constant} for a detailed description of the classical processing procedure corresponding to the logical Bell measurements. As stated in Theorem~\ref{thm:qldpc-tanner}, there exists a parallel classical decoder with complexity  $O(\mathrm{polylog}(m_r))$. However, for the sake of simplicity, we will assume the classical decoder complexity to be $O(\poly(m_r))$ since then we have $s_1, s_2 = O(\poly(m_r))$. Moreover, in our construction of the interface for Theorem~\ref{thm:ft-cons-int-main}, we will need to sequentially apply $\Gamma_{r,r'}, \Gamma_{r',r''}, \dots $; therefore, $O(\poly(m_r))$ error correction steps will in any case be applied on  $B_i, i \in [2^{r-r'}]$ after applying $\Gamma_{r,r'}$.
\end{enumerate}

We note that only protected logical qubits of the subsystem code $\cC_r$ would be teleported into the two code blocks of $\cC_{r-1}$. The gauge qubits of $\cC_r$ are set to classical bits during the teleportation protocol.  Fresh gauge qubits are introduced in $\cC_{r-1}$ independent of those in $\cC_{r}$.

\paragraph{Qubit count of $\Gamma_{r,r'}$:} Using that the state preparation circuit $\overline{\Psi}_{r, r'}$ operates on $O(m_r \poly(m_r))$ qubits, and the error correction circuit operates on $O(m_r)$ qubits for any $r$, it follows that the interface $\Gamma_{r, r'}$ operates on fewer than $\theta m_r p_1(m_r)$ qubits for some constant $\theta > 0$, and a polynomial $p_1(\cdot)$.

\paragraph{Depth of $\Gamma_{r,r'}$:} Moreover, using that the depth of $\overline{\Psi}_{r, r'}$ is $O(\poly(m_r))$, and the depth of classical processing circuit is $O(\poly(m_r))$, it follows that the depth of $\Gamma_{r, r'}$ is upper bounded by $d p_2(m_r)$, where $d > 0$ is a constant and $p_2(\cdot)$ is a polynomial.

\medskip Therefore, the number of qubits in $\Gamma_{r, r'}$ and its depth are according to Lemma~\ref{lem:succ-int}. A detailed error analysis of $\Gamma_{r, r'}$ is done in Section~\ref{sec:proof-succ-int}, where it is shown that the noisy version of $\Gamma_{r, r'}$ satisfies Eq.~\eqref{eq:gamma-r-r'} and Eq.~\eqref{eq:bound-tau-rk}, hence, completing the proof of Lemma~\ref{lem:succ-int}.

\smallskip Finally, we briefly explain the choice $r-r' = O(1)$: the failure probability of the error correction $\Phi_{\mathrm{EC}^{r', s_2}}$ is given by $s_2 \varepsilon(r', \delta)$, where $\varepsilon(r', \delta) $ is from 
Lemma~\ref{lem:mult-ec}, which provides the upper bound  
$s_2 \varepsilon(r', \delta) = O(\poly(m_r) (c\delta)^{n_{r'}})$.
Choosing $r- r' = O(1)$ makes sure that $n_{r'}=O(m_r)$ and thus that the error probability goes to zero for large $r$.

\paragraph{Quantum circuit for $\Gamma_{r, 1}$ for Lemma~\ref{lem:succ-int-1}:} 

The structure of the quantum circuit $\Gamma_{r, 1}$ is the same as $\Gamma_{r, r'}$ in Fig.~\ref{fig:tele-circ} by taking $r' = 1$. Recall that the code $\cC_{1}$ corresponds to the trivial code encoding one logical qubit in one physical qubit. Therefore, systems $B_i, i \in [2^{r-1}]$ contain only one qubit and $V_{r'}$ in Eq.~\eqref{eq:ent-code-state} corresponds to the identity channel on single qubit. Moreover, the error correction step $\Phi_{\mathrm{EC}^{r',s_2}}$ is replaced by $s_2$ layers of idle gate applied on the single qubit.

\subsection{Error analysis of partial decoding interfaces} \label{sec:proof-succ-int}
In this section, we provide an error analysis of the quantum circuit $\Gamma_{r,r'}$ from Fig.~\ref{fig:tele-circ} and show that its noisy version satisfies Fig.~\ref{fig:tele-circ} satisfies Eq.~\eqref{eq:gamma-r-r'} and Eq.~\eqref{eq:bound-tau-rk}; hence proving  Lemma~\ref{lem:succ-int}. The error analysis of $\Gamma_{r, 1}$ is similar to the analysis of $\Gamma_{r,r'}$ with one key difference. We mention this difference and sketch a proof of Lemma~\ref{lem:succ-int-1} in Remark~\ref{rem:proof-gamma-r-1} at the end of this section.

\smallskip Below, we prove Lemma~\ref{lem:succ-int}.
\begin{proof}[Proof of Lemma~\ref{lem:succ-int}]
We consider the channel $\tilde{\cT}_{\Gamma_{r,r'}}$ corresponding to the noisy version of $\Gamma_{r, r'}$ under circuit-level stochastic channel with parameter $\delta > 0$. We take the input state on system $Q$ as $\cN(\proj{\phi_r})$, where $\proj{\phi_r}$ is a code state $\ket{\phi_r} := V_r\ket{\phi}$ of $\cC_r$, with $\ket{\phi}$ being an arbitrary $m_r$ qubit state, and $\cN$ is a quantum channel with stabilizer-reduced weight $\mu n_r$. 

\smallskip We will show that $\tilde{\cT}_{\Gamma_{r,r'}}\left(\cN(\proj{\phi_r})\right)$ is equal to $\cN'( \proj{\phi_{r'}})$ with a probability $(1 - \tau_{r})$, where $\tau_{r}$ is according to Lemma~\ref{lem:succ-int}, $\ket{\phi_{r'}} = V_{r'}^{\otimes 2^{r-r'}}\ket{\phi}$ is the encoded version of $\ket{\phi}$ in $2^{r-r'}$ blocks of $\cC_{r'}$ denoted by $B_1, B_2, \dots, B_{2^{r-r'}}$ and $\cN'$ is a quantum channel having stabilizer-reduced weight $\mu n_{r'}$ with respect to each $B_1, B_2, \dots, B_{2^{r-r'}}$. We will do this in the following three steps: 
\begin{enumerate}
    \item [(i)] This step implements the state preparation circuit $\overline{\Psi}_{r,r'}$, preparing the quantum state $\proj{\psi_{r,r'}}$ on systems $A, B_1, B_2, \dots, B_{2^{r-r'}}$ as defined in Eq.~\eqref{eq:ent-code-state}. While the state preparation circuit $\overline{\Psi}_{r,r'}$ is executed, error correction steps $\Phi_{\mathrm{EC}^{r, s_1}}$, for $s_1 = O(\poly(n_r))$ are applied on system $Q$.

For the noisy version, we have $\cT_{\Phi_{\mathrm{EC}^{r, s_1}}} \to \tilde{\cT}_{\Phi_{\mathrm{EC}^{r, s_1}}}$ and $\cT_{\overline{\Psi}_{r,r'}} \to \tilde{\cT}_{\overline{\Psi}_{r,r'}}$.
Using Lemma~\ref{lem:mult-ec} for fault-tolerant error correction, we can remove $\tilde{\cT}_{\Phi_{\mathrm{EC}^{r, s_1}}}$  while changing $\cN$ to $\cN^{(1)}$, whose stabilizer-reduced weight is also $\mu n_r$ with respect to the code block of $\cC_r$. The error probability associated with this is given by, $\tau_e(1) \leq s_1 \varepsilon(r, \delta) \leq O(\poly(m_r)(c\delta)^{c'n_r})$.

From Theorem~\ref{thm:state-prep},  $\tilde{\cT}_{\overline{\Psi}_{r,r'}}$ prepares the state $\proj{\psi_{r,r'}}$ up to a local stochastic channel $\cW$ with parameter $2 c\delta$ applied on systems $A, B_1, \dots, B_{2^{r-r'}}$, with an error probability $\tau_s(1) = O(|\Psi_{r, r'}|(c\delta)^k) = O(\poly(m_r)(c\delta)^{c'm_r})$ using $|\Psi_{r, r'}| = O(m_r)$ and $k = O(m_r)$. Since our fault-tolerant error correction provides protection against a low-weight error, we approximate the local stochastic channel $\cW$ as a low-weight channel using Lemma~\ref{lem:ls-lw}.\footnote{ Recall that Lemma~\ref{lem:ls-lw} provides bound in terms of ordinary weight as defined in Def.~\ref{def:weight}.} In particular, we approximate $\cW$ by a channel $\cN_s$, which has weight $\mu n_{r}$ with respect to system $A$ and weight $\mu n_{r'}$ with respect to each $B_1, \dots, B_{2^{r-r'}}$, with an error probability $\tau_s(2) \leq (2^{r-r'} + 1) (2^{\frac{h_2(\mu)}{\mu}} 2 c \delta)^{\mu n_{r'}}$.

\smallskip Therefore, the noisy version of this step realizes a quantum channel $\tilde{\cT}_1$ as follows,
\begin{equation} \label{eq:step-1-noisy}
\tilde{\cT}_1\left(\cN(\proj{\phi_r})\right) = (1 - \tau_1) {\cN}^{(1)}(\proj{\phi_r}) \otimes \cN_s(\proj{\psi_{r,r'}}) + \tau_1 \cZ_1\left(\cN(\proj{\phi_r})\right), 
\end{equation}
where $\cZ_1$ is an arbitrary error channel and using a union bound, we bound the error probability as
 \begin{align} 
  \tau_1 &\leq  \tau_e(1) + \tau_s(1) + \tau_s(2)  \nonumber \\  
  & \leq O(\poly(m_{r'}) (c\delta)^{c' m_{r'}}) + (2^{r-r'} + 1) (2^{\frac{h_2(\mu)}{\mu}} 2c\delta)^{\mu m_{r'}} \label{eq:ub-tau-1}
 \end{align}
  where in the last line, we have used $\poly(m_r) =   O(\poly(m_{r'}))$ using $r-r' = O(1)$ and $n_r > m_{r'}$.
  
\item [(ii)] The second step corresponds to applying transversal Bell measurements on qubits in $Q, A$ while error correction steps $\Phi_{\mathrm{EC}^{r', s_2}}$ for $s_2 = O(\mathrm{polylog}(n_r))$ are applied on systems $B_1, \dots, B_{2^{r-r'}}$. Note that Bell measurements are applied by first applying transversal $\cnot$, then transversal Hadamard, and then transversal single qubit measurements as depicted in Fig.~\ref{fig:tele-circ}.

Below, we provide an error analysis of this step. To this end, we will post-select on the first ``good" term on the right hand side of Eq.~\eqref{eq:step-1-noisy} in Step (i) and apply the noisy version of this step on it.

We first remove the noisy version of the error correction circuit $\Phi_{\mathrm{EC}^{r', s_2}}$ using Lemma~\ref{lem:mult-ec}, while changing $\cN_s$ to $\cN'_s$, whose stabilizer-reduced weight is also $\mu n_r$ with respect to $A$ and $\mu n_{r'}$ with respect to $B_1, \dots, B_{2^{r-r'}}$. The error probability associated with this is given by $\tau_e(2) \leq s_2 2^{r-r'} \varepsilon(r', \delta) \leq O(\poly(m_{r'}) (c\delta)^{c'n_{r'}})$.

We now note that transversal Bell measurements on systems $Q$ and $A$ can be seen as ideal transversal Bell measurements up to a noise channel of weight $\mu n_r$ with respect to each system $Q$ and $A$. This can be seen as follows:

Using Remark~\ref{rem:noisy-g-ls}, we have that noisy version of transversal gates under circuit level stochastic noise with parameter $\delta$ is equal to the ideal gate up to a local stochastic noise with parameter $\delta$, where the local stochastic channel is applied after the ideal gate for unitary gates and before the ideal gate for measurement gates.

For a circuit layer, where only single qubit unitary gates are applied, we can propagate a local stochastic channel through it without changing its parameter. In particular, for a local stochastic channel $\cW$  with parameter $\delta$ with respect to a set of $n$ qubits and $U$ being $n$-fold tensor product of single qubit unitaries, $U W U^\dagger$ is a local stochastic channel with parameter $\delta$ with respect to the set of $n$ qubits. Using this fact, we can accumulate the local stochastic noise from Remark~\ref{rem:noisy-g-ls} associated with $\cnot$, Hadamard, and measurement gates just before the transversal measurements. Using the fact that the composition of two local stochastic channels is again a local stochastic channel with a parameter equal to the sum of parameters of the original channels~\cite{christandl2025fault, bravyi2020quantum}, we have that the accumulated channel is a local stochastic channel with parameter $3 \delta$, with respect to systems $Q$ and $A$.

Using Lemma~\ref{lem:ls-lw}, we can further approximate this local stochastic channel as a channel of weight $\mu n_r$ with respect to each system $Q$ and $A$, with an error probability $\tau_b \leq 2(2^{\frac{h_2(\mu)}{\mu}} 3\delta)^{\mu n_r}$.

Moreover, we can propagate channels $\cN^{(1)}$ acting on $Q$ and $\cN'_s$ acting on $A, B_1, \dots, B_{2^{r-r'}}$ through the transversal $\cnot$ and Hadamard, yielding a joint channel ${\cN}_s^{(1)}$, which has stabilizer-reduced weight $2\mu n_r$ with respect to $Q$ and stabilizer-reduced weight $2\mu n_{r}$ with respect to $A$ and  $\mu n_{r'}$ with respect to $B_1, \dots, B_{2^{r-r'}}$.

\smallskip Combining the error due to Bell measurement with ${\cN}_s^{(1)}$, we have an overall channel $\overline{\cN}$, which has stabilizer-reduced weight $3\mu n_r$ with respect to $Q$ and $A$, and $\mu n_{r'}$ with respect to $B_i, i \in  [2^{r-r'}]$ (see also Remark~\ref{rem:com-red-weight}). Hence, we have a quantum circuit that corresponds to the ideal teleportation circuit up to the noise channel $\overline{\cN}$ applied just before the single qubit measurements on systems $Q$ and $A$.

Using standard teleportation protocol, the state of the joint system quantum system before the Bell measurements is given by $\overline{\cN}(\proj{\Lambda})$, where
\begin{equation} \label{eq:ovr-T-12}
  \ket{\Lambda} :=  \frac{1}{2^{m_r}} \sum_{\bu, \bv \in \{0, 1\}^{m_r}} \ket{{\bu}_r} \otimes \ket{{\bv}_r} \otimes P^{r'}_{\bu, \bv}\ket{\phi_{r'}},
\end{equation}
where $\ket{\bu_r} := V_r \ket{\bu_r}$, $\ket{\bv_r} :=  V_r \ket{\bv_r}$. $P^{r'}_{\bu, \bv}$ is the logical Pauli in code $\cC_{r'}$ corresponding to the Pauli $P_{\bu, \bv} = X^\bu  Z^\bv$, that is, $P^{r'}_{\bu, \bv}  \circ V_{r'}^{\otimes 2^{r-r'}} = V_{r'}^{\otimes 2^{r-r'}} \circ P_{\bu, \bv}$.

\smallskip Since the stabilizer-reduced weight of $\overline{\cN}$ with respect to systems $Q$, and $A$ is bounded by $3 \mu n_r$, we can get the correct value of the Bell measurements $\bu, \bv \in \{0, 1\}^{m_r}$ from outcomes of (physical) Bell measurements on systems $Q, A$ using classical error correction given that $\mu$ is chosen to be sufficiently small, that is, $\mu \leq \frac{\beta}{3}$, where $\beta$ is from Eq.~\ref{eq:min-distance}.

The error probability corresponding to this step is given by,
\begin{align}
    \tau_2  &\leq \tau_e(2) + \tau_b \nonumber \\
    & = O(\poly(m_{r'}) (c\delta)^{c' m_{r'}}) + (2^{\frac{h_2(\mu)}{\mu}} 3\delta)^{\mu m_{r'}}. \label{eq:ub-tau-2}
\end{align}

\item[(iii)] In the last step, we apply the corresponding logical Pauli correction $P^{r'}_{\bu, \bv}$ corresponding to logical measurement outcomes $\bu, \bv \in \{0, 1\}^{m_r}$.

After the logical Pauli correction, we are left with the state $\cN'(\proj{\phi_{r'}})$, where $\cN'$ has stabilizer-reduced weight $\mu n_{r'}$ with respect to $B_i, i \in [2^{r-r'}]$ with the overall error probability bounded as below,
\begin{align}
 & \tau_{r} \leq \tau_1  + \tau_2  \nonumber\\
  & \leq O(\poly(m_{r'}) (c\delta)^{c' m_{r'}}) +  (2^{r-r'} + 1) (2^{\frac{h_2(\mu)}{\mu}} 2c\delta)^{\mu m_{r'}} + (2^{\frac{h_2(\mu)}{\mu}} 3\delta)^{\mu m_{r'}}, \nonumber \\
\end{align}
where the second inequality uses bound on $\tau_1$ and $\tau_2$ from Eq.~\eqref{eq:ub-tau-1} and Eq.~\eqref{eq:ub-tau-2}, respectively. 

By redefining $c \leftarrow \max(c, 2^{\frac{h_2(\mu)}{\mu}} 2c, 3)$ and $c' \leftarrow  \max(c', \mu)$ and using $r-r' = O(1)$, we can absorb the term $(2^{r-r'} + 1) (2^{\frac{h_2(\mu)}{\mu}} \delta)^{\mu m_{r'}}$ in $O(\poly(m_{r'})(c \delta)^{c' m_{r'}})$ giving us,
\begin{align}
    \tau_{r} &= O(\mathrm{poly}(m_{r'})(c \delta)^{c' m_{r'}}). \label{eq:tau-rk}
\end{align}

\medskip Since steps (i), (ii) and (iii) hold for arbitrary state $\ket{\phi}$, we have
\begin{multline} \label{eq:final-proof}
    (\tilde{\cT}_{\Gamma_{r, r'}} \otimes \cI[R]) \circ \cN \circ (\cE_r \otimes \cI[R])    = (1 - \tau_{r}) \:  \cN' \circ \left((\otimes_{i \in [2^{r-r'}]}\cE^i_{r'}) \otimes  \cI[R]\right)  \\ + \tau_{r} \: (\cZ_{r \to r'} \otimes \cI[R]) \circ \cN \circ (\cE_r \otimes \cI[R]),  
\end{multline}
where $\tau_{r}$ is bounded as in Eq.~\eqref{eq:tau-rk}.

\smallskip Finally, from Eq.~\eqref{eq:final-proof} and Eq.~\eqref{eq:tau-rk}, we have that Eq.~\eqref{eq:gamma-r-r'} and Eq.~\eqref{eq:bound-tau-rk} in Lemma~\ref{lem:succ-int} hold.
\end{enumerate}
\end{proof}

\begin{remark} \label{rem:proof-gamma-r-1}
The system $B_i, \forall i \in [2^{r-1}]$ in the circuit $\Gamma_{r, 1}$ are single qubit systems; therefore, the idle gates are applied on them instead of error correction. The noise on them accumulates when Bell measurements and classical processing on measurement outcomes are performed. Since Bell measurement corresponds to a constant depth circuit and the time of classical processing is $O(\poly(n_r))$, hence $O(\poly(n_r))$ number of layers of idle gates are applied. Using the fact that $r$ is fixed to a constant, we have a fixed number of layers of idle gate depending on $r$. Using Remark~\ref{rem:noisy-g-ls}, we have that the noisy version of idle gates correspond to a local stochastic channel with parameter $\delta$ with respect to systems to the joint system of $2^{r-1}$ qubits corresponding to $B_1, \dots, B_{2^{r-1}}$.

Moreover, for the local stochastic channel $\cW$ due to the noisy realization of state preparation circuit $\overline{\Psi}_{r, r'}$ with $r' = 1$, we only need to approximate it with a low weight error channel $\cN_s$ with respect to system $A$ according to Lemma~\ref{lem:ls-lw}. The channel $\cN_s$ has weight $\mu n_r$ with respect to system $A$ and is local stochastic with parameter $2 c \delta$ with respect to the joint system of $2^{r-1}$ qubits corresponding to $B_1, \dots, B_{2^{r-1}}$. Therefore, combining all the local stochastic channel on the joint system of $2^{r-1}$ qubits, due to state preparation and idle layers, we have a local stochastic channel with parameter $\lambda'_r \delta$ acting on it, where $\lambda'_r$ is a constant depending on $r$.

\smallskip The rest of the proof is the same as Lemma~\ref{lem:succ-int-1}.
\end{remark}

\subsection{Partial decoding interfaces for multiple code blocks} \label{sec:many-blocks}
The statement of Lemma~\ref{lem:succ-int} is modular, and can be easily extended to multiple input blocks of $\cC_r$ as given in the following corollary.

\begin{corollary} \label{cor:succ-int-h}
There exists a polynomial $p_3(\cdot)$, a fixed threshold value $\delta_{th} > 0$ and constants $\mu, \lambda, c, c', > 0$ such that the following holds:

\smallskip Consider codes $\cC_r$ and $\cC_{r'}$, where $r > r'$ such that $r - r' = O(1)$. Let $\cN: \bL((\bC^2)^{\otimes n_r h} \otimes R) \to \bL((\bC^2)^{\otimes n_r h}  \otimes R)$, with $R$ being a reference system, be a quantum channel with stabilizer-reduced weight (relative to $\cC_r$) bounded as $|\cN|_{\mathrm{red}} \leq \mu n_r$ with respect to each $n_r$ qubit in its input.

\smallskip Then, for the interface circuit $\Gamma_{r, r'}$ from Lemma~\ref{lem:succ-int}, under circuit-level stochastic channel with parameter $\delta$, we have\footnote{$\Gamma^i_{r, r'}$, $\cE^i_r$ and $\cE_{r'}^{i, j}$ are copies of $\Gamma_{r, r'}$, $\cE_r$ and $\cE_{r'}$, respectively.}
\begin{multline} \label{eq:int-h}
    \left((\otimes_{i \in [h]}\tilde{\cT}_{\Gamma^i_{r, r'}}) \otimes \cI[R]\right) \circ \cN \circ \left((\otimes_{i \in [h]} \cE^i_r) \otimes \cI[R]\right)  \\ = \sum_{F \subseteq [h]} (1 - \tau_{r})^{h - |F|} \tau^{|F|}_{r} \: \left( (\otimes_{i \in [h]\setminus F} \cI^i_r) \otimes (\otimes_{i \in F} \cZ^{i}_{r \to r'}) \right)\circ \cN'_F \\ \circ \left((\otimes_{i \in h \setminus F, j \in [2^{r-r'}] }\cE_{r'}^{i, j}) \otimes(\otimes_{i \in F} \cE^i_r) \otimes \cI[R] \right), 
\end{multline}
where $\tau_{r} \le \lambda p_3(m_{r'}) (c \delta)^{c' m_{r'}}$.\footnote{$\tau_{r}$ is the same bound as in Lemma~\ref{lem:succ-int}.} $\cN'_F$ is a quantum channel with stabilizer-reduced weight $\mu n_r$ with respect to each block of $\cC_r$ and $\mu n_{r'}$ with respect to each block of $\cC_{r'}$ in its input.  $\cZ^{i}_{r \to r'}$ is an arbitrary channel with $n_r$ qubit input and  $2^{r-r'} n_{r'}$ qubit output.
\end{corollary}

\begin{proof}
We will apply Lemma~\ref{lem:succ-int} iteratively for $i = 1, 2, \dots, h$ and remove the interface $\tilde{\cT}_{\Gamma^i_{r,r'}}$ from the left hand side of Eq.~\eqref{eq:int-h} one by one.
For $i = 1$, we treat the remaining systems $2, \dots, h$ as part of the system $R$ in Lemma~\ref{lem:succ-int}. Then, we have
\begin{multline} \label{eq:iter-1}
    ( \tilde{\cT}_{\Gamma^{1}_{r, r'}} \otimes (\otimes_{i \in [2:h]} \cI^i_r) \otimes  \cI[R]) \circ \cN \circ ((\otimes_{i \in [h]} \cE^i_r) \otimes \cI[R]) \\ = (1 - \tau_{r}) \cN'^{(1)} \circ ( (\otimes_{j \in [2^{r-r'}]}\cE_{r'}^{1, j}) \otimes (\otimes_{i \in [2:h]} \cE^i_r) \otimes \cI[R]) \\ + \tau_{r}( \cZ^{1}_{r \to r'} \otimes (\otimes_{i \in [2:h]} \cI^i_r) \otimes  \cI[R]) \circ \cN \circ ((\otimes_{i \in [h]} \cE^i_r) \otimes \cI[R]), 
\end{multline}
where $\cN'^{(1)}: \bL((\bC^2)^{\otimes 2^{r-r'} n_{r'}} \otimes (\bC^2)^{\otimes (h-1) n_r} \otimes R) \to  \bL((\bC^2)^{\otimes 2^{r-r'} n_{r'}} \otimes (\bC^2)^{\otimes (h-1) n_r} \otimes R)$ is a channel with stabilizer-reduced weight $\mu n_{r}$ with respect to each block $\cC_{r'}$ in its input. Moreover,
the stabilizer-reduced weight of $(\tilde{\cT}_{\Gamma^{1}_{r, r'}} \otimes (\otimes_{i \in [2:h]} \cI^i_r) \otimes  \cI[R]) \circ \cN$ remains $\mu n_{r'}$ with respect to the $(h-1)$ blocks of $\cC_r$ labelled by $2, \dots, h$. Therefore, the stabilizer-reduced weight of $\cN'^{(1)}$ with respect to these blocks is also $\mu n_{r}$. Similarly, for the second term on the right hand side, the channel $( \cZ^{1}_{r \to r'} \otimes (\otimes_{i \in [2:h]} \cI^i_r) \otimes  \cI[R]) \circ \cN$ has stabilizer-reduced weight $\mu n_{r}$ with respect to the blocks of $\cC_r$ labelled by $2, \dots, h$. Obviously, the stabilizer-reduced weight of this channel is not controlled on the first block of $\cC_r$ due to the arbitrary channel.

\smallskip For the both terms on the right hand side of Eq.~\eqref{eq:iter-1}, we can apply Lemma~\ref{lem:succ-int} for $i = 2$, while treating systems $(1,j), j \in [2^{r-r'}]$ and $i \in  [3:h]$ as part of the reference system $R$. We now have
\begin{align}
& ( (\otimes_{i \in [2] }\tilde{\cT}_{\Gamma^{i}_{r, r'}})
   \otimes (\otimes_{i \in [3:h]} \cI^i_r)
   \otimes \cI[R]) \circ \cN \circ
   ((\otimes_{i \in [h]} \cE^i_r) \otimes \cI[R])  \nonumber \\
& = (1 - \tau_{r})^2 \: {\cN}^{(1, 2)} \circ
    \left( (\otimes_{i \in [2], j \in [2^{r-r'}]}\cE_{r'}^{ i, j})
    \otimes (\otimes_{i \in [3:h]} \cE^i_r)
    \otimes \cI[R] \right)  \nonumber \\
&\quad + (1 - \tau_{r}) \tau_{r} \:
    \left(\cI^{1}_r \otimes \cZ^{2}_{r \to r'}
    \otimes (\otimes_{i \in [3:h]} \cI^{1}_r) \right)
    \circ \cN^{(1)} \circ
    \left( (\otimes_{j \in [2^{r-r'}]}\cE_{r'}^{1, j})
    \otimes (\otimes_{i \in [2:h]} \cE^i_r)
    \otimes \cI[R] \right)  \nonumber \\
&\quad + \tau_{r} (1 - \tau_{r}) \:
    \left(\cZ^{1}_{r \to r'} \otimes (\otimes_{i \in [2:h]} \cI^i_r)
    \otimes \cI[R]\right)
    \circ \cN^{(2)} \nonumber \\
& \quad \quad \quad \quad \quad \quad \quad \quad  \quad \quad \quad \quad \quad \quad \quad \quad   \circ \left( \cE^{1}_{r}
    \otimes (\otimes_{j \in [2^{r-r'}]} \cE^{2, j}_{r})
    \otimes (\otimes_{i \in [2:h]} \cE^i_r)
    \otimes \cI[R] \right)  \nonumber \\
&\quad + \tau_{r}^2 \:
    \left(\cZ^{1}_{r \to r'} \otimes \cZ^{2}_{r \to r'}
    \otimes (\otimes_{i \in [3:h]} \cI^{1}_r) \right)
    \circ \cN \circ
    ((\otimes_{i \in [h]} \cE^i_r) \otimes \cI[R]),
\end{align}

where $\cN^{(1, 2)}$, $\cN^{(1)}$ and $\cN^{(2)}$ are channels with stabilizer-reduced weights $\mu n_{r'}$ with respect to blocks of $\cC_{r'}$ in their input and $\mu n_{r}$ with respect to blocks of $\cC_r$ in their input.

\smallskip Continuing in the above fashion for $i = 3, 4, \dots, h$, it is easy to see that we get Eq.~\eqref{eq:int-h}. 
\end{proof}
The following Corollary provides an extension of Lemma~\ref{lem:succ-int-1} to multiple blocks of $\cC_r$. Since the argument follows the same steps as in Corollary~\ref{cor:succ-int-h} while using Lemma~\ref{lem:succ-int-1} instead of Lemma~\ref{lem:succ-int}, we have omitted the proof.
\begin{corollary} \label{cor:succ-int-h-1}
There exists a fixed threshold value $\delta_{th} > 0$ and a constant $\mu, c, c' > 0$ such that the following holds:

\smallskip Consider codes $\cC_r$ where $r = O(1)$.
 Let $\cN: \bL((\bC^2)^{\otimes n_r h} \otimes R) \to \bL((\bC^2)^{\otimes n_r h}  \otimes R)$, with $R$ being a reference system, be a quantum channel with stabilizer-reduced weight relative to $\cC_r$ bounded as $|\cN|_{\mathrm{red}} \leq \mu n_r$ with respect to each $n_r$ qubit in its input.
Then, for the interface circuit $\Gamma_{r, 1}$ from Lemma~\ref{lem:succ-int-1}, under circuit-level stochastic channel with parameter $\delta < \delta_{th}$, we have
\begin{multline} \label{eq:int-h-1}
    \left((\otimes_{i \in [h]}\tilde{\cT}_{\Gamma^i_{r, 1}}) \otimes \cI[R]\right) \circ \cN \circ \left((\otimes_{i \in [h]} \cE^i_r) \otimes \cI[R]\right)  \\ = \sum_{F \subseteq [h]} (1 - \tau_{r})^{h - |F|} \tau^{|F|}_{r}  \: \left((\otimes_{i \in [h] \setminus F, j \in [2^{r-r'}] }\cI_{m_r}^{i, j}) \otimes (\otimes_{i \in F} \cZ^{i}_{r \to 1})  \right)\circ \cV_F \\ \circ \left((\otimes_{i \in [h] \setminus F, j \in [2^{r-r'}] }\cI_{m_r}^{i, j}) \otimes(\otimes_{i \in F} \cE^i_r) \otimes \cI[R] \right), 
\end{multline}
where $\tau_r \leq \lambda_r (c \delta)^{c'm_r}$, with $\lambda_r$ only depending on $r$.\footnote{It is the same bound as in Lemma~\ref{lem:succ-int-1}.} The map $\cV_F$ is a local stochastic channel with parameter $\delta' = \lambda'_r\delta$ with respect to joint systems of blocks of $m_r$ corresponding to the set $\{(i, j) \mid i \in [h] \setminus F, j \in [2^{r-r'}]  \}$, with $\lambda'_{r} > 0$ only depending on $r$.  $\cZ^{i}_{r \to 1}$ is an arbitrary channel with $n_r$ qubit input and  $m_r$ qubit output.
\end{corollary}

\subsection{Decoding interfaces for QLDPC codes} \label{sec:const-ovr-interface}
In this section, we provide the constant qubit overhead decoding interface $\Xi^{[h]}_r$ for Theorem~\ref{thm:ft-cons-int-main}. To do this, we first provide a decoding interface in Lemma~\ref{lem:cons-over-interface}, denoted by $\Xi^{[h]}_{r, r'}$ for arbitrary $r' < r$ ($r-r' = O(1)$ is not needed anymore), mapping $h$ blocks of $\cC_r$ into $2^{r-r'} h$ blocks of $\cC_{r'}$. The interface $\Xi^{[h]}_r$ is then obtained using $\Xi^{[h]}_{r, r'}$. 
\paragraph{Notation:}
Consider codes $\cC_r$ and $\cC_{r'}$, $r > r'$. For a positive integer $h$, consider $h$ quantum systems, labelled by $i \in [h]$, each containing $m_r$ qubits. Further, group each system of $m_r$ qubits into $2^{r-r'}$ subsystems, labelled by $(i, j), i\in h, j \in [2^{r-r'}]$, each containing $m_{r'}$ qubits. For each system $i \in [h]$, we consider the encoding isometry $\cE^i_r$ of $\cC_r$, mapping $m_r$ qubits to $n_r$ qubits. Similarly for the system $(i, j)$, we consider the encoding isometry $\cE^{i,j}_{r'}$, mapping system $m_{r'}$ qubits to $n_{r'}$ qubits. We denote by $\cI^{i}_{n_r}$ and $\cI^{i, j}_{n_{r'}}$, the identity channel on systems $i$ and $(i, j)$, respectively. 
\begin{lemma}  \label{lem:cons-over-interface}
There exists a polynomial $p(\cdot)$, a constant $\mu > 0$, a threshold value $\delta_{th} > 0$, qubit overhead constants $\theta, \theta' > 0$, such that the following holds:

\smallskip Consider circuit-level stochastic noise with parameter $\delta < \delta_{th}$. For $h \geq 1$, let $\cN: \bL((\bC^2)^{\otimes n_r h} \otimes R) \to \bL((\bC^2)^{\otimes n_r h} \otimes R)$, with $R$ being a reference system, be a channel with stabilizer-reduced weight $\mu n_{r}$ with respect to each $n_r$ qubit system in its input. Then, for any target error rate $\overline{\delta}$, there is a decreasing function $\ovr(\overline{\delta}) > 0$, such for any $r > r' \geq \ovr$, the following holds:\footnote{for larger target error $\overline{\delta}$, the level of protection $\ovr$ is smaller.}

\smallskip There exists a circuit $\Xi^{[h]}_{r, r'}$, with $n_r h$ qubit input and  $2^{r-r'} n_{r'} h$ qubit output, and working on fewer than $\theta p(m_r) m_r+ \theta' m_r h$ qubits, such that its noisy realization $ \tilde{\cT}_{\Xi^{[h]}_{r, r'}}$ satisfies
\begin{equation} \label{eq:const-over-interface-1}
    \tilde{\cT}_{\Xi_{r, r'}^{[h]}} \circ \cN \circ \left((\otimes_{i \in [h]} \cE^i_r) \otimes \cI[R] \right) =     \sum_{F \subseteq \{(i, j)| i \in [h], j \in [2^{r-r'}] \} } \: \sum_{\omega \in \Omega_F} \pr(F, \omega) \: \cT_{F, \omega},
\end{equation}
where \footnote{The set $\Omega_F$ arises from the error analysis, where its elements $\omega$ correspond to different ways in which a block $F$ can become erroneous.} $\Omega_F$  is a finite set depending on $F$ and
\begin{equation} \label{eq:const-over-interface}
\cT_{F,\omega} := \Big((\otimes_{(i, j) \not\in F} \cI^{i, j}_{n_{r'}}) \otimes  \cZ_{F, \omega} \otimes  \cI[R] \Big) \circ \cN_{F, \omega} \circ \Big((\otimes_{(i, j) \not\in F} \cE^{i, j}_{r'}) \otimes  \cE_{F, \omega} \otimes  \cI[R] \Big),
\end{equation}
with $\cE_{F, \omega}$ being a quantum channel from $|F| m_{r'}$ qubits to a larger system.\footnote{More precisely: $\cE_{F, \omega}: \bL ((\bC^2)^{\otimes |F| m_{r'}} ) \to
\bL((\bC^2)^{\otimes \sum_{r'' = r'}^{r}  g_{r''} \, n_{r''}})$
for some integers $g_{r''} \in \mathbb{Z}_{\geq 0}$, $r' \le r'' \le r$
such that $\sum_{r'' = r'}^{r} g_{r''} 2^{r''-r'} = |F|$.
Regrouping the blocks in $F$ (each of which contain $m_{r'}$ qubits) into $g_{r''}$ blocks of $m_{r''}$ qubits for all $r' \leq r'' \leq r$,  the channel $\cE_{F, \omega}$ is given by,$
    \cE_{F, \omega} := \otimes_{r'' = r'}^{r}  \: \cE_{r''}^{\otimes g_{r''}}.$
}
$\cZ_{F,\omega}$ is an arbitrary channel from this larger system to $|F| n_{r'}$ qubits. $\cN_{F, \omega}$ being a channel with stabilizer-reduced weight $\mu n_{r'}$ with respect to $n_{r'}$ qubit systems corresponding to encoding maps applied on $(i, j) \not \in F$ in its input.

Importantly, 
$\pr(F) = \sum_{\omega\in \Omega_F} \pr(F, \omega)$ is such that for any $T \subseteq \{(i, j) \mid i \in [h], j \in [2^{r-r'}]\}$, we have
\begin{equation} \label{eq:cons-err-rate}
 \mathrm{Pr}(T \subseteq F) \leq {\overline{\delta}}^{|T|}.
\end{equation}
\end{lemma}
Before giving the construction of the interface $\Xi^{[h]}_{r,r'}$ for Lemma~\ref{lem:cons-over-interface}, we first show how the lemma can be used in order to construct the interface $\Xi^{[h]}_r$ for Theorem~\ref{thm:ft-cons-int-main}.
\begin{proof}[Proof of Theorem~\ref{thm:ft-cons-int-main}]
For a fixed $r' \geq \ovr$ (i.e., $r' = O(1)$), consider the interface $\Gamma_{r', 1}$ from Lemma~\ref{lem:succ-int}. Recall that $\Gamma_{r', 1}$ maps $m_{r'}$ logical qubits encoded in $\cC_{r'}$ to $m_{r'}$ logical qubits that are encoded in $m_{r'}$ blocks of $\cC_1$, with $\cC_1$ being the trivial code, encoding a logical qubit in a physical qubit, i.e., $\cE_1 = \cI_{\bC^2}$.

\smallskip Moreover, consider the interface $\Xi^{[h]}_{r, r'}$ from Lemma~\ref{lem:cons-over-interface}, mapping $h$ blocks of $\cC_r$ into $2^{r-r'} h$ blocks of $\cC_{r'}$. 
For any $r > 0$, we obtain the interface $\Xi^{[h]}_{r}$ by first applying $\Xi^{[h]}_{r, r'}$ on $h$ blocks of $\cC_r$ and then applying $\Gamma_{r', 1}$ on each block of $\cC_{r'}$ corresponding to the output of $\Xi^{[h]}_{r, r'}$ as follow (see also Fig.~\ref{fig:main-inteface})
\begin{equation} \label{eq:def-Xi-r}
    \Xi^{[h]}_{r} := (\otimes_{i \in [h], j \in [2^{r-r'}] } \Gamma^{i, j}_{r', 1})     \circ \Xi^{[h]}_{r, r'}
\end{equation}
From Lemma~\ref{lem:cons-over-interface}, we know that $\Xi^{[h]}_{r, r'}$ operates on fewer than $ \theta p(m_r) m_r+ \theta' m_r h$ qubits, where $p(\cdot)$ is a polynomial.
Note that $\Gamma^{i, j}_{r', 1}$ operates on a fixed number of qubits $\theta''$ , only depending on $r'$ (see also Lemma~\ref{lem:succ-int-1}). Therefore, $h2^{r-r'}$ copies of $\Gamma^{i, j}_{r', 1}$ operate on total $\theta'' h 2^{r-r'} \leq \theta''  h m_r$ (using $m_r = 2^{r-r'} m_{r'}$) qubits. By redefining $\theta' \leftarrow \max(\theta', \theta'')$, we have that $ \Xi^{[h]}_{r}$ operates on less than $\theta p(m_r) m_r+ \theta' m_r h$ qubits.

\smallskip We will now show that $\cV' := \tilde{\cT}_{\Xi_{r}^{[h]}} \circ \cN \circ \left((\otimes_{i \in [h]} \cE^i_r) \otimes \cI[R] \right)$, which is a channel acting on $m_r h$ qubits and a reference system $R$, is local stochastic with parameter $\overline{\delta}^{\kappa}$ with respect to $m_r h$ qubit, with $\overline{\delta} = O(\delta)$ and some constant $\kappa > 0$. To do this, we will use Corollary~\ref{cor:succ-int-h-1} and Lemma~\ref{lem:cons-over-interface}.

\smallskip From Lemma~\ref{lem:cons-over-interface} and Eq.~\eqref{eq:def-Xi-r}, we have 
\begin{equation} \label{eq:V-prime-1}
   \cV' = \sum_{F} \sum_{\omega \in \Omega_F} \pr(F, \omega)(\otimes_{i \in [h], j \in [2^{r-r'}]} \tilde{\cT}_{\Gamma^{i, j}_{r', 1}}) \circ \cT_{F, w},
\end{equation}
where $\cT_{F, \omega}$ is from Eq.~\eqref{eq:const-over-interface} and $\pr(F) = \sum_{\omega \in \Omega_F} \pr(F, \omega)$ satisfies Eq.~\eqref{eq:cons-err-rate}.

\begin{figure}[t]
\centering
\begin{tikzpicture}[
  wire/.style={draw, thick},
  box/.style={draw, minimum width=2.2cm, minimum height=3.0cm, align=center}, 
  gbox/.style={draw, minimum width=2.2cm, minimum height=1.35cm, align=center} 
]

\def\yTop{1.2}
\def\yBot{-1.2}

\def\gH{0.42} 

\def\xIn{0.0}
\def\xXiL{2.4}
\def\xXiR{4.6}
\def\xMid{5.8}

\def\xGamC{7.2}
\def\xGamW{2.2}    
\def\xGamL{6.1}    
\def\xGamR{8.3}    

\def\xPhys{9.6}

\node[box] (Xi) at (3.5,0) {$\Xi^{[h]}_{r,r'}$};

\draw[wire] (\xIn,\yTop) -- (\xXiL,\yTop);
\draw[wire] (\xIn,\yBot) -- (\xXiL,\yBot);
\node at (1.2,0) {$\vdots$};

\draw[decorate,decoration={brace,mirror,amplitude=5pt}]
  (-0.3,\yTop+0.15) -- (-0.3,\yBot-0.15)
  node[midway,xshift=-0.6cm] {$h$};

\draw[wire] (1.1,\yTop-0.12) -- (1.22,\yTop+0.12);
\node[above] at (1.16,\yTop+0.12) {$r$};
\draw[wire] (1.1,\yBot-0.12) -- (1.22,\yBot+0.12);
\node[below] at (1.16,\yBot-0.12) {$r$};

\draw[wire] (\xXiR,\yTop) -- (\xMid,\yTop);
\draw[wire] (\xXiR,\yBot) -- (\xMid,\yBot);
\node at (5.3,0) {$\vdots$};
\node[right] at (4.8,0.3) {{$2^{\,r-r'}h$}};

\draw[wire] (5.0,\yTop-0.12) -- (5.12,\yTop+0.12);
\node[above] at (5.06,\yTop+0.12) {$r'$};

\draw[wire] (5.0,\yBot-0.12) -- (5.12,\yBot+0.12);
\node[below] at (5.06,\yBot-0.12) {$r'$};

\node[gbox] (Gtop) at (\xGamC,\yTop) {$\Gamma_{r',1}$};
\node[gbox] (Gbot) at (\xGamC,\yBot) {$\Gamma_{r',1}$};
\node at (\xGamC,0) {$\vdots$};

\draw[wire] (\xMid,\yTop) -- (\xGamL,\yTop);
\draw[wire] (\xMid,\yBot) -- (\xGamL,\yBot);

%
\draw[wire] (\xGamR,\yTop+\gH) -- (\xGamR,\yTop-\gH); 
\draw[wire] (\xGamR,\yTop+\gH) -- (\xPhys,\yTop+\gH);
\draw[wire] (\xGamR,\yTop-\gH) -- (\xPhys,\yTop-\gH);
\node at (8.95,\yTop) {$\vdots$};

\draw[decorate,decoration={brace,amplitude=5pt}]
  (\xPhys+0.2,\yTop+\gH+0.15) -- (\xPhys+0.2,\yTop-\gH-0.15)
  node[midway,xshift=0.6cm] {$m_{r'}$};

\draw[wire] (\xGamR,\yBot+\gH) -- (\xGamR,\yBot-\gH); 
\draw[wire] (\xGamR,\yBot+\gH) -- (\xPhys,\yBot+\gH);
\draw[wire] (\xGamR,\yBot-\gH) -- (\xPhys,\yBot-\gH);
\node at (8.95,\yBot) {$\vdots$};

\draw[decorate,decoration={brace,amplitude=5pt}]
  (\xPhys+0.2,\yBot+\gH+0.15) -- (\xPhys+0.2,\yBot-\gH-0.15)
  node[midway,xshift=0.6cm] {$m_{r'}$};

\end{tikzpicture}
\caption{Construction of the decoding interface $\Xi_r^{[h]}$. For a
fixed encoding level $r'$, the effective interface
$\Xi_{r,r'}^{[h]}$ first maps the $h$ input blocks of $\mathcal C_r$ to $2^{r-r'}h$ blocks of $\mathcal C_{r'}$. The fixed-level interface $\Gamma_{r',1}$ is then applied in parallel to every output block, yielding $m_rh$ bare physical qubits.}
\label{fig:main-inteface}
\end{figure}

\medskip We have 
\begin{align}
    & (\otimes_{i \in [h], j \in [2^{r-r'}]} \tilde{\cT}_{\Gamma^{i, j}_{r', 1}}) \circ \cT_{F, w} \nonumber  \\
    & \quad = \Big((\otimes_{(i, j) \not\in F} \tilde{\cT}_{\Gamma^{i, j}_{r', 1}}) \otimes  \overline{\cZ}_{F, \omega} \otimes  \cI[R] \Big) \circ  \cN_{F, \omega} \circ \Big( (\otimes_{(i, j) \not\in F} \cE^{i, j}_{r'}) \otimes  \cE_{F, \omega} \otimes  \cI[R] \Big) \nonumber \\
    & = \sum_{F' \subseteq \{(i, j)| i \in [h], j \in [2^{r-r'}] \} \setminus F}  \pr(F' \mid F) \left((\otimes_{(i, j) \not\in F \cup F'} \cI^{i, j}_{m_{r'}}) \otimes  \cZ_{F'} \otimes \overline{\cZ}_{F, \omega} \otimes  \cI[R] \right)  \nonumber \\
    & \quad \quad \quad \quad \quad \quad \quad \quad \quad  \quad \quad \quad \quad \circ {\cV}^{(1)}_{F \cup F', \omega} \circ \Big( (\otimes_{(i, j) \not\in F \cup F'} \cI^{i, j}_{m_{r'}}) \otimes \cE_{F'} \otimes \cE_{F, \omega} \otimes  \cI[R] \Big), \label{eq:gamma-r'-1-int} 
\end{align}
where the first equality uses Eq.~\eqref{eq:const-over-interface}, and   $\overline{\cZ}_{F, \omega} := (\otimes_{(i, j) \in F} \tilde{\cT}_{\Gamma^{i, j}_{r', 1}}) \circ \cZ_{F, \omega}$, with $\cZ_{F, \omega}$ being the arbitrary channel in Eq.~\eqref{eq:const-over-interface}.
The second equality follows from Corollary~\ref{cor:succ-int-h-1}, we push $\overline{\cZ}_{F, \omega}$ to the left and $\cE_{F, \omega}$ to the right and let $F$ be part of the reference; our $r'$ will be the $r$ from the corollary (note that assumption $r'=O(1)$ is satisfied in the corollary). Moreover, $\cE_{F'}: \bL(\mathbb{C}^2)^{\otimes m_{r'} |F'|}) \to (\mathbb{C}^2)^{\otimes n_{r'} |F'|})$, $\cE_{F'} := \otimes_{(i, j) \in F'} \cE^{i, j}_{r'}$ and $\cZ_{F'}: (\mathbb{C}^2)^{\otimes n_{r'} |F'|}) \to (\mathbb{C}^2)^{\otimes m_{r'} |F'|})$ is an arbitrary channel, and  
\begin{equation} \label{eq:F-Fprime}
\pr(F' \mid F) = (1- \tau_{r'})^{2^{r-r'}h - |F| - |F'|} \: \tau_{r'}^{|F'|},
\end{equation}
 where $\tau_{r'} \leq \lambda_{r'} (c \delta)^{c' m_{r'}}$, with constants $c, c' > 0$ and $\lambda_{r'} > 0$ only depending on $r'$ as in Corollary~\ref{cor:succ-int-h-1}. $\cV^{(1)}_{F \cup F', \omega}$ is a local stochastic channel with parameter $\delta' = \lambda'_{r'} \delta$, with respect to $m_{r'}$ qubit systems corresponding to labels $(i, j) \not\in F \cup F'$, where $\lambda'_{r'}$ is another parameter only depending on $r'$.

\smallskip We take $r'$ such that $c' m_{r'} \geq 1$, therefore, we have 
\begin{equation} \label{eq:ub-tau-rprime}
   \tau_{r'} \leq  b \delta, \text{ where } b :=  \lambda_{r'} c.
\end{equation}
For any $T \subseteq \{(i, j)| i \in [h], j \in [2^{r-r'}] \} \setminus F$, from Eq.~\eqref{eq:F-Fprime} and Eq.~\eqref{eq:ub-tau-rprime}, it is easy to see that
\begin{align}
\pr(T \subseteq  F' \mid F) 
& \leq   (b \delta)^{|T|}. \label{eq:F'-T}
\end{align}
At this point, we have managed to remove all the perfect encoders and corresponding noisy decoding interfaces except on blocks, labelled by $F$ and $F'$, where error happened. The total map (our $\cV'$) we arrived at is therefore a locally stochastic channel sandwiched by some perfect encoders on the input and arbitrary errors on the output confined to the erroneous blocks (see Eq.~\ref{eq:gamma-r'-1-int}). 

In the remaining part of the proof, we first obtain a form for $\cV'$ as a probabilistic sum over error maps acting on blocks. Then we convert this to error maps acting on the qubits. Finally, we will analyse the probability distribution to show that the overall channel is locally stochastic.

For the block labelled by $(i, j)$, let $B(i, j)$ denote the corresponding quantum system, containing $m_{r'}$ qubits. For $F \subseteq \{(i, j) \mid i \in [h],  j\in [2^{r-r'}]  \}$, let $B(F) : = \cup_{(i, j) \in F} B(i, j)$ be $m_{r'} |F|$ qubit system corresponding to blocks in $F$. We denote the quantum system corresponding to all the blocks by  $B := \cup_{i \in [h], j\in [2^{r-r'}]} B(i, j)$.
For any set $F$, we consider another quantum system $\overline{B}(F)$ that contains at least $m_{r'} |F|$ qubits, i.e., the size of $\overline{B}(F)$ is at least the size of $B(F)$.

\smallskip We define 
\begin{multline} \label{eq:v-2}
    {\cV}^{(2)}_{F \cup F', \omega} := \left((\otimes_{(i, j) \not\in (F \cup F')} \cI^{i, j}_{m_{r'}}) \otimes  \cZ_{F \cup F'} \otimes  \cI[R] \right) \circ  {\cV}^{(1)}_{F \cup F', \omega} \\
    \circ  \Big( (\otimes_{(i, j) \not\in (F \cup F')} \cI^{i, j}_{m_{r'}}) \otimes \cE_{F \cup F'} \otimes  \cI[R] \Big),
\end{multline}
where $\cE_{F \cup F'} = \cE_{F'} \otimes \cE_{F, \omega} $ and $\cZ_{F \cup F'} = \cZ_{F'} \otimes \overline{\cZ}_{F, \omega}$.

\smallskip Let $(F \cup F')^C := \{(i, j) \mid i \in [h],  j\in [2^{r-r'}]  \} \setminus F \cup F'$ be the complement of $F$. Let the output systems of $\cE_{F \cup F'}$ denoted by $\overline{B}(F \cup F')$.

\smallskip Note that ${\cV}^{(1)}_{F \cup F', \omega}$ acts on systems $B\left((F \cup F')^C\right) $, $\overline{B}(F \cup F')$, and the reference $R$. Using the local stochastic property of ${\cV}^{(1)}_{F \cup F', \omega}$ with respect to the subsystem $B\left((F \cup F')^C\right) $ and treating systems $\overline{B}(F \cup F')$ also as a part of the reference, we have 
\begin{equation} \label{eq:v-1}
    {\cV}^{(1)}_{F \cup F', \omega} = \sum_{A \subseteq B\left((F \cup F')^C\right)} \pr(A) \: \cI[B\left((F \cup F')^C\right) \setminus A] \otimes  \overline{\cN}_\omega [A, \overline{B}(F \cup F'), R].
\end{equation}
Moreover, for all $T \subseteq B(F^C)$, we have
\begin{equation} \label{eq:local-stoc-v-1}
    \pr(T \subseteq A) \leq (\lambda'_r \delta)^{|T|}.
\end{equation}
From Eq.~\eqref{eq:v-2} and Eq.~\eqref{eq:v-1}, we get
\begin{equation} \label{eq:v-3}
    {\cV}^{(2)}_{F \cup F', \omega} = \sum_{A \subseteq B\left((F \cup F')^C\right)} \pr(A) \: \cI[B\left((F \cup F')^C\right) \setminus A] \otimes  \cN_\omega[A, B(F \cup F'), R],
\end{equation}
where $\cN_\omega[A, B(F \cup F'), R] = (\cI[A] \otimes \cZ_{F\cup F'}) \circ \overline{\cN}_\omega \circ (\cI[A] \otimes \cE_{F \cup F'})$.

\smallskip We now have (in the following, we write $\sum_F$ in place of $\sum_{F \subseteq \{(i, j) \mid i \in [h], j \in [2^{r-r'}]  \}}$ for brevity)
\begin{align}
  & \cV' = \sum_{F} \sum_{\omega \in \Omega_F} \sum_{F' \subseteq F^C} \pr(F, \omega) \pr(F' \mid F) \: {\cV}^{(2)}_{F \cup F', \omega} \nonumber \nonumber \\
  & =  \sum_{F} \sum_{\omega \in \Omega_F} \sum_{F' \subseteq F^C} \sum_{A \subseteq B\left((F \cup F')^C\right)} \pr(F, \omega) \pr(F' \mid F)  \pr(A) \: \Big( \cI[B\left((F \cup F')^C\right) \setminus A] \nonumber \\
  & \quad \quad \quad \quad \quad \quad \quad \quad \quad \quad \quad \quad \quad \quad \quad \quad \quad \quad \quad \quad \quad \quad \quad \quad \otimes  \cN_\omega[A, B(F \cup F'), R]  \Big)\nonumber \\
  & = \sum_{F} \sum_{F' \subseteq F^C} \sum_{A \subseteq B\left((F \cup F')^C\right)} \pr(F) \pr(F' \mid F)  \pr(A) \: \Big(\cI[B\left((F \cup F')^C\right) \setminus A]  \nonumber \\ 
  & \quad \quad \quad \quad \quad \quad \quad \quad \quad \quad \quad \quad \quad \quad \quad \quad \quad \quad \quad \quad \quad \quad \quad \quad \otimes \cN[A, B(F \cup F'), R] \Big), \label{eq:final-local-stoc}
\end{align}
where the first equality follows from Eq.~\eqref{eq:gamma-r'-1-int} and Eq.~\eqref{eq:v-2}. The second equality follows from Eq.~\eqref{eq:v-3}. In the last equality, we have used
\begin{equation}
  \pr(F) = \sum_{\omega \in \Omega_F} \pr(F, \omega), \text{ and } \cN := \frac{1}{\pr(F)} \sum_{\omega \in \Omega_F} \pr(F, \omega) \cN_\omega   
\end{equation}
We will now convert the distributions that are over blocks in Eq.~\eqref{eq:final-local-stoc}, i.e., $\pr(F)$ and $\pr(F' \mid F)$ to distributions over qubit systems that are subset of $B$. For $\pr(F), F \subseteq \{(i, j) \mid i \in [h], j \in [2^{r-r'}]\}$, we define a corresponding probability distribution on the quantum system $B$ as follows: 

\smallskip For any $G \subseteq B$, $\pr(G) = \pr(F)$ if $G = B(F)$ for some $ F \subseteq \{(i, j) \mid i \in [h], j \in [2^{r-r'}]\}$, else $\pr(G) = 0$. 

Similarly, for $\pr(F' \mid F), F' \subseteq F^C$, we define the corresponding conditional probability distribution on the set $B \setminus G$, where $G = B(F)$ for some $F \subseteq \{(i, j) \mid i \in [h], j \in [2^{r-r'}]\}$ as follows:

For any $G' \subseteq B \setminus G$,  $\pr(G' \mid G) = \pr(F' \mid F)$ if $G' = B(F')$ for some $F' \subseteq F^C$ or else $\pr(G' \mid G) = 0$. 

\medskip We rewrite Eq.~\eqref{eq:final-local-stoc} as
\begin{equation} \label{eq:exp-loc-stoc-vprime}
    \cV' = \sum_{G \subseteq B} \: \sum_{G' \subseteq B \setminus G} \: \sum_{A \subseteq B \setminus (A \cup A')} \: \pr(G) \pr(G') \pr(A) \: \left(\cI[B \setminus (G \cup G' \cup A)] \otimes  \cN[A, G, G', R] \right).
\end{equation}
From Eq.~\eqref{eq:exp-loc-stoc-vprime}, we get the expression of $\cV'$ as a convex combination of arbitrary channels acting on subsets of $m_r h$ qubit system $B$ as in Eq.~\eqref{eq:W-exp-1}. In order to show that $\cV'$ is locally stochastic, it remains to analyze the probability distribution. More precisely, it remains to show that $P(T \subseteq G \cup G' \cup A) \leq (\kappa_1 \delta)^{\kappa_2 |T|}$ for some constants $\kappa_1, \kappa_2 > 0$.

\smallskip For any $T \subseteq B$, we define the following set
\begin{equation}
    \overline{T} := \{ (i, j) \mid i \in [h], j \in [2^{r-r'}], \text{ and } T \cap B(i, j) \neq \emptyset \}. 
\end{equation}
Note that   $T \subseteq G$ if and only $\overline{T} \subseteq F$ and $|\overline{T}| \geq \frac{|T|}{m_{r'}}$. By choosing a $\ovr$ in Lemma~\ref{lem:cons-over-interface}, such that $\overline{\delta} \leq b \delta$, we have
\begin{equation}
    \pr(T \subseteq G) = \pr(\overline{T}  \subseteq F) \leq (b \delta)^{|\overline{T}|} = (b \delta)^{\frac{|T|}{m_{r'}}},
\end{equation}
where the inequality uses Eq.~\eqref{eq:cons-err-rate}.
Similarly, using Eq.~\eqref{eq:F'-T}, we can show that for any $T \subseteq B \setminus G$, we have
\begin{equation}
    \pr(T \subseteq G') = (b \delta)^{\frac{|T|}{m_{r'}}}.
\end{equation}
Therefore, for any $T \subseteq B$
\begin{align} \label{eq:T-A-A}
    \pr(T \subseteq G \cup G') &= \sum_{T_1 \subseteq T} \pr(T_1 \subseteq G) \pr(T \setminus T_1 \subseteq G') \nonumber\\
    &   \leq (2b\delta)^{\frac{|T|}{m_{r'}}}.
\end{align}
Similarly, from Eq.~\eqref{eq:local-stoc-v-1} and Eq.~\eqref{eq:T-A-A}, and using $\delta' = \lambda'_{r'}\delta$, it follows that 
\begin{equation} \label{eq:loc-stoc-v-prime}
    \pr(T \subseteq G \cup G' \cup A) \leq (\kappa_1 \delta)^{\kappa_2|T|},
\end{equation}
where $\kappa_1 = 2b + \lambda'_{r'}$, and $\kappa_2 := \frac{1}{m_{r'}}$. Note that since $r'$ is fixed, $\kappa_1, \kappa_2$ are constants not depending on $r$. Finally, from Eq.~\eqref{eq:exp-loc-stoc-vprime} and Eq.~\eqref{eq:loc-stoc-v-prime}, it follows that $\cV'$ is a local stochastic channel with parameter $(\kappa_1 \delta)^{\kappa_2}$ with respect to system $B$.   
\end{proof}

\subsection{Construction of the interface for Lemma~\ref{lem:cons-over-interface}} \label{sec:xi-con}  
We will first construct the decoding interface $\Xi^{[h]}_{r, r-1}$, which maps $h$ copies of $\cC_r$ into $2 h$ copies of $\cC_{r-1}$. We will then obtain $\Xi^{[h]}_{r, r'}$ by  decreasing the encoding level one by one from $r$ to $r'$. We construct the interface $\Xi^{[h]}_{r, r-1}$ in several steps, using the interface $\Gamma_{r, r-1}$ from Lemma~\ref{lem:succ-int}. During each step, we apply $\Gamma_{r, r-1}$ on a fraction of the total blocks of $\cC_r$, mapping these blocks to twice as many blocks of $\cC_{r-1}$. In parallel, we apply the appropriate error correction circuits on the remaining blocks of $\cC_r$ and on the blocks of $\cC_{r-1}$ created in earlier steps.

\smallskip We consider the interface $\Gamma_{r, r-1}$ from Lemma~\ref{lem:succ-int}. We have that $\Gamma_{r, r-1}$ operates on $\theta m_r p_1(m_r)$ qubits. 

\smallskip As the total number of blocks is increasing as we decrease the encoding,  we will denote by $h^{(r'')}$ the total number of blocks of $\cC_{r''}$ at layer $r' \leq r'' \leq r$ that correspond to the input of $\Xi^{[h^{(r'')}]}_{r'', r''-1}$. We start with $h^{(r)} = h$ blocks of $\cC_r$.
Since $\Xi^{[h^{(r'')}]}_{r'', r''-1}$ maps $h^{(r'')}$ blocks of $\cC_{r''}$ into $2 h^{(r'')}$ blocks of $\cC_{r''-1}$, we have $h^{(r'')} = 2^{r-r''} h$.

\medskip We construct the circuit $\Xi^{[h^{(r)}]}_{r,r-1}$ in several steps. In the first step, we apply the following two operations in parallel on the $h$ blocks of $n_r$ qubits, labelled by $i \in [h^{(r)}]$
\begin{itemize}
    \item[(a)]  Consider $h_r < h^{(r)}$ such that
    \begin{equation} \label{eq:h_r-h}
     h_r = \left\lceil \frac{h^{(r)}}{\theta p_1(m_r)} \right\rceil.
    \end{equation}
 We apply $\Gamma_{r, r-1}$ on blocks of $\cC_r$ corresponding to $i = 1, \dots,  h_r$. We denote by $\Gamma^i_{r, r-1}$, the $i^{th}$ copy of $\Gamma_{r, r-1}$.    
    \item[(b)] On the remaining blocks $ i = h_r + 1, \dots, h^{(r)}$, we apply $s$ error correction steps, i.e. $\Phi_{\mathrm{EC}^{r,s}}$ from Eq.~\eqref{eq:s-cor-step}. Since the depth of $\Gamma_{r, r-1}$ is upper bounded by $d p_2(m_r)$ (from Lemma~\ref{lem:succ-int} and using $k = m_r$), we 
    have $s = d p_2(m_r)$. We denote by $\Phi^i_{\mathrm{EC}^{r,s}}$, the $i^{th}$ copy of $\Phi^i_{\mathrm{EC}^{r,s}}$.    
\end{itemize}
The output of $\Gamma_{r, r-1}$ corresponds to two blocks of $\cC_{r, r-1}$. We label the output blocks of $\Gamma^i_{r, r-1}$ by $(i, j_1)$, where $j_1 \in \{0, 1\}$. 
The quantum circuit corresponding to (a) and (b) is given by (here tensor product between two quantum circuits  represents their parallel realization)
\begin{equation}
    \Psi^{[h^{(r)}]}_{r, 1} := (\otimes_{i \in [h_r]} \:\Gamma^i_{r, r-1}) \bigotimes (\otimes_{i \in [h_r+1: h^{(r)}]} \:\Phi^i_{\mathrm{EC}^{r,s}})
\end{equation}
In the second step, we apply the following operations on the output systems of $\Psi_{r,1}^{[h^{(r)}]}$
\begin{itemize}
    \item[(a')] We apply $\Gamma_{r, r-1}$ on blocks of $\cC_r$ corresponding to $ i \in [h_r:2h_r]  = h_r + 1, \dots, 2h_r $.

    \item[(b')] We apply $\Phi_{\mathrm{EC}^{r-1, s}}$ on the blocks of $\cC_{r-1}$ corresponding to $(i, j_1), i \in [h_r], j_1 \in \{0, 1\}$, and $\Phi_{\mathrm{EC}^{r, s}}$ on the blocks of $\cC_r$ corresponding to  $i \in  [2h_r:h^{(r)}]$.  
\end{itemize}
 Then, the circuit corresponding to (a') and (b') is given by,
\begin{multline}
    \Psi^{[h^{(r)}]}_{r, 2} := (\otimes_{i \in [h_r], j_1 \in \{0, 1\}} \: \Phi^{i, j_1}_{\mathrm{EC}^{r-1, s}}) \bigotimes (\otimes_{i \in [h_r:2h_r]} \: \Gamma^i_{r, r-1}) \bigotimes (\otimes_{i \in [2h_r:h^{(r)}]} \: \Phi^i_{\mathrm{EC}^{r, s}}) 
\end{multline}
We continue in the above fashion until the interface $\Gamma^i_{r, r-1}$ has been applied on all the $h^{(r)}$ blocks of $\cC_r$. In particular, we keep applying  $\Psi^{[h^{(r)}]}_{r, l}$ for any $1 \leq l \leq \lceil \frac{h}{h_r} \rceil$, where $\Psi_{r, l}$ applies $\Phi^{i, j_1}_{\mathrm{EC}^{r-1, s}}$ on blocks of $\cC_{r-1}$ corresponding to $i \in [(l-1)h_r]$, $\Gamma^i_{r, r-1}$ on blocks of $\cC_r$ corresponding to $i \in [lh_r: (l+1) h_r]$, and $\Phi^{i, j_1}_{\mathrm{EC}^{r, s}}$ on blocks of $\cC_r$ corresponding to $[(l+1)h_r:h^{(r)}]$, where $s =  \poly(m_r)$, that is,
\begin{multline} \label{eq:psi-r-j}
    \Psi^{[h^{(r)}]}_{r, l} := (\otimes_{i \in [(l-1) h_r], j_1 \in \{0, 1\} } \: \Phi_{\mathrm{EC}^{r-1, s}})  \bigotimes (\otimes_{i \in [(l-1)h_r:lh_r]} \: \Gamma_{r, r-l}) \bigotimes (\otimes_{i \in [lh_r:h^{(r)}]} \: \Phi_{\mathrm{EC}^{r,s}}).
\end{multline}
We now define the interface $\Xi_{r, r-1}^{[h^{(r)}]}$ as follows (here the composition $\Phi_2 \circ \Phi_1$ means that $\Phi_1$ is applied first and $\Phi_2$ is applied to its output, assuming the output systems of $\Phi_1$ match the input systems of $\Phi_2$)
\begin{equation} \label{eq:xi-r-1}
    \Xi_{r, r-1}^{[h^{(r)}]} :=        \Psi^{[h^{(r)}]}_{r,\lceil \frac{h^{(r)}}{h_r} \rceil} \circ \cdots \circ \Psi^{[h^{(r)}]}_{r, l} \circ \cdots \circ \Psi^{[h^{(r)}]}_{r, 2} \circ \Psi^{[h^{(r)}]}_{r, 1}.
\end{equation}
The output of $\Xi_{r, r-1}^h$ corresponds to $h^{(r-1)} = 2h^{(r)}$ blocks of $\cC_{r-1}$ labelled by $(i, j_1),  i \in [h^{(r)}], j_1 \in \{0, 1\}$.
On these blocks, we apply $\Xi_{r-1, r-2}^{[h^{(r-1)}]}$, whose output corresponds to $h^{(r-2)} = 2^2 h^{(r)}$ blocks of $\cC_{r-2}$ labelled by $(i, j_1, j_2), \: i \in [h^{(r)}], \: j_1, j_2 \in \{0, 1\}$. 
Subsequently on $(i, j_1, j_2)$, $i \in [h^{(r)}], j_1, j_2 \in \{0, 1\}$, we apply $\Xi_{r-2, r-3}^{[h^{(r-2)}]}$. To define $\Xi^{[h]}_{r, r'}$, we continue doing this until we have $h^{(r')} = 2^{r-r'} h^{(r)}$ blocks of $\cC_{r'}$, that is,
\begin{equation} \label{eq:xi-r-r'}
  \Xi^{[h^{(r)}]}_{r, r'} = \Xi_{r'+1, r'}^{[h^{(r' + 1)}]} \: \circ \cdots \circ \Xi^{[h^{(r'')}]}_{r'', r''-1} \circ  \cdots \cdots \circ \: \Xi_{r-1, r-2}^{[h^{(r-1)}]}  \circ \Xi_{r, r-1}^{[h^{(r)}]}.
\end{equation}
Note that for any $r' < r'' \leq r$, the corresponding $\Xi^{[h^{(r'')}]}_{r'', r''-1}$  as,
\begin{equation} \label{eq:xi-r-prime2}
\Xi^{[h^{(r'')}]}_{r'', r''-1} =  \Psi^{[h^{(r'')}]}_{r'', \left\lceil \tfrac{h^{(r'')}}{h_{r''}} \right\rceil} \circ \cdots \circ \Psi_{r'', l} \circ \cdots \circ \Psi^{[h^{(r'')}]}_{r'', 2} \circ \Psi^{[h^{(r'')}]}_{r'', 1},
\end{equation}
where for any $\Psi_{r'', l}$, $l \leq \left\lceil \tfrac{2^{r -r''} h}{h_{r''}} \right\rceil$, the number of blocks on which $\Gamma_{r'', r''- 1}$ is applied in parallel is given by, 
\begin{equation} \label{eq:h_r-prime}
     h_{r''} = \left\lceil \frac{h^{(r'')}}{\theta p_1(m_{r''})} \right\rceil. 
\end{equation}
We note that the fraction of qubits on which $\Gamma_{r'', r''-1}$ is applied in parallel for $\Psi_{r'', l}$, that is, $\tfrac{h_{r''}}{h^{(r'')}} \propto \frac{1}{p_1(m_{r''})}$, increases with decreasing $r''$ (see also Fig.~\ref{fig:sequential-interface}). This implies that the depth of the circuit $\Xi^{[h^{(r'')}]}_{r'', r''-1}$, which is determined by $\tfrac{h^{(r'')}}{h_{r''}}$, decreases as $r''$ is reduced from $r$ to $r'$. Therefore, the logical information is decoded faster from an encoding layer $r''$ to the layer below it for smaller values of $r''$.

\smallskip For $\Xi_{r-y, r-y-1}^{[h^{(r-y)}]}$, where $y < r-r'$, its input corresponds to $h^{(r-y)} =2^{y} h^{(r)}$ blocks of $\cC_{r-y}$, which are labelled  by $(i, j_1, j_2, \dots, j_{y})$, $i \in [h^{(r)}], j_1, j_2, \dots, j_{y} \in \{0, 1\}$, and its output corresponds to $h^{(r-y-1)} = 2^{(y+1)} h$ blocks of $\cC_{r-y-1}$, labelled  by $(i, j_1, j_2, \dots, j_y, j_{y+1})$, $i \in [h^{(r)}], j_1, j_2, \dots, j_{y}, j_{y+1}  \in \{0, 1\}$.
\begin{remark} \label{rem:notation-cons-ovr}
To connect with the statement of Lemma~\ref{lem:cons-over-interface}, we  denote the output of $\Xi^{[h]}_{r, r'}$ by $ (i, j), i \in [h], j \in [2^{r-r'}]$,
where $(i, j)$ is the block corresponding to  $(i, j_1, j_2, \dots, j_{r-r'})$, with $j - 1$ being the decimal representation of binary string $j_1, j_2, \dots, j_{r-r'}$ ($j_1$ being the least significant bit), that is,
\begin{equation}
     j - 1  = \sum_{y = 1}^{r-r'} 2^{y-1} j_{y}.
\end{equation}
\end{remark}
\paragraph{Counting qubits in the interface $\Xi^{[h]}_{r, r'}$:}
In this paragraph, we will show that the circuit $\Xi^{[h]}_{r, r'}$ described above operates on fewer than $ \theta p(m_r) m_r+ \theta' m_r h$ qubits for some constants $\theta, \theta' > 0$ and a polynomial $p(\cdot)$.

\smallskip Recall that the circuit $\Xi^{[h]}_{r, r'}$
is constructed using the interface $\Gamma_{r'', r''-1}$ given in Lemma~\ref{lem:succ-int} for $r'' = r, r-1, \dots, r'+1$  and the error correction circuit $\Phi_{\mathrm{EC}^{r'', s}}$ where $r'' = r, r-1, \dots, r'$. The error correction circuit has a constant qubit overhead with respect to its input size. Although $\Gamma_{r'', r''-1}$ does not have a constant qubit overhead with respect to its input, since $\Gamma_{r'', r''-1}$ is applied on only a few blocks in parallel, the overall contribution of $\Gamma_{r'', r''-1}$'s to the overhead is negligible. This is explicitly shown in the following.

Consider $r''$ such that $r''= r, r-1, \dots, r'+1$ and the circuit $\Xi^{[h^{(r'')}]}_{r'', r''-1}$ according to Eq.~\eqref{eq:xi-r-prime2}-~\eqref{eq:h_r-prime}.
It suffices to show that $\Psi_{r'', l}$ from Eq.~\eqref{eq:xi-r-prime2} for any $l \leq \lceil \frac{h^{(r'')}}{h_{r''}} \rceil$ operates on $O(hm_r)$ qubits. The input of $\Psi_{r'', l}$ is given by the following sets,
\begin{itemize}
    \item[$(1)$]  A set of $2 (l-1) h_{r''}$ blocks, each containing $n_{r''-1}$ qubits. These blocks correspond to the set of $(l-1) h_{r''}$ blocks of size $n_{r''}$ on which $\Gamma_{r, r-1}$ has been applied during $\Psi_{r'', l'}, l' < l$.

    \item[$(2)$] A set of $h^{(r'')} - (l-1)  h_{r''}$ blocks, each containing $n_{r''}$ qubits.
\end{itemize}

\smallskip Moreover, during $\Psi^{[h^{(r'')}]}_{r'', l}$, the following operations  are applied in parallel
\begin{itemize}
    \item[$(a)$] $\Gamma_{r'', r''-1}$ is applied on the first $h_{r''}$ blocks of $n_{r''}$ qubits.

    \item[$(b)$] For $s = O(\mathrm{poly}(m_{r''}))$, error correction steps $\Phi_{\mathrm{EC}^{r''-1, s}}$  are applied on $2(l-1)  h_{r''}$ blocks of $n_{r'' - 1}$ qubits, and on the remaining $2^{r-r''} h - l  h_{r''}$ blocks of $n_{r''}$ qubits, error correction steps $\Phi_{\mathrm{EC}^{r'', s}}$ are applied.    
\end{itemize}
Error correction circuits $\Phi_{\mathrm{EC}^{(r''-1, s)}}$ and $\Phi_{\mathrm{EC}^{r'', s}}$ operate on $\theta_1 m_{r''-1}$ and $\theta_1 m_{r''}$ qubits for some constant $\theta_1 > 0$ (see also Def.~\ref{def:ft-err-cor}). Therefore, the total number of qubits corresponding to the blocks on which $\Phi_{\mathrm{EC}^{(r''-1, s)}}$ and $\Phi_{\mathrm{EC}^{r'', s}}$ is applied is given by,
\begin{align}
     \eta_1 &:= 2 (l-1) h_{r''}\: \beta    m_{r''- 1} + (h^{(r'')} - l  h_{r''})  \theta_1 m_{r''} \nonumber \\
     & = \big( (l- 1) h_{r''}  +  2^{r-r''} h   - l  h_{r''} \big) \theta_1 m_{r''}  \nonumber \\
     & \leq   \theta_1 h  2^{r-r''} m_{r''}  \nonumber \\
     & \leq \theta_1 h m_{r} \label{eq:theta1}
\end{align}
where in the second line, we have used $m_{r''} = 2 m_{r''-1}$, and in the last line $m_{r} = 2^{r-r''} m_{r''}$.

\smallskip From Lemma~\ref{lem:succ-int}, we have that $\Gamma_{r'', r''-1}$ operates on fewer than $\theta p_1(m_{r''}) m_{r''}$ qubits, the total number of qubits corresponding to $h_{r''}$ blocks on which $\Gamma_{r, r-1}$ is applied is given by,
\begin{align}
    \eta_2 &:= h_{r''} \theta p_1(m_{r''}) m_{r''} \nonumber \\
   &=   \left\lceil \frac{h^{(r'')}}{\theta p_1(m_{r''})} \right\rceil \theta p_1(m_{r''}) m_{r''} \nonumber \\
& \leq \left(\frac{2^{r-r''} h}{\theta p_1(m_{r''})}  + 1 \right) \theta p_1(m_{r''}) m_{r''} \nonumber \\
& \leq \theta m_r h + \theta p_1(m_{r}) m_{r}
   \label{eq:theta2}
\end{align}
where for the first equality, we have used Eq.~\eqref{eq:h_r-prime}, for the first inequality, we have used $h^{(r'')} = 2^{r-r''} h$, and the second inequality, we have used $m_r = 2^{r-r''} m_{r''}$.  Finally, from Eq.~\eqref{eq:theta1} and Eq.~\eqref{eq:theta2}, the quantum circuit $\Psi_{r'', l}$ operates on fewer than $ \theta p_1(m_{r}) m_{r} + \theta' h m_{r}$ qubits, where $\theta' = \theta + \theta_1$.

\paragraph{Effective interface:}
Using the fact that the interface $\Xi^{[h^{(r'')}]}_{r'', r''-1}$ is a tensor product of quantum circuits applied across the $h^{(r'')}$ blocks  of $\cC_{r''}$, we can define an effective interface circuit for each block $i \in [h^{(r'')}]$. More precisely, the effective interface circuit corresponding to a block $i \in [h^{(r'')}]$ is the circuit in the tensor product that acts on the $i$th block. Although the effective interface circuit is different for each block, they have the same structure which corresponds to sandwiching the interface circuit $\Gamma_{r'', r''-1}$ by an error correction circuit $\Phi_{\mathrm{EC}^{r'',s'_1}}$ from the right and  the error correction circuit $\Phi_{\mathrm{EC}^{r'' - 1,s'_2}}$ from the left. Here $s'_1, s'_2$ can vary depending on the position of block $i \in [h^{(r'')}]$ but they are in $O(\poly(m_{r''}))$. We then obtain the effective interface for $\Xi^{[h^{(r'')}]}_{r, r'}$ by appropriately composing the effective interface circuits corresponding to  $ \displaystyle \Xi^{[h^{(r'')}]}_{r'', r''-1}$, $r < r'' \leq r$. The effective interface will be very helpful in error analysis of the interface $\Xi^{[h^{(r)}]}_{r, r'}$ in Section~\ref{sec:err-analysis} as it allows to analyze each of the blocks $i \in [h]$ separately. We explain this now in more detail.

\smallskip Consider the interface $ \displaystyle \Xi^{[h^{(r'')}]}_{r'',r''-1}$ from Eq.~\eqref{eq:h_r-prime} applied on $h^{(r'')}$ blocks of $\cC_{r''}$. By construction, $\displaystyle \Xi^{[h^{(r'')}]}_{r'',r''-1}$ is a composition of quantum circuits $\displaystyle \Psi^{[h^{(r'')}]}_{r'', l}$, $l = 1, 2, \dots, \lceil \frac{h^{(r'')}}{h_r} \rceil$ as given in Eq.~\eqref{eq:xi-r-prime2}. Moreover, the quantum circuit $\displaystyle \Psi^{[h^{(r'')}]}_{r'', l}$ is a tensor product across $h$ blocks, applying either $\Phi_{\mathrm{EC}^{r'', s'_1}}$ or $\Gamma_{r'', r''-1}$ or on any block of $\cC_{r''}$, and $\Phi_{\mathrm{EC}^{r'' - 1, s'_2}}$ on any block of $\cC_{r''-1}$, where $s'_1, s'_2$ may vary as a function of $i$, however, they are in $O(\poly(m_r))$ for any $i \in [h^{(r'')}]$. In other words, for any $i \in [h^{(r'')}]$, interface $\Gamma_{r'', r''-1}$ is applied on $i^{th}$ block of $\cC_{r''}$ during $\displaystyle \Psi^{[h^{(r'')}]}_{r'', l_i}$ for some $l_i  \leq \lceil \frac{h^{(r'')}}{h_r} \rceil$. For all $l < l_i$, $\Phi^i_{\mathrm{EC}^{r'', s'_1}}$ is applied on it and for all $l > l_i$, the error correction circuit $\Phi^{i, j}_{\mathrm{EC}^{r'' - 1, s'_2}}$ is applied on all the blocks $(i, j_1), j_1 \in \{0, 1\}$ of $\cC_{r''-1}$.

\smallskip Therefore, for any $i \in [h^{(r'')}]$, the following effective circuit acts on the $i^{th}$ block during $\Xi^{[h^{(r'')}]}_{r'',r''-1}$
\begin{equation} \label{eq:eff-circ-1}
\overline{\Gamma}_{r'', r'' - 1} := (\otimes_{j_1 \in \{0, 1\}} \Phi^{i, j_1}_{\mathrm{EC}^{r''-1,s'_2}}) \circ \Gamma^i_{r'', r''-1} \circ \Phi^{i}_{\mathrm{EC}^{r'',s'_1}},
\end{equation}
where $s'_1, s'_2 = O(\poly(m_{r''}))$. Note that the interface circuit $\Gamma^i_{r'', r''-1}$ itself applies a $s_1, s_2 = O(\poly(m_{r''}))$ error correction steps on the input and output, respectively (see also Fig.~\ref{fig:tele-circ}). Therefore, we  absorb $s'_1, s'_2$ error correction circuits in $\Gamma_{r'', r'' - 1}$, and simply replace the effective circuit $\overline{\Gamma}_{r'', r'' - 1}$ by $\Gamma_{r'', r'' - 1}$.

\smallskip By composing the circuits $\Gamma_{r'', r'' - 1} \equiv \overline{\Gamma}_{r'', r'' - 1}$, $r' < r'' \leq r$, we have that the following effective circuit is applied on a block $i \in [h^{(r)}]$ during the execution of $\Xi^{[h^{(r)}]}_{r, r'}$
\begin{multline} \label{eq:eff-circ}
\overline{\Xi}^{i}_{r, r'} := (\otimes_{j_1, \dots, j_{r-r'-1} \in \{0, 1\}} \: \Gamma^{i, j_1, \dots, j_{r-r'-1}}_{r' + 1, r'}) \circ \cdots \circ (\otimes_{j_1 \in \{0, 1\}} \: \Gamma^{i, j_1}_{r-1, r-2})  \circ \Gamma^i_{r, r-1}.
\end{multline}
Finally, we can write the interface circuit $\Xi^{[h^{(r)}]}_{r, r'}$ as a tensor product of the effective interface circuits
\begin{equation} \label{eq:eff-interface}
  \Xi^{[h^{(r)}]}_{r, r'} = \bigotimes_{i \in [h^{(r)}]} \: \overline{\Xi}^{i}_{r, r'}.
\end{equation}

\paragraph{Proof roadmap for Theorem~\ref{thm:ft-cons-int-main}:}
We briefly summarize how the results in this section and Section~5 combine to prove Theorem~\ref{thm:ft-cons-int-main}.

First, Lemmas~\ref{lem:succ-int} and~\ref{lem:succ-int-1} provide partial decoding interfaces. In particular, $\Gamma_{r,r-1}$ from  Lemma~\ref{lem:succ-int} maps one block of $\cC_r$ into two blocks of $\cC_{r-1}$, and $\Gamma_{r_0,1}$ from Lemma~\ref{lem:succ-int-1} maps a block at a fixed encoding level $r_0=O(1)$ to bare physical qubits. Corollaries~\ref{cor:succ-int-h} and~\ref{cor:succ-int-h-1} extend these interfaces to multiple blocks of $\cC_r$ and $\cC_{r_0}$, respectively.

\smallskip In Section~\ref{sec:xi-con}, using the interface $\Gamma_{r, r-1}$ and error-correction steps, an interface $\Xi^{[h]}_{r, r-1}$ on $h$-code blocks of $\cC_r$ is given. Then, the interfaces $\Xi^{[2^{r- \hat r}h]}_{\hat r, \hat r-1}, \: \hat r = r, r-1, \dots, r'$  are composed together to construct the interface $\Xi^{[h]}_{r,r'}$. Using the tensor product structure of $\Xi^{[h]}_{\hat r, \hat r-1}$ across $h$-code blocks, an effective interface $ \overline{\Xi}^{i}_{r, r'}$ for each code block $i = 1, 2, \dots, h$ is identified.

\smallskip In Section~\ref{sec:const-ovr-interface}, the fault-tolerance of the interface $\Xi^{[h]}_{r,r'}$ is stated in Lemma~\ref{lem:cons-over-interface}. Moreover, the interface $\Xi^{[h]}_{r}$ for Theorem~\ref{thm:ft-cons-int-main} is obtained by combining $\Xi^{[h]}_{r,r_0}$ and $\Gamma_{r_0,1}$, where $r_0=O(1)$ is a fixed encoding level independent of $r$. Assuming that Lemma~\ref{lem:cons-over-interface} holds, we provide a proof of Theorem~\ref{thm:ft-cons-int-main}. The qubit overhead bound required for $\Xi^{[h]}_{r,r'}$ in Lemma~\ref{lem:cons-over-interface} is given in Sec~\ref{sec:xi-con}.

Therefore, it only remains to prove Lemma~\ref{lem:cons-over-interface}, which is done in Section~\ref{sec:err-analysis}. By the tensor-product decomposition above, it is enough to analyze one effective interface $\overline{\Xi}^i_{r,r_0}$. The main result of this single-effective-interface analysis is Lemma~\ref{lem:err-eff-int}, which is proved in the following three steps:
\begin{enumerate}
    \item  Lemma~\ref{lem:choice-ovr} shows that the failure probability of a partial interface decreases doubly exponentially with its encoding level. This compensates for the fact that a failure at a high level can corrupt many descendant output blocks.

    \item In Section~\ref{sec:block-err-path}, Corollary~\ref{cor:succ-int-h} is applied successively to the partial interfaces appearing in $\overline{\Xi}^i_{r,r_0}$. The resulting failure history is recorded by a block error pattern (see Section~\ref{sec:inf-err-analysis} for an intuitive explanation of block error pattern)
\[
(F_0,F_1,\ldots,F_{r-r_0-1}).
\]
Definitions~\ref{def:block-err-path}-\ref{def:partition} associate this failure history with a set $F$ of erroneous output blocks. Conditioned on $F$, the blocks outside $F$ are affected only by controlled errors, i.e., errors with low stabilizer-reduced weight, while the blocks in $F$ may be affected by arbitrary errors. At this point, proving Lemma~\ref{lem:err-eff-int} reduces to showing that the inclusion probability $\Pr(T\subseteq F)$ decreases exponentially with the number of output blocks in $T$.

\item Section~\ref{sec:bloc-path-tree} represents the hierarchy of partial interfaces by a perfect binary tree. A failure at a node corrupts all of its descendant leaves, which accounts for the correlations among the erroneous output blocks. Lemma~\ref{lem:rec} recursively decomposes the event that all the nodes in a given set are corrupted. Using the node weight introduced in Definition~\ref{def:nweight}, Lemma~\ref{lem:final-upper-bound} iterates this decomposition and proves the required inclusion-probability bound. This completes the proof of Lemma~\ref{lem:err-eff-int}.
\end{enumerate}
Finally, applying Lemma~\ref{lem:err-eff-int} to each effective interfaces $\overline{\Xi}^i_{r,r_0}, i \in [h]$ then proves Lemma~\ref{lem:cons-over-interface}, hence, completing the proof of Theorem~\ref{thm:ft-cons-int-main}.

\section{Error analysis of decoding interfaces} \label{sec:err-analysis}
In this section, we show that Lemma~\ref{lem:cons-over-interface} holds for the quantum circuit $\Xi^{[h]}_{r, r'}$ described in Section~\ref{sec:xi-con}. Recall that the number of qubits in $\Xi^{[h]}_{r, r'}$ has already been shown to be upper bounded as required in Lemma~\ref{lem:cons-over-interface} in Section~\ref{sec:xi-con}. Therefore, we only need to show that Eq.~\eqref{eq:const-over-interface-1}-\eqref{eq:cons-err-rate} hold.

To do so, we provide an error analysis of the effective interface circuit $\overline{\Xi}^i_{r, r'}$ from Eq.~\eqref{eq:eff-circ} under circuit-level stochastic noise. Using our error analysis of $\overline{\Xi}^i_{r, r'}$ and tensor product structure of $\Xi^{[h]}_{r, r'}$ in terms of $\overline{\Xi}^i_{r, r'}$ given in Eq.~\eqref{eq:eff-interface}, it follows that Lemma~\ref{lem:cons-over-interface} holds for $\Xi^{[h]}_{r, r'}$.

The main statement regarding the error analysis of $\overline{\Xi}^i_{r, r'}$ is given as Lemma~\ref{lem:err-eff-int} in Section \ref{subsec:error-analysis-effect-interf}, and its proof is done in Sections~\ref{sec:choice-ovr}, \ref{sec:block-err-path} and \ref{sec:bloc-path-tree}.

\subsection{Error analysis of the effective interface}
 \label{subsec:error-analysis-effect-interf}
We recall that $\overline{\Xi}^i_{r, r'}$ outputs $2^{r-r'}$ blocks of $\cC_{r'}$, denoted by pairs $(i, j), j \in [2^{r-r'}]$. For any $(i,j), j \in [2^{r-r'}]$, we will equivalently denote it by  $(i, j_1, j_2, \dots, j_{r-r'})$, where $j_1, j_2, \dots, j_{r-r'} \in \{0, 1\}$ as in Remark~\ref{rem:notation-cons-ovr}.

\smallskip In Lemma~\ref{lem:err-eff-int} we show that the noisy version of the effective interface $\overline{\Xi}^i_{r, r'}$ maps a code block of $\cC_r$ with a low stabilizer-reduced weight error to $2^{r-r'}$ code blocks of $\cC_{r'}$. Some of the output blocks are erroneous (i.e. arbitrary error channels are applied on them); the remaining blocks are affected controlled errors, i.e., low stabilizer-reduced weight error. Importantly, the probability that a subset $T \subseteq [2^{r-r'}]$ of blocks is erroneous decreases exponentially with its size $|T|$. Due to  Eq.~\eqref{eq:eff-interface}, Lemma~\ref{lem:err-eff-int} is a special case of Lemma~\ref{lem:cons-over-interface} for $h = 1$. For $h > 1$, Lemma~\ref{lem:cons-over-interface} can be obtained by iteratively applying Lemma~\ref{lem:err-eff-int} in an analogous fashion to Corollary~\ref{cor:succ-int-h}.
\begin{lemma}  \label{lem:err-eff-int}
There exists a constant $\mu > 0$ and a threshold value $\delta_{th} > 0$, such that the following holds:

Consider circuit-level noise with parameter $\delta < \delta_{th}$. Let $\cN: \bL((\bC^2)^{\otimes n_r} \otimes R) \to \bL((\bC^2)^{\otimes n_r} \otimes R)$ with $R$ being a reference system $R$, be a channel with stabilizer-reduced weight $\mu n_r$ with respect to $n_r$ qubits in its input. Then, for any target error rate $\overline{\delta} > 0$, there exists a decreasing function $\ovr(\overline{\delta}) > 0$, such for any $r > r' \geq \ovr$, the following holds:\footnote{for larger target error $\overline{\delta}$, the level of protection $\ovr$ is smaller.}

\smallskip The noisy realization $\tilde{\cT}_{\overline{\Xi}^i_{r, r'}}$ of the effective interface $\overline{\Xi}^i_{r, r'}$ satisfies,
\begin{equation}  \label{eq:eff-circ-noise}
(\tilde{\cT}_{\overline{\Xi}^i_{r, r'}} \otimes \cI[R]) \circ \cN \circ \big( \cE^i_r \otimes \cI[R] \big) =  \sum_{F \subseteq \{(i, j) \mid j \in [2^{r-r'}] \}}  \sum_{\omega \in \Omega_F} \mathrm{Pr}(F, \omega) \: \cT_{F, \omega},
\end{equation}
where $\Omega_F$ is a finite set and 
\begin{equation}  \label{eq:can-F}
\cT_{F, w} = \left((\otimes_{j \in [2^{r-r'}] \setminus F} \cI^{i, j}_{n_{r'}}) \otimes \cZ_{F, \omega} \right) \circ  \cN_{F, \omega} \circ \Big((\otimes_{j \in [2^{r-r'}] \setminus F} \cE^{i, j}_{r'}) \otimes  \cE_{F, \omega} \otimes  \cI[R] \Big),
\end{equation}
with $\cE_{F, \omega}$ being a quantum channel from $m_{r'} |F|$ qubit input to a larger system, and $\cZ_{F, \omega}$ being an arbitrary channel from this larger system to $m_{r'} |F|$ qubits. $\cN_{F, \omega}$ is a channel with stabilizer-reduced weight $\mu n_{r'}$ with respect to each block of $n_r$ qubit corresponding to $(i, j) \not \in F$. 

\smallskip Finally, $\pr(F) = \sum_{\omega \in \Omega_F} \pr(F, \omega)$ is such that for any $T \subseteq  \{(i, j) \mid j \in [2^{r-r'}]\}$, we have
\begin{equation}  \label{eq:pr-T}
 \mathrm{Pr}(T \subseteq F) \leq (2\overline{\delta})^{|T|}.
\end{equation}
\end{lemma}

\subsection{Informal overview of the error analysis} \label{sec:inf-err-analysis}
Before presenting the formal error analysis, we describe here the main idea behind it. Let $z:=r-r_0$. The effective interface $\Xi^i_{r,r_0}$ recursively applies partial interfaces: one block of $\cC_r$ is mapped into two blocks of $\cC_{r-1}$, each of these blocks is mapped into two blocks of $\cC_{r-2}$, and so on, until $2^z$ blocks of $\cC_{r_0}$ are obtained. 

\smallskip This branching structure can naturally be represented by a binary tree. A node at depth $y$ represents a partial interface $\Gamma_{r-y,r-y-1}$, which is applied on a block of $\cC_{r-y}$, and its each child represents an interface $\Gamma_{r-y-1,r-y-2}$ applied on a block of $\cC_{r-y-1}$.

\smallskip When a partial interface receives a valid input, namely a code state affected  by a low stabilizer-reduced weight, it has two possible outcomes. If it succeeds, its two output blocks are again valid inputs for the interfaces at the next level. If it fails, an arbitrary error may act on its output. In that case, the descendant interfaces no longer receive valid inputs, and their fault-tolerance guarantees cannot be applied. We therefore regard all final output blocks descending from the failed interface as erroneous.

\smallskip A \emph{block error pattern} $(F_0,F_1,\ldots,F_{z-1})$ records the interfaces where the fault-tolerance guarantee is lost. More precisely, $F_y$ contains the interfaces  which fail at depth $y$, provided that no interface on the path from the root to that block has already failed. Thus, a failed node and its descendants do not both appear in the block error pattern: once a failure occurs, the entire descendant subtree is absorbed into the corresponding arbitrary error channel. The set $F$ of erroneous final output blocks is obtained by taking, for every node appearing in the block error pattern, all of its descendants at the final level.

\smallskip For example, suppose that $z=4$ (see also~Fig.~\ref{fig:block-error-path}), so that one input block of $\cC_r$ produces $16$-blocks of $\cC_{r-4}$. A failure of the interface at the root may corrupt entire input block $\cC_{r}$; hence all $16$-output blocks are corrupted. If the root interface succeeds but one of the two interfaces at the next level fails, then the eight of the final blocks descending from that interface may be corrupted. A failure at the last level affects only two final output blocks. A block error pattern may contain failures in different branches of the tree, but it does not contain two failures related by an ancestor-descendant relation.

\smallskip The main difficulty is that the erroneous output blocks are correlated. A single failure near the root can make a large collection of output blocks erroneous, so the events that individual output blocks are erroneous cannot be treated as independent. Nevertheless, failures near the root are much less likely than failures near the leaves. Lemma~\ref{lem:choice-ovr} shows that the failure probability of a partial interface decreases doubly exponentially with its encoding level. Therefore, an interface whose failure affects a large descendant subtree has a relatively smaller failure probability.

\smallskip Section~\ref{sec:bloc-path-tree} formalizes this tradeoff using a binary tree. Each node $v$ in the tree is assigned a random variable $X_v\in\{0,1,2\}$. The value $X_v=0$ means that the interface at $v$ receives a valid input and succeeds, while $X_v=1$ means that it receives a valid input and fails. The value $X_v=2$ means that an ancestor of $v$ has already failed, so that the interface at $v$ does not receive a valid input and is absorbed into the earlier arbitrary error channel. Thus, the event $X_v\in\{1,2\}$ means that the block corresponding to $v$ either failed or lies below a failure.

\smallskip The node weight introduced in Definition~\ref{def:nweight} counts how many leaf nodes lie below a given collection of nodes. It therefore measures the number of final output blocks that can be affected by failures associated with those nodes. Lemma~\ref{lem:rec} gives a recursive decomposition of the event that a collection of nodes lies below failures, and Lemma~\ref{lem:final-upper-bound} iterates this decomposition to show that its probability decreases exponentially with the corresponding node weight. For nodes at the final level, the node weight is simply the number of selected blocks. This gives the bound $\Pr(T\subseteq F)\leq (2\overline{\delta})^{2|T|}$ for every set $T$ of final output blocks.

\smallskip Each leaf node has two children corresponding to final output blocks. Hence, if a set $T$ of final output blocks is erroneous, then at least $|T|/2$ distinct leaf nodes must either fail or lie below a failure. It follows that $\Pr(T\subseteq F)\leq (2\overline{\delta})^{|T|}$,  which is the local stochastic property needed in Lemma~\ref{lem:err-eff-int}.

\subsection{Doubly-exponential decrease of error probability of the partial interface} \label{sec:choice-ovr} 
We consider the error probability $\tau_{r}$ corresponding to the partial interface $\Gamma_{r, r'}$ from Lemma~\ref{lem:succ-int}, with $r' = r-1$. For Lemma~\ref{lem:err-eff-int}, the choice of  $\ovr$ is made according to Lemma~\ref{lem:choice-ovr}. Note from Eq.~\eqref{eq:double-exp} that the error probability decreases doubly exponentially with respect to $r-\ovr$ for $r \geq \ovr$. This scaling of the error probability will be important for the proof  of  Lemma~\ref{lem:err-eff-int} given in Sections~\ref{sec:block-err-path} and~\ref{sec:bloc-path-tree}.
\begin{lemma} \label{lem:choice-ovr}
 Consider the threshold value $\delta_{th} > 0$  and the error probability $\tau_{r}$ from Lemma~\ref{lem:succ-int}. Then, for any $\delta < \delta_{th}$ and any $\overline{\delta} > 0$, there exists an $\ovr(\delta, \overline{\delta}) > 0$, which is an increasing function of $\delta$ and decreasing function of $\overline{\delta}$ such that the following holds for all $r > \ovr$
\begin{equation} \label{eq:double-exp}
   \tau_r \leq {\overline{\delta}}^{2^{r - \ovr}}.
\end{equation}
\end{lemma}
\begin{proof}
    \smallskip From Corollary~\ref{cor:succ-int-h}, the probability $\tau_r$ is upper bounded as follows
\begin{equation} \label{eq:tau-r}
    \tau_r \le \lambda p_3(m_{r-1}) (c \delta)^{c' m_{r-1}} .
\end{equation}
Thus, for some constants $\lambda' > 0$, we have 
\begin{equation}
   \tau_r \leq \lambda m^{\lambda'}_{r-1} (c \delta)^{c' m_{r-1}}
\end{equation}
For any given $\overline{\delta} >0$, we consider the smallest $\overline{r} \geq 3$, such that
\begin{equation} \label{eq:bound-ovr}
 \lambda_1 m^{\lambda_2}_{\overline{r}-1} (c\delta)^{c' m_{\overline{r}-1}} \leq \overline{\delta}.
\end{equation}
Note that, for any fixed $\delta$, the value of $\ovr(\delta,\overline{\delta})$ increases as $\overline{\delta}$ decreases.
Moreover, for any fixed $\overline{\delta}$, the value of $\ovr(\delta,\overline{\delta})$ decreases as $\delta$ decreases.

Consider $r \geq \ovr$. Then, we have for $r > \overline{r}$
\begin{align} 
   \tau_r & \leq   \lambda_1 m^{\lambda_2}_{r-1} (c\delta)^{m_{r-1}} \nonumber\\
   &= \lambda_1 (2^{r - \ovr} m_{\ovr - 1   })^{\lambda_2} (c \delta)^{c' 2^{r - \ovr} m_{\ovr - 1}}  \nonumber \\
   &  <  (\lambda_1  m_{\ovr - 1}^{\lambda_2} \: (c\delta)^{c' m_{\ovr - 1}})^{2^{r - \ovr}} 
\nonumber   \\ 
   &\leq {\overline{\delta}}^{2^{r - \ovr}}, \label{eq:bound-ovr-y}
\end{align}
where in the first equality, we have used $m_{r - 1} = 2^{r -\ovr} m_{\ovr - 1}$ and in the second inequality, we have used $2^{r - \ovr} m_{\ovr-1} < (m_{\ovr-1})^{2^{r - \ovr}}$ as $ m_{\ovr-1} \geq 2$ when $\ovr \geq 3$ (see the property of $\cC_r$ stated in Eq.~\eqref{eq:m-r-r-1}). The last inequality follows from Eq.~\eqref{eq:bound-ovr}. 
\end{proof}

\subsection{Block error pattern: error analysis of the effective interface} \label{sec:block-err-path}
In this section, we introduce a notion of block error pattern, as a tool in the error analysis of the effective circuit $\overline{\Xi}^i_{r, r'}$. This notion will play a crucial role in the proof of Lemma~\ref{lem:err-eff-int}.

\smallskip We have the following for the left hand side of Eq.~\eqref{eq:eff-circ-noise} using Eq.~\eqref{eq:eff-circ}
\begin{multline} \label{eq:eff-circ-noisy}
 (\tilde{\cT}_{\overline{\Xi}^i_{r, r'}} \otimes \cI[R]) \circ \cN \circ \big( \cE^i_r \otimes \cI[R] \big) \\ = \Big(  \big( (\otimes_{j_1, \dots, j_{r-r'-1} \in \{0, 1\}} \: \tilde{\cT}_{\Gamma^{ i, j_1, \dots, j_{r-r'-1}}_{r' + 1, r'}}) \circ \cdots \circ (\otimes_{ j_1 \in \{0, 1\}} \: \tilde{\cT}_{\Gamma^{i, j_1}_{r-1, r-2}}) \circ \tilde{\cT}_{\Gamma^i_{r, r-1}} \big) \otimes \cI[R] \Big) \\ \circ \cN \circ \big( \cE_r \otimes \cI[R] \big),
\end{multline}
To evaluate the right hand side of Eq.~\eqref{eq:eff-circ-noisy}, we will apply Corollary~\ref{cor:succ-int-h} many times in sequence; thereby, decreasing the encoding level of the ideal encoder as well as removing the noisy interface one by one.

\smallskip From Lemma~\ref{lem:succ-int}, we have
\begin{multline} \label{eq:step-1}
\big(\tilde{\cT}_{\Gamma^i_{r, r-1}}  \otimes \cI[R] \big) \circ \cN \circ \big(\cE^{i}_r \otimes \cI[R] \big) \\ = (1 - \tau_r)  \cN^{(1)} \circ (\otimes_{j_1 \in \{0, 1\} }\cE^{i,j_1}_{r-1} \otimes \cI[R])  + \tau_r \: (\cZ^i_{r \to r-1} \otimes \cI[R]) \circ \cN \circ (\cE^i_r \otimes \cI[R]),  
\end{multline}
where $\cN^{(1)}$ is a channel with stabilizer-reduced weight $\mu n_{r-1}$ with respect to each $n_{r-1}$ qubit system $(i, j_1), j_1\in \{0, 1\}$ and $\cZ^i_{r \to r-1}$ is an arbitrary channel.

\smallskip We note that in the second term on the right hand side of Eq.~\eqref{eq:step-1}, the logical information encoded in $\cC_r$ is corrupted by an arbitrary error channel $\cZ^i_{r \to r-1}$. In this event, we can not further apply Corollary~\ref{cor:succ-int-h} as the stabilizer-reduced weight of the effective channel $(\cZ^i_{r \to r-1} \otimes \cI[R]) \circ \cN $ is not controlled, and therefore, absorb the remaining interfaces in $\cZ^i_{r \to r-1}$ leading to an overall error
\begin{equation}
    \cZ^i_{r \to r'} := \big( (\otimes_{j_1, \dots, j_{r-r'-1} \in \{0, 1\}} \: \tilde{\cT}_{\Gamma^{i, j_1, \dots, j_{r-r'-1}}_{r' + 1, r'}}) \circ \cdots \circ (\otimes_{j_1 \in \{0, 1\}} \: \tilde{\cT}_{\Gamma^{i, j_1}_{r-1, r-2}})  \circ \cZ^i_{r \to r - 1} \big).
\end{equation}
 On the first term on the right hand side of Eq.~\eqref{eq:step-1}, the logical information encoded in $\cC_r$ is successfully decoded into two blocks of $\cC_{r-1}$, where we may further apply Corollary~\ref{cor:succ-int-h} as follows:
\begin{multline}
    \big((\otimes_{j_1 \in \{0, 1\}} \: \tilde{\cT}_{\Gamma^{i, j_1}_{r-1, r-2}}) \otimes \cI[R] \big) \circ \cN^{(1)} \circ (\otimes_{j_1 \in \{0, 1\} }\cE^{i, j_1}_{r-1} \otimes \cI[R]) \\ = \sum_{F_1 \subseteq \{(i, j_1) \mid j_1 \in \{0, 1\} \}} (1 - \tau_{r-1})^{2 - |F_1|} \: \tau^{|F_1|}_{r-1} \: \left((\otimes_{\substack{(i, j_1) \not\in F_1, \\ j_2 \in \{0, 1\}} } \: \cI_{n_{r-2}}^{i, j_1, j_2}) \otimes ( \otimes_{(i, j_1) \in F_1} \cZ^{i, j_1}_{r-1 \to r- 2})\right) \\ \circ \cN_{F_1} \circ \left((\otimes_{\substack{(i, j_1) \not\in F_1, \\ j_2 \in \{0, 1\} }} \: \cE_{r-2}^{i, j_1, j_2}) \otimes(\otimes_{(i, j_1) \in F_1} \: \cE^{i, j_1}_{r-1}) \otimes \cI[R]\right),  
\end{multline}
where $\cN_{F_1}$ is a quantum channel with stabilizer-reduced weight $\mu n_{r-2}$ with respect to each $n_{r-2}$ qubit system labelled by $(i, j_1) \not \in F_1$ and $\cZ^{i, j_1}_{r-1 \to r- 2}$ is an arbitrary channel applied on the output of $\cE^{i, j_1}_{r-1}$. The logical information corresponding to $\cE^{i, j_1}_{r-1}, (i, j_1) \in F_1$ is corrupted by an arbitrary error channel $\cZ^{i, j_1}_{r-1 \to r- 2}$, and on these blocks, we absorb the remaining interfaces leading to the following overall error
\begin{equation}
\cZ^{i, j_1}_{r-1 \to r'} := \left( (\otimes_{j_2, \dots, j_{r-r'-1} \in \{0, 1\}} \: \tilde{\cT}_{\Gamma^{i, j_1, \dots, j_{r-r'-1}}_{r' + 1, r'}}) \circ \cdots \circ (\otimes_{j_2 \in \{0, 1\}} \: \tilde{\cT}_{\Gamma^{i, j_1, j_2}_{r-2, r-3}})  \circ \cZ^{i, j_1}_{r-1 \to r - 2} \right).
\end{equation}
An element from $\{(i, j_1) \mid j_1 \in \{0, 1\}\}$ is independently selected to be in $F_1$ with probability $\tau_{r-1}$. If $(i, j_1) \not\in F_1$, the logical information is successfully decoded into two blocks of $\cC_{r-2}$ up to a controlled error, with the corresponding encoding denoted by $\cE^{i, j_1, j_2}_{r-2}, (i, j_1) \not\in F_1, j_2 \in \{0, 1\}$.

\smallskip We further apply Corollary~\ref{cor:succ-int-h} for $\tilde{\cT}^{i, j_1, j_2}_{\Gamma_{r-2, r-3}}$ on the blocks of $\cC_{r-2}$ corresponding to $(i, j_1) \not\in F_1, j_2 \in \{0, 1\}$. Continuing as above, for a subset $F_2 \subseteq \{ (i, j_1, j_2) \mid (i, j_1) \not\in F_1, j_2 \in \{0, 1\} \}$, the corresponding $\tilde{\cT}^{i, j_1, j_2}_{\Gamma_{r-2, r-3}}$ applies an arbitrary channel $\cZ^{j_1, j_2}_{r-2, r-3}$  on the output system corresponding to $\cE^{i, j_1, j_2}_{r-2}$.  An element from $\{ (i, j_1, j_2) \mid (i, j_1) \not\in F_1, j_2 \in \{0, 1\} \}$ is independently selected to be in $F_2$ with probability $\tau_{r-2}$. For the remaining $\tilde{\cT}^{i, j_1, j_2}_{\Gamma_{r-2, r-3}}, (i, j_1, j_2) \not\in  F_2$, the logical information corresponding to $\cE^{i, j_1, j_2}_{r-2}$  is successfully decoded into two blocks of $\cC_{r-3}$, denoted by $\cE^{i, j_1, j_2, j_3}_{r-2}, (i, j_1, j_2) \not\in F_2, j_3 \in \{0, 1\}$.

\smallskip Continuing in this fashion, we have random sets $(F_0, F_1, F_2, \dots, F_{r-r'-1})$, where $F_0 \subseteq \{ i\}$ and for $0 < y \leq r-r'-1$, $F_y \subseteq \{ (i, j_1, j_2, \dots, j_y) \mid j_1, j_2, \dots, j_y \in \{0, 1\}\}$, such that the noisy interface $\Gamma^{i, j_1, \dots, j_y}_{r-y, r-y-1}$ applies an arbitrary channel on the code block of $\cC_{r-y}$ denoted by $\Gamma^{i, j_1\dots j_y}_{r-y}$ for all $(i, j_1, j_2, \dots, j_y) \in F_y$. The sets $F_0, F_1, F_2, \dots, F_{r-r'-1}$ are not independent. In particular, $F_y$ for any $0 < y \leq r-r'-1$ depends on $F_0, \dots, F_{y-1}$ in the following way:

\smallskip If $(i, j_1, \dots, j_{y'})  \in F_{y'}, 0 \leq y' < y$, then the corresponding noisy $\Gamma^{i, j_1, \dots, j_{y'}}_{r-y', r-y'-1}$ applies an arbitrary channel $\cZ^{i, j_1, \dots, j_{y'}}_{r-y' \to r-y'-1}$ on the output system of $\cE^{i, j_1, \dots, j_{y'}}_{r-y}$. In this event, one can not further apply Corollary~\ref{cor:succ-int-h} on the output of $\cZ^{i, j_1, \dots, j_{y'}}_{r-y' \to r-y'-1}$, which would have otherwise sequentially yielded blocks of $r-y'-1, \dots, r-y$ labelled by $(i, j_1, \dots, j_{y'}, j_{y'+1}), \break (i, j_1, \dots, j_{y'}, j_{y'+1}, j_{y'+2}), \dots, (i, j_1, \dots, j_{y'}, j_{y'+1}, \dots, j_y)$, where $j_{y'+1}, \dots, j_y \in \{0, 1\}$. Hence, $(i, j_1, \dots, j_{y'}, j_{y'+1}, \dots, j_y)$ is not anymore available for selection in $F_y$. 
\begin{figure}[!t]
\centering
\resizebox{\linewidth}{!}{
\begin{tikzpicture}[
  level distance=1.5cm,
  level 1/.style={sibling distance=8cm},
  level 2/.style={sibling distance=4cm},
  level 3/.style={sibling distance=2cm},
  every node/.style={draw,circle,minimum size=7mm},
  edge from parent/.style={draw,-}]
  
\node[label={[label distance = -2pt] above:$i$}] {}
  child {node[fill=red!40, label= {[label distance = -8pt] above:{$(i, 0)$}}] {} 
    child {node[label= {[label distance = -12pt]above:{$(i, 0, 0)$}}]{ } 
      child {node[label= {[label distance = -16pt] below:{$(i, 0, 0, 0)$}}] { }} 
      child {node[label= {[label distance = -16pt] below:{$(i, 0, 0, 1)$}}] { }}}
    child {node[label= {[label distance = -16pt] above:{$(i, 0, 1)$}}] { } 
      child {node[label= {[label distance = -16pt] below:{$(i, 0, 1, 0)$}}]{ }}
      child {node[label= {[label distance = -16pt] below:{$(i, 0, 1, 1)$}}] { }}}}
  child {node[label= {[label distance = -8pt] above:{$(i, 1)$}}] { } 
    child {node[label= {[label distance = -16pt] above:{$(i, 1, 0)$}}] { } 
      child {node [label= {[label distance = -16pt] below:{$(i, 1, 0, 0)$}}][fill=red!40]{ }}
      child {node[label= {[label distance = -16pt] below:{$(i, 1, 0, 1)$}}] [fill=red!40]{ }}}
    child {node[fill=red!40, label= {[label distance = -8pt] above:{$(i, 1, 1)$}}] { } 
      child {node[label= {[label distance = -16pt] below:{$(i, 1, 1, 0)$}}] { }}
      child {node[label= {[label distance = -16pt] below:{$(i, 1, 1, 1)$}}] { }}}};

\node[draw= none] at (-8.0,0) {$y=0$};
\node[draw= none] at (-8.0,-1.5) {$y=1$};
\node[draw= none] at (-8.0,-3) {$y=2$};
\node[draw= none] at (-8.0,-4.5) {$y=3$};
\end{tikzpicture}
}
\caption{The figure illustrates an example of block error pattern for $r-r' -1 = 3$. The set of nodes at depth $0 \leq y \leq r-r'-1$ is given by $\{ (i, j_1, \dots, j_y) \mid j_1, \dots, j_y \in \{0, 1\} \}$. The red nodes denote the nodes where error maps are applied; therefore, the block error pattern is given by $(F_0, F_1, F_2, F_3)$, where $F_0 = \emptyset$, $F_1 = \{ (i, 0)\}$, $F_2 = \{(i, 1, 1)\}$, and $F_3 = \{{(i, 1, 0, 0), (i, 1, 0, 1)}\}$.}
\label{fig:block-error-path}
\end{figure}

\medskip Therefore,  the set $F_y$ is obtained by independently selecting elements from  $\{ (i, j_1, \dots, j_y) \mid (i, j_1, \dots, j_{y'}) \not\in F_{y'}, \forall y' < y\}$, with probability $\tau_{r-y}$. This leads to the definition of block error pattern in Def.~\ref{def:block-err-path} (see also Fig.~\ref{fig:block-error-path}). 
\begin{definition}[Block error pattern] \label{def:block-err-path}
Consider $(F_{0}, F_{1}, \dots, F_{r-r'- 1})$, where $F_0 \subseteq \{ i\}$ and for $0 < y \leq r-r'-1$, $F_{y} \subseteq \{ (i, j_1, j_2, \dots, j_{y}) \mid j_1, j_2, \dots, j_{y} \in \{0, 1\}  \}$. We say $(F_{0}, F_{1}, \dots, F_{r-r'-1})$ is a block error pattern if the following condition holds for any $y < r-r'-1$: 
\begin{multline} \label{eq:bl-err-path-cond}
  \text{If } (i, j_1, \dots, j_{y}) \in F_{y}, \text{ then }  (i, j_1, \dots, j_{y}, j_{y+1}, \dots, j_{y'})  \not\in F_{y'}, \: \forall y' \text{with } y < y' \leq r-r'-1 \\ \text{ and } \forall j_{y+1}, \dots, j_{y'} \in \{0, 1\}.    \end{multline}
\end{definition}
\smallskip Define $\bF_{r-r'-1} : = \{ (F_0, F_1, \dots, F_r-r'-1) \text{ is a block error pattern}  \} $. Then, it follows that we can write $(\tilde{\cT}_{\overline{\Xi}^i_{r, r'}} \otimes \cI[R]) \circ \cN \circ \big( \cE^i_r \otimes \cI[R] \big)$ as convex combination over block error patterns, which denote failure of interfaces, as follows, 
\begin{multline} \label{eq:noisy-interface} (\tilde{\cT}_{\overline{\Xi}^i_{r, r'}} \otimes \cI[R]) \circ \cN \circ \big( \cE^i_r \otimes \cI[R] \big)  \\ = \sum_{(F_{0} F_1, \dots, F_{r-r'-1}) \in \bF_{r-r'-1}} \pr(F_{0}, F_1, \dots, F_{r-r'-1})  \cT_{(F_{0}, F_1, \dots, F_{r-r'-1})}, 
\end{multline}
where $\pr(F_{0}, F_1, \dots, F_{r-r'-1})$ is the probability that the block error pattern $(F_{0}, \dots, F_{r-r'-1})$ happens (see Section~\ref{sec:prob-block-err} below). The channel $\cT_{(F_{0}, F_1, \dots, F_{r-r'-1})}$ is obtained by applying an arbitrary channel $\cZ^{(i, j_1, \dots, j_y)}_{r-y \to r-y-1}$ on the encoding blocks corresponding to $(i, j_1, \dots, j_{y}), \: 0 \leq y \leq r-r'-1$ as follows
\begin{align}
\cT_{(F_{0}, \dots, F_{r-r'-1})}
&=
\Bigg[
\Bigg(
\otimes_{\substack{
(i, j_1, \dots, j_{r-r'-1}) :\\
(i, j_1, \dots, j_y)\notin F_y,\ \forall\, 0 \le y \le r-r'-1,\\
j_{r-r'} \in \{0,1\}
}}
\cI^{i, j_1, \dots, j_{r-r'}}_{r'}
\Bigg)
\nonumber\\[-0.2em]
&\qquad \otimes
\Bigg(
\otimes_{y=0}^{r-r'-1}
\ \otimes_{(i, j_1, \dots, j_y)\in F_y}
\cZ^{(i, j_1, \dots, j_y)}_{r-y \to r-y-1}
\Bigg) \otimes \cI[R]
\Bigg] \notag
\\[0.8em]
&\qquad\circ\, \cN_{F_1, \dots, F_{r-r'-1}} \,\circ \Bigg[
\Bigg(
\otimes_{\substack{
(i, j_1, \dots, j_{r-r'-1}) :\\
(i, j_1, \dots, j_y)\notin F_y,\ \forall\, y \le r-r'-1,\\
j_{r-r'} \in \{0,1\}
}}
\cE^{i, j_1, \dots, j_{r-r'}}_{r'}
\Bigg)
\nonumber\\[-0.2em] 
&\qquad \qquad \otimes
\Bigg(
\otimes_{y=0}^{r-r'-1}
\ \otimes_{(i, j_1, \dots, j_y)\in F_y}
\cE^{(i, j_1, \dots, j_y)}_{r-y}
\Bigg) \otimes \cI[R]
\Bigg],\label{eq:can-block-path}
\end{align}
where $\cN_{F_0, F_1, \dots, F_{r-r'-1}}$ is a quantum channel with stabilizer-reduced weight $\mu n_{r'}$ with respect to each block $(i, j_1, \dots, j_{r-r'})$ of $\cC_{r'}$, where $(i, j_1, \dots, j_{r-r'-1})$ is such that $(i, j_1, \dots, j_y) \not \in F_y$, for all $0 \leq y \leq r-r'-1$.

\paragraph{Probability associated with a block error pattern:} \label{sec:prob-block-err}
By the chain rule of probability, 
\begin{equation} \label{eq:pr-err-path}
   \pr(F_{0}, F_{1}, \dots, F_{r-r'-1})
   = \pr(F_{0}) \prod_{y = 1}^{r-r'-1} 
   \pr(F_{y} \mid F_{0}, F_{1}, \dots, F_{y-1}),
\end{equation}
where 
\begin{equation} \label{eq:pr-f-0}
   \pr(F_{0}) = 
   \begin{cases}
     1 - \tau_r, & \text{if } F_{0} = \varnothing, \\[6pt]
     \tau_r, & \text{if } F_{0} = \{i\}.
   \end{cases}
\end{equation}
For $y > 0$, conditioned on $(F_0, F_{1}, \dots, F_{y-1})$, the set $F_y$ is obtained by independently selecting elements from $\{ (i, j_1, \dots, j_y) \mid (i, j_1, \dots, j_{y'}) \not \in F_{y'}, \: \forall y' < y\}$, where the probability of an element being selected is given by $\tau_{r-y}$. Note that $|\{ (i, j_1, \dots, j_y) \mid (i, j_1, \dots, j_{y'}) \not \in F_{y'}, \: \forall y' < y\}| = 2^y - \sum_{y'=0}^{y-1} 2^{(y-y')} |F_{y'}|$. Therefore, we have
\begin{align} 
   \pr(F_{y} \mid F_{0}, F_{1}, \dots, F_{y-1})
   &= (1 - \tau_{r-y})^{2^{y} - \sum_{y'=0}^{y-1} 2^{(y-y')} |F_{y'}| - |F_y|}
     \: \tau_{r-y}^{|F_{y}|} \nonumber\\
   & = (1 - \tau_{r-y})^{2^{y} - \sum_{y'=0}^{y} 2^{(y-y')} |F_{y'}|}
     \: \tau_{r-y}^{|F_{y}|} \label{eq:pr-f-y}.   
\end{align}
Using the notion of block error pattern, we will now show that Eq.~\eqref{eq:eff-circ-noise}-\eqref{eq:can-F} hold. We first consider the following definitions of \emph{extension} and \emph{partition}.
\begin{definition}[Extension] \label{def:proj}
Consider $F_{y} \subseteq \{(i, j_1, \dots, j_y) \mid j_1, \dots, j_y \in \{0, 1\} \}$. For any $y' > y$, we define the extension set $F_{y \to y'}$ by extending all binary strings in $F_y$ to length $y'$ as follows 
\begin{equation} \label{eq:proj-def}
 F_{y \to y'} :=  \{(i, j_1, \dots,j_y, j_{y+1}, \dots, j_{y'}) \mid  (i, j_1, \dots,j_y) \in F_y \text{ and }  j_{y+1}, \dots, j_{y'} \in \{0, 1\}\}. 
\end{equation}  
\end{definition}
\begin{definition}[Partition] \label{def:partition}
Consider $r-r'-1 > 0$ and $\overline{F} \subseteq \{(i, j_1, \dots, j_{r-r'-1}) \mid  j_1, \dots, j_{r-r'-1} \in \{0, 1\} \}$.  We say a block error pattern $(F_{0}, F_1, \dots, F_{r-r'-1})$ induces a partition of $\overline{F}$ and denote it by $(F_{0}, \dots, F_{r-r'-1}) \triangleright \overline{F}$ if
\begin{equation} \label{eq:part-cond}
 \overline{F} = \cup_{y = 0}^{r-r'-1} \: F_{y \to r-r'-1},
 \end{equation}
 where $F_{y \to r-r'-1}$ is according to Def.~\ref{def:proj}.
\end{definition}

\smallskip Let $F = \{(i, j_1, \dots, j_{r-r'}) \mid (i, j_1, \dots, j_{r-r'-1}) \in \overline{F},  j_{r-r'} \in \{0, 1\} \}$ be the extension of $\overline{F}$ from layer $r-r'-1$ to $r-r'$. Using Remark~\ref{rem:notation-cons-ovr}, we identify $j_1, \dots, j_{r-r'}$ with a $j \in [2^{r-r'}]$; therefore $F \subseteq \{(i, j) \mid i \in [h], j \in [2^{r-r'}] \}$. This set $F$ will play the role of the random set $F$ in Lemma~\ref{lem:err-eff-int} as we will see below.

\vspace*{2mm}
\begin{proof}[Proof of Eq.~\eqref{eq:eff-circ-noise}-\eqref{eq:can-F}:] 
Consider a subset $\overline{F} \subseteq \{(i, j_1, \dots, j_{r-r'-1}) \mid j_1, \dots, j_{r-r'-1} \in \{0, 1\}\}$ and let $(F_0, \dots, F_{r-r'-1})$ be a block error pattern such that $(F_0, \dots, F_{r-r'-1}) \triangleright \overline{F}$. Let $F$ be the extension of $\overline{F}$ from layer $r-r'-1$ to $r-r'$. Note that sets $F$ and $\overline{F}$ completely determine each other. We define the set $\Omega_F$ in Lemma~\ref{lem:err-eff-int} as
\begin{equation} \label{eq:omega-F}
  \Omega_F := \{ (F_0, \dots, F_{r-r'-1}) \in \bF_{r-r'-1} \mid (F_0, \dots, F_{r-r'-1}) \triangleright \overline{F}
 \}  
\end{equation}
We define the channel $\cT_{F, \omega}$ for any $\omega = (F_0, \dots, F_{r-r'-1}) \in \Omega_F$ as $\cT_{(F_{0}, \dots, F_{r-r'-1})}$ from Eq.~\eqref{eq:can-block-path}.

\smallskip Then, we can rewrite Eq.~\eqref{eq:noisy-interface} as
\begin{align}
(\tilde{\cT}_{\overline{\Xi}^i_{r, r'}} \otimes \cI[R]) \circ \cN \circ \big( \cE^i_r \otimes \cI[R] \big)  
 &= \sum_F \: \sum_{(F_{0}, \dots, F_{r-r'-1}) \triangleright \overline{F}} \pr(F_{0}, \dots, F_{r-r'-1})  \cT_{(F_{0}, \dots, F_{r-r'-1})} \nonumber \\
 & \qquad = \sum_F \: \sum_{\omega \in \Omega_{F, \omega}} \pr(F, \omega) \cT_{F, \omega},  \label{eq:can-T-F}
\end{align}
where $\pr(F, \omega) := \pr(F_{0}, \dots, F_{r-r'-1})$ for  $\omega = (F_0, \dots, F_{r-r'-1}) \in \Omega_F$ and $\cT_{F, \omega} :=\cT_{(F_{0}, \dots, F_{r-r'-1})} $.
From Eq.~\eqref{eq:can-block-path}, we have 
\begin{multline}  \label{eq:T-path}
  \cT_{F, \omega} =\cT_{(F_{0}, \dots, F_{r-r'-1})} = \left(\otimes_{(i, j_1, \dots, j_{r-r'}) \not \in F } \: \cI^{i, j_1, \dots, j_{r-r'}}_{r'}) \otimes \cZ_F \otimes \cI[R] \right) \\ \circ \cN_{F, \omega} \circ \left( (\otimes_{(i, j_1, \dots, j_{r-r'}) \not \in F} \: \cE^{i, j_1, \dots, j_{r-r'}}_{r'}) \otimes \cE_{F} \otimes \cI[R] \right),
\end{multline}
where $\cE_F = \otimes_{y = 0}^{r-r'-1}(\otimes_{(i, j_1, \dots, j_y) \in F_{y}} \cE^{i, j_1, \dots, j_y}_{r-y})$ and $\cZ_F = \otimes_{y = 0}^{r-r'-1}(\otimes_{(i, j_1, \dots, j_y) \in F_{y}} \cZ^{i, j_1, \dots, j_y}_{r-y \to r-y-1})$. 

\medskip From Eq.~\eqref{eq:can-T-F} and Eq.~\eqref{eq:T-path},  and using the binary representation of $(i, j)$ as in Remark~\ref{rem:notation-cons-ovr}, we get Eq.~\eqref{eq:eff-circ-noise} and  Eq.~\eqref{eq:can-F}. 
\end{proof}
We note that 
\begin{equation} \label{eq:pr-F}
    \pr(F) = \sum_{\omega \in \Omega_F} \pr(F, \omega) = \sum_{(F_0, F_1, \dots, F_{r-r'-1}) \triangleright \overline{F}}  \pr(F_0, F_1, \dots, F_{r-r'-1}).
\end{equation}

\medskip To prove Lemma~\ref{lem:err-eff-int}, it remains to show that for any $T \subseteq  \{(i, j_1, \dots, j_{r-r'}) \mid j_1, \dots, j_{r-r'} \break \in \{0, 1\} \}$, the inclusion probability $\pr(T \subseteq F)$ satisfies Eq.~\eqref{eq:pr-T}. Note that since the random set $F$ is an extension of the $\overline{F}$, the inclusion probability $ \pr(T \subseteq F)$ is upper bounded by $\pr( \overline{T} \subseteq \overline{F})$ for a $\overline{T} \subseteq \{(i, j_1, \dots, j_{r-r'-1}) \mid j_1, \dots, j_{r-r'-1} \in \{0, 1\} \}$ such that its extension from layer $r-r'-1$ to $r-r'$ includes the set $T$. In the following, we will show that for any $\overline{T} \subseteq \{(i, j_1, \dots, j_{r-r'-1}) \mid j_1, \dots, j_{r-r'-1} \in \{0, 1\} \}$ 
\begin{equation} \label{eq:up-T-bar}
    \pr( \overline{T} \subseteq \overline{F}) \leq (2\overline{\delta})^{2|\overline{T}|}.
\end{equation}
Using Eq.~\eqref{eq:up-T-bar} and  $|\overline{T}| \geq \frac{|T|}{2}$, we get $\pr(T \subseteq F) \leq (2\overline{\delta})^{|T|}$; hence proving Eq.~\eqref{eq:pr-T}.

\subsection{Block error pattern on a binary tree} \label{sec:bloc-path-tree}
In this section, we simulate block error pattern $(F_0, \dots F_{r-r'-1})$ from Eq.~\eqref{eq:pr-err-path}-\eqref{eq:pr-f-y} and $\overline{F}$ from Eq.~\ref{eq:pr-F} on a binary tree. We will then use this correspondence to upper bound $\pr( \overline{T} \subseteq \overline{F})$ for any $\overline{T} \subseteq \{(i, j_1, \dots, j_{r-r'-1}) \mid j_1, \dots, j_{r-r'-1} \in \{0, 1\}\}$.

\paragraph{Mapping interfaces to a perfect binary tree:} 
We represent the interfaces $\Gamma^i_{r, r-1}, \break \Gamma^{i, j_1}_{r-1, r-2}, \dots, \Gamma^{i, j_1, \dots, j_{r-r'-1}}_{r'+1, r'}$ on a perfect binary tree of depth $r-r'-1$ using their indices $j_1, \dots, j_{r-r'-1} \in \{0, 1\}$ as follows (see Fig.~\ref{fig:block-error-path}):\footnote{A perfect binary tree is a binary tree where each node has two children except the leaf nodes and all the leaf nodes are at the same depth.}

\smallskip Let $i$ denote the root node, and let nodes at depth $0< y \leq r-r'-1$ be denoted by $(i, j_1, \dots, j_y)$, where $j_1, \dots, j_y \in \{0, 1\}$. The leaf nodes correspond to $(i, j_1, \dots, j_{r-r'-1}), \break j_1, \dots, j_{r-r'-1} \in \{0, 1\}$. For any $y < r-r'-1$, each node $(i, j_1, \dots, j_y)$ has two child nodes, namely left node $(i, j_1, \dots, j_y, 0)$ and right node $(i, j_1, \dots, j_y, 1)$, respectively.

\smallskip  For any node $v = (i, j_1, \dots, j_y), 0 \leq y \leq r-r'-1$ (note that for $y = 0$, $(i, j_1, \dots, j_y)$ corresponds to the root node $i$), the descendants of $v$ correspond to the perfect ``subtree" of depth $r-r'-1 - y$, with $(i, j_1, \dots, j_{y})$ being the root node of the subtree and nodes $(i, j_1, \dots, j_{y}, j_{y + 1}, \dots, j_{y'})$, where $j_{y + 1}, \dots, j_{y'} \in \{0, 1\}$ and $y < y' \leq r-r'-1$, being the nodes at depth $y' - y$ of the subtree. Moreover, its ancestors are given by the sequence $i \leftarrow (i, j_1) \leftarrow (i, j_1, j_2) \leftarrow \dots \leftarrow (i, j_1, j_2, \dots, j_{y-1}) \leftarrow (i, j_1, j_2, \dots, j_{y})$. 

\paragraph{Notation:} For any node $v = (i, j_1, \dots, j_y)$, we will use $\mathscr{C}(v)$ and $\tilde{\mathscr{C}}(v)$ to denote the set of its descendants and ancestors including the node $v$, that is,
\begin{align}
    \mathscr{C}(v) &:= \{ (i, j_1, \dots, j_y, j_{y+1}, \dots, j_{y'}) \mid j_{y+1}, \dots, j_{y'} \in \{0, 1\}, y \leq y' \leq r-r'-1  \}. \\
    \tilde{\mathscr{C}}(v) &:= \{ (i, j_1, \dots, j_y), (i, j_1, \dots, j_{y-1}), (i, j_1, \dots, j_{y-2}), \dots, i \}.
\end{align}
\paragraph{Failure events on the perfect binary tree:} We will now represent on the perfect binary tree, the block error pattern $(F_0, \dots, F_{r-r'-1})$ and the random set $F$, which is induced by the block error pattern. We recall that the probabilities associated with the random sets $(F_0, \dots, F_{r-r'-1})$ and $F$ are given in Eq.~\eqref{eq:pr-err-path}-\eqref{eq:pr-f-y} and Eq.~\eqref{eq:pr-F}, respectively.

\smallskip  Recall from Section~\ref{sec:block-err-path} that for the partial interface $\Gamma_{r-y, r-y-1}^{i, j_1, \dots, j_y}$ we have a notion of success and failure when applied on a valid input state, i.e., code state up to a low stabilizer-reduced weight channel. However, when $\Gamma_{r-y, r-y-1}^{i, j_1, \dots, j_y}$ fails on a valid input, an arbitrary error channel is applied on the corresponding block; hence, remaining interfaces $\Gamma_{r-y-1, r-y-2}^{i, j_1, \dots, j_y, j_{y + 1}}, \Gamma_{r-y-2, r-y-3}^{i, j_1, \dots, j_y, j_{y + 1}, j_{y + 2}},...$ do not receive a valid input for any $j_{y+1}, j_{y + 2}, \dots  \in \{0, 1\}$, and are simply absorbed in the error channel.

\smallskip We now associate a random variable $X_{v}$, taking values in $\{0, 1, 2\}$, with each node $v = (i, j_1, \dots, j_{y})$ in the perfect binary tree. 

\paragraph{Notation:} For any $0 \leq b \leq y$, we shall denote 
\begin{equation}
 v^{(-b)} := (i, j_1, \dots, j_{y-b}).
\end{equation}

\medskip\noindent Sequentially, starting from the root to the leaf nodes, $X_v$ is defined as follow,
\begin{enumerate}
\item[(i)]  \textbf{Root node:}  We start with a valid input on the root node. The root node takes only two values $0$ and $1$. If $\Gamma^i_{r, r-1}$ succeeds, $X_v = 0$. If it fails, $X_v = 1$.

\item[(ii)]  \textbf{Intermediate nodes:}
If the interface $\Gamma^{i, j_1, \dots, j_y}_{r-y, r-y-1}$ receives a valid input, i.e., $X_{v^{(-1)}} = 0$, with $v^{(-1)} = (i, j_1, \dots, j_{y-1})$ being the parent of $v$, the corresponding $X_{v}$ takes values in $\{0, 1\}$ and is independent of any $v' = (i, j'_1, \dots, j'_y)$ such that $v' \neq v$. In particular, if $\Gamma^{i, j_1, \dots, j_{y}}_{r-y, r-y-1}$ succeeds, we have $X_{v} = 0$, otherwise $X_{v} = 1$. We have $X_{v} = 2$ if the interface $\Gamma^{i, j_1, \dots, j_y}_{r-y, r-y-1}$ does not receive a valid input, i.e., $X_{v^{(-1)}} \in \{1, 2\}$.
\end{enumerate}

\smallskip We take a fixed constant $\overline{\delta} > 0$ and the corresponding $\overline{r}(\overline{\delta})$ as given in Lemma~\ref{lem:choice-ovr}. We compute the probability distribution corresponding to $X_{v}$ for all nodes $v$ in the perfect binary tree.

\smallskip In the remaining part of the paper, we will use the notation $z := r-r'$. When $v$ is the root node, we have 
\begin{align*}
\pr(X_i = 1) &= \tau_r \leq {\odelta}^{2^z}   \\
\pr(X_i = 0) &= 1 - \tau_r  \\
\pr(X_i = 2) &= 0,
\end{align*}
where the upper bound in the first line is due to Lemma~\ref{lem:choice-ovr}. When $v$ is an intermediate node, that is, $v = (i, j_1, \dots, j_y), y> 0$, the value of $X_v$ is determined by its parent $v^{(-1)}$ as follows
\begin{align}
  &\pr(X_{v} = 1  \mid X_{v^{(-1)}} = 0) = \tau_{r-y} \leq {\odelta}^{2^{z-y}} \label{eq:prob-one}\\
 &   \pr(X_{v} = 0   \mid X_{v^{(-1)}} = 0) = 1 - \tau_{r-y} \label{eq:prob-zero} 
    \\
 &   \pr(X_{v} = 2 \mid X_{v^{(-1)}} = 0) = 0 \label{eq:prob-two} \\
&    \pr(X_{v} = 2 \mid X_{v^{(-1)}} \in \{1, 2\}) = 1 \\
&    \pr(X_{v} \in \{0, 1\} \mid X_{v^{(-1)}} \in \{1, 2\}) = 0 \label{eq:prob-one-one}
\end{align}
In the following, we shall assume that $\overline{\delta} \leq \frac{1}{2}$. We note the probability that $X_{v} = 1$ for any node $v$ and conditioned on an event $E = \{ X_v = 0 \mid v \in S \}$, where $S$ is a subset of nodes on the binary tree, can be upper bounded as follows:\footnote{This upper bound will be useful in the proof of Lemma~\ref{lem:rec}.}
\begin{align}
    \pr(X_{v} = 1 \mid E) & =   \pr(X_{v} = 1  \mid X_{v^{(-1)}} = 0, E) \pr(X_{v^{(-1)}} = 0 \mid E) \nonumber \\ 
    & \qquad \qquad + \pr(X_{v} = 1  \mid X_{v^{(-1)}}  \in \{1, 2\}, E) \pr(X_{v^{(-1)}} \in \{1, 2\} \mid E) \nonumber \\
    & \leq \pr(X_{v} = 1  \mid X_{v^{(-1)}} = 0, E) \leq {\odelta}^{2^{z-y}}, \label{eq:up-xvone}
\end{align}
where the equality uses law of total probability. The first inequality uses  Eq.~\eqref{eq:prob-one-one}. For the last inequality consider the following two cases. If $S$ contains $v$ or a descendant of $v$, then the event $E$ implies $X_v = 0$ and the bound is immediate. Otherwise, conditioned on $X_{v^{(-1)}}=0$, the failure event at $v$ is independent of the zero-valued nodes specified by $E$, and the last inequality follows from Eq.~\eqref{eq:prob-one}.

\smallskip Let $F_0, F_1, \dots, F_{z-1}$ be random events such that for any $0 \leq y \leq z - 1$,
\begin{equation}
    F_y := \{ v = (i, j_1, \dots, j_y) \mid X_{v} = 1 \}.
\end{equation}
Note that if $(F_0, F_1, \dots, F_{z-1})$ is not a block error pattern, the probability that the events $F_0, F_1, \dots, F_{z-1}$ happen is given by $\pr(F_0, F_1, \dots, F_{z-1}) = 0$.  If $(F_0, F_1, \dots, F_{z-1})$ is a block error pattern, the corresponding probability $\pr(F_0, F_1, \dots, F_{z-1})$ is obtained to be the same as given in Eq.~\eqref{eq:pr-err-path}-\eqref{eq:pr-f-y}.

\smallskip Let $\overline{F}$ be the random set
\begin{equation} \label{eq:random-set-F}
    \overline{F} := \{v = ( i, j_1, \dots, j_{z-1}) \mid X_{v} \in \{1, 2\} \}
\end{equation}
Using the fact that $X_{v} \in \{1,2\}$ if and only if $X_{v'} = 1$ for some $v' \in \tilde{\mathscr{C}}(v)$, it follows that 
$F$ happens if and only if the events $(F_0, F_1, \dots, F_{z-1})$ happen such that  $F = \cup_{y=0}^{z-1} F_{y \to z-1}$, where $F_{y \to z-1}$ is the extension set of $F_y$ as given in Def.~\ref{def:proj}. Therefore, $(F_0, F_1, \dots, F_{z-1})$ must induce a partition of $F$, giving us,
\begin{equation}
    \pr(\overline{F}) = \sum_{(F_0, F_1, \dots, F_{z-1}) \triangleright F }\pr(F_0, \dots, F_{z-1}).
\end{equation}
\paragraph{Upper bound on the inclusion probability:}
For any set $\overline{T} \subseteq \{( i, j_1, \dots, j_{z-1}) \mid j_1, \dots, j_{z-1} \break \in \{0, 1\} \} $, we will now upper bound $\pr(\overline{T} \subseteq \overline{F})$.
From Eq.~\eqref{eq:random-set-F}, we have
\begin{equation} \label{eq:incl-prob-ub}
   \pr(\overline{T} \subseteq \overline{F}) = \pr(X_{v'} \in \{1, 2\}, \forall v' \in \overline{T}).
\end{equation}
We will show that as claimed in Eq.~\eqref{eq:up-T-bar}
\begin{equation} \label{eq:pr-ub-T-bar}
   \pr( X_{v'} \in \{1, 2\}, \forall v' \in \overline{T}) \leq (2\overline{\delta})^{2|\overline{T}|}.
\end{equation}
\begin{figure}[!t]
\centering
\resizebox{\linewidth}{!}{
\begin{tikzpicture}[
  level distance=1.5cm,
  level 1/.style={sibling distance=8cm},
  level 2/.style={sibling distance=4cm},
  level 3/.style={sibling distance=2cm},
  every node/.style={draw,circle,minimum size=7mm},
  edge from parent/.style={draw,-}]  
\node[label={[label distance = -2pt] above:$i$}] {}
  child {node[label= {[label distance = -8pt] above:{$(i, 0)$}}][fill=red!60]{} 
    child {node[label= {[label distance = -12pt]above:{$(i, 0, 0)$}}]{ } 
      child {node[label= {[label distance = -16pt] below:{$(i, 0, 0, 0)$}}][fill=black!60] { }} 
      child {node[label= {[label distance = -16pt] below:{$(i, 0, 0, 1)$}}] { }}}
    child {node[label= {[label distance = -16pt] above:{$(i, 0, 1)$}}] { } 
      child {node[label= {[label distance = -16pt] below:{$(i, 0, 1, 0)$}}] [fill=black!60]{ }}
      child {node[label= {[label distance = -16pt] below:{$(i, 0, 1, 1)$}}] { }}}}
  child {node[label= {[label distance = -8pt] above:{$(i, 1)$}}] { } 
    child {node[label= {[label distance = -16pt] above:{$(i, 1, 0)$}}] { } 
      child {node [label= {[label distance = -16pt] below:{$(i, 1, 0, 0)$}}]{ }}
      child {node[label= {[label distance = -16pt] below:{$(i, 1, 0, 1)$}}]{ }}}
    child {node[label= {[label distance = -8pt] above:{$(i, 1, 1)$}}] { } 
      child {node[label= {[label distance = -16pt] below:{$(i, 1, 1, 0)$}}] { }}
      child {node[label= {[label distance = -16pt] below:{$(i, 1, 1, 1)$}}] { }}}};
\node[draw= none] at (-8.0,0) {$y=0$};
\node[draw= none] at (-8.0,-1.5) {$y=1$};
\node[draw= none] at (-8.0,-3) {$y=2$};
\node[draw= none] at (-8.0,-4.5) {$y=3$};
\end{tikzpicture}
}
\caption{The figure depicts a perfect binary tree of depth $z-1 = r-r'-1 = 3$ and a set of leaf nodes $\overline{T} = \{v_0 = (i, 0, 0, 0), v_1 = (i, 0, 1, 0) \}$ are colored in black. For $v_0$, the corresponding sets $S_b(v_0)$ for $b= 0, 1, 2, 3$ from Eq.~\eqref{eq:s-b} are given by $S_0(v_0) = S_1(v_0) = \{v_1\}$ and $S_2 = S_3 = \emptyset$. It follows that $f(v_0)$ from Eq.~\eqref{eq:f-v-a} is equal to $2$. The node $v_0^{(-f(v_0))}$ corresponds to $(i, 0)$, which is colored in red. Note that events $X_{v_0} \in \{1, 2\}$ and $X_{v_1} \in \{1, 2\}$ are not independent since $X_{v_0^{(-f(v_0))}} \in \{1, 2\}$ implies that $X_{v_0} = X_{v_1} = 2$.}
\label{fig:block-error-path-1}
\end{figure}

We will provide the proof of Eq.~\eqref{eq:pr-ub-T-bar} in Lemma~\ref{lem:rec} and Lemma~\ref{lem:final-upper-bound}. First, for the purpose of illustration, we will consider the following example.
\begin{example}[Inclusion probability]
Consider the perfect binary tree of depth $z-1  = 3$ and the set $\overline{T} = \{v_0 = (i, 0, 0, 0), v_1 = (i, 0, 1, 0) \}$ as in Fig.~\ref{fig:block-error-path-1}. We have $X_{v_0}, X_{v_1} \in \{1, 2\}$ if $X_{v'} = 1$ for some $v' \in \tilde{\mathscr{C}}(v_0)$ and  $X_{v'} = 1$ for some $v' \in \tilde{\mathscr{C}}(v_1)$. Note that the node $(i, 0)$ and the root node $(i)$ are elements of both $\tilde{\mathscr{C}}(v_0)$ and $\tilde{\mathscr{C}}(v_1)$. Therefore, $X_{v_0}, X_{v_1} \in \{1, 2\}$ is true if and only if one of the following events happen;
\begin{enumerate}
    \item [(1)] $X_{v'} = 1$ for some $v' \in \{(i, 0), (i)\}$.

    \item [(2)] Given that the event $ E = \{X_{v'} = 0$ for all $v' \in \{(i, 0), (i)\} \}$ happens, we have $X_{v'} = 1$ for some $v' \in \{v_0 = (i, 0, 0, 0), v_0^{(-1)} = (i, 0, 0)\}$ and $X_{v'} = 1$ for some $v' \in \{v_1 = (i, 0, 1, 0), v_1^{(-1)} = (i, 0, 1)\}$. 
\end{enumerate}
Using a union bound, the probability associated with the (1)st event is upper bounded by (here $z = 4$)
\begin{align}
  p_1 &\leq   \pr(X_{v'} = 1, v' = (i, 0)) + \pr(X_{v'} = 1, v' =  (i))  \nonumber \\
& \leq \overline{\delta}^{2^{z-1}} + \overline{\delta}^{2^{z}} \leq 2 \overline{\delta}^{2^{z-1}} \leq 2 \overline{\delta}^8.  
\end{align}
Given that the event $E$ happens,  the value of $X_{v}$ for any $v \in \{v_0, v_0^{(-1)}\}$ is independent of $X_{v'}$ for all $v' \in \{v_1 = (i, 0, 1, 0), v_1^{(-1)} = (i, 0, 1)\}$. Therefore, the probability associated with (2)nd event is given by,
\begin{align}
    p_2 &\leq \pr(X_{v'} = 1 \text{ for some } v' \in \{v_0, v_0^{(-1)}\} \mid E) \pr(X_{v'} = 1 \text{ for some } v' \in \{v_1, v_1^{(-1)}\} \mid E) \nonumber \\
    &\leq (\odelta^{2^{z-3}} + \odelta^{2^{z-2}}) (\odelta^{2^{z-3}} + \odelta^{2^{z-2}}) \nonumber \\
    & \leq (2 \odelta^2)^2
\end{align}
Therefore, we have
\begin{equation}
    \pr(X_{v_0}, X_{v_1} \in \{1, 2\}) \leq (2 \odelta^2)^2 + 2 \odelta^8 \leq 6 \odelta^{4}.
\end{equation}
Hence, the bound in Eq.~\eqref{eq:pr-ub-T-bar} is satisfied.
\end{example}
To obtain the upper bound in Eq.~\eqref{eq:pr-ub-T-bar}, we will use a recursive procedure for which it will be convenient to consider a more general set $\overline{T}$, where 
$v := (i, j_1, \dots, j_y) \in \overline{T}$
is not necessarily a leaf node, $i.e.$, $0 \leq y \leq z-1$; however, the following conditions on set $\overline{T}$ are satisfied,
\begin{equation} \label{eq:cond-not-child}
 v_1 \not\in \mathscr{C}(v_2), \text{ for any two distinct nodes $v_1, v_2 \in \overline{T}$},
\end{equation}
or equivalently, 
\begin{equation} \label{eq:cond-not-child-1}
  v_1 \not \in \tilde{\mathscr{C}}(v_2), \text{ for any two distinct nodes $v_1, v_2 \in \overline{T}$}.
\end{equation}
In other words any two distinct nodes $v_1, v_2 \in \overline{T}$ are not related by ancestor/descendant relationship. Note that a set $\overline{T}$ containing only leaf nodes (in fact any set of nodes chosen from the same depth) satisfies Eq.~\eqref{eq:cond-not-child}.

\smallskip For any $v = (i, j_1, \dots, j_{y}) \in \overline{T}$, consider the sequence $\{S_b(v)\}_{b=0}^y$, where (see also Fig.~\ref{fig:block-error-path-1})
\begin{equation} \label{eq:s-b}
   S_b(v) = \{v' \in \overline{T} \mid v' \not \in \mathscr{C}(v^{(-b)})\},
\end{equation}
where recall that $v^{(-b)} = (i, j_1, \dots, j_{y-b})$. Note that $|S_b(v)| \leq |\overline{T}| - 1, \forall 0 \leq b \leq y$ as $v\notin S_b(v)$. Moreover, for $b < b'$, we have $|S_{b'}(v)| \leq |S_b(v)|$, with $|S_y(v)| = 0$ since $v^{(-b)}$ corresponds to the root node for $b = y$. For $|\overline{T}| > 1$, we define (see also Fig.~\ref{fig:block-error-path-1})
\begin{equation} \label{eq:f-v-a}
    f(v) := \min \{b > 0: |S_b(v)| < |\overline{T}| - 1\}.
\end{equation}
In other words, $f(v)$ corresponds to the first index $b \geq 1$, where the size of $S_b(v)$ strictly decreases compared to $S_{b-1}(v)$. 

\smallskip In Lemma~\ref{lem:rec}, we upper bound the probability that $X_{v'} \in \{1, 2\}$ for all $v' \in \overline{T}$ in terms of the sum of probabilities that $X_{v'} \in \{1, 2\}$ for all $v' \in T_0$ and for all $v' \in T_1$, for some sets $T_0, T_1$ with smaller size $|T_0|, |T_1| \leq |\overline{T}| - 1$. Moreover, the condition in Eq.~\eqref{eq:cond-not-child} is again satisfied for both sets $T_0$ and $T_1$. When $|T| = 1$, Lemma~\ref{lem:rec} gives an upper bound on the probability that $X_v \in \{1, 2\}$ for the element $v \in T$.

\smallskip Since the sets $T_0$ and $T_1$ satisfy condition in Eq.~\eqref{eq:cond-not-child}, this implies that we can apply Lemma~\ref{lem:rec} on them; hence, further reducing the size. As shown in Lemma~\ref{lem:final-upper-bound}, we can apply Lemma~\ref{lem:rec} $|T|$ times to obtain an upper bound on the probability that $X_{v'} \in \{1, 2\}$ for all $v' \in \overline{T}$ in terms of the parameter $\overline{\delta}$ and a quantity called node weight of the set $\overline{T}$ (see Def.~\ref{def:nweight} below). If $\overline{T}$ only contains leaf nodes, the node weight relates with the size of the set $\overline{T}$ and this upper bound is the same as the one in Eq.~\eqref{eq:pr-ub-T-bar}. We below provide an intuitive explanation of the basic idea which is used to obtain the bound in Lemma~\ref{lem:rec}.

\smallskip Fix a node $v\in\overline T$. Moving upward from $v$, the node $v^{(-f(v))}$ corresponds to the first ancestor whose descendant subtree contains another node of $\overline{T}$. There are two possible mechanisms by which  $v \in \overline{T}$ can become erroneous.

\smallskip First, a failure may occur on one of the nodes from $v$ to the node immediately below $v^{(-f(v))}$. Since such a failure does not by itself corrupt any other node in $\overline{T}$, we pay this failure probability and remove $v$ from the set, giving $\overline T^0$.

\smallskip Second, the corruption of $v$ may originate at or above $v^{(-f(v))}$. Such a failure would simultaneously corrupt every node in $\overline{T}$ that is in the subtree rooted at  $v^{(-f(v))}$. We may therefore replace all these nodes in $\overline{T}$ by $v^{(-f(v))}$, giving $\overline T^1$.

\smallskip  In both cases the resulting set contains strictly fewer nodes and continues to satisfy the condition in Eq.~\eqref{eq:cond-not-child} (or equivalently, Eq.~\eqref{eq:cond-not-child-1}). This makes the decomposition recursive.
\begin{lemma} \label{lem:rec}
Consider an arbitrary set $\overline{T}\subseteq \{( i, j_1, \dots, j_y) \mid j_1, \dots, j_y \in \{0, 1\}, 0 \leq y \leq z-1\}$, whose elements satisfy Eq.~\eqref{eq:cond-not-child}. For $|\overline{T}| > 1$ and any fixed $v = (i, j_1, \dots, j_{y}) \in \overline{T}$, consider the node $v^{(-f(v))} = (i, j_1, \dots, j_{y - f(v)}) \in \tilde{\mathscr{C}}(v)$, where $f(v)$ is according to  Eq.~\eqref{eq:f-v-a}. For $v$ and $v^{(-f(v))}$, we define sets $\ot_{0}$ and $\ot_{1}$ as follows:\footnote{Note that $|\ot_0|, |\ot_1| \leq |\ot| - 1$.} 
\begin{align}
 \ot_{0} &:= \overline{T} \setminus \{ v \}. \label{eq:T-0}\\
 \ot_{1} &:= \{ v' \in \overline{T} \mid v' \not \in \mathscr{C}(v^{(-f(v))})\} \cup \{v^{(-f(v))}\}. \label{eq:T-1}
\end{align}
Let $S$ be another set of nodes on the perfect binary tree of depth $z-1$ and define the event $E := \{X_{v'} = 0 \mid v' \in S\}$.
Then, we have 
\begin{align} 
    \pr(X_{v'} \in \{1, 2\}, \forall v' \in \overline{T} \mid E) &\leq 2 {\odelta}^{2^{z - y}}  \pr\big( X_{v'} \in \{1, 2\}, \forall v' \in \overline{T}_0 \mid E' \big) \nonumber\\ 
    & \qquad + \pr(X_{v'} \in \{1, 2\}, \forall v' \in \overline{T}_1 \mid E), \label{eq:rec-inc}
\end{align}
where $E' = \{ X_{v'} = 0, v' \in S' \}$, with $S' \supseteq S$ being a specific extension of $S$, defined using the ancestors of the fixed node $v$. 
Moreover, if $|\overline{T} | = 1$, we have the following upper bound
\begin{equation} \label{eq:size-T-one}
    \pr(X_{v'} \in \{1, 2\}, \forall v' \in \overline{T} \mid E) \leq 2 {\odelta}^{2^{z - y}}.
\end{equation}
\end{lemma}
\begin{proof}
Consider the fixed node $v = (i, j_1, \dots, j_{y})\in \overline{T}$ and note that if there exists a $\mathscr{C}(v) \cap S \neq \emptyset$, then conditioned on $E$, we have $X_v = 0$, with probability one, implying that $\pr(X_{v'} \in \{1, 2\}, \forall v' \in \overline{T} \mid E) = 0$. Hence, the bound in Eq.~\eqref{eq:rec-inc} is satisfied trivially. In the following, we will prove that Eq.~\eqref{eq:rec-inc} also holds for $\mathscr{C}(v) \cap S = \emptyset$.

\smallskip Using a union bound, we have the following for the fixed node $v = (i, j_1, \dots, j_{y})\in \overline{T}$
\begin{align}
    &\pr(X_{v'}  \in \{1, 2\},  \forall v' \in \overline{T} \mid E) \leq     \pr(X_v = 1, X_{v'} \in \{1, 2\}, \forall v' \in \ot_0 \mid E) \nonumber\\ 
    & \quad \quad \quad \quad \quad \quad \quad \quad  \quad \quad \quad \quad \quad + \pr(X_v = 2, X_{v'} \in \{1, 2\}, \forall v' \in \ot_0 \mid E) \label{eq:ub-step-1}
\end{align}
We upper bound the first term on the right hand side of the inequality in Eq.~\eqref{eq:ub-step-1} as
\begin{align}
&\pr(X_v = 1, X_{v'} \in \{1, 2\}, \forall v' \in \ot_0 \mid E) = \pr(X_v = 1 \mid E) \pr(X_{v'} \in \{1, 2\}, \forall v' \in \ot_0 \mid E, X_v = 1)   \nonumber \\
& \leq {\odelta}^{2^{z-y}} \pr(X_{v'} \in \{1, 2\}, \forall v' \in \ot_0 \mid E, X_v = 1, X_{v'} = 0, \forall v' \in \tilde{\mathscr{C}}(v)\setminus  \{v\} ) \nonumber \\
 &= {\odelta}^{2^{z - y}} \pr(X_{v'} \in \{1, 2\}, \forall v' \in \ot_0 \mid E, X_{v'} = 0, \forall v' \in \tilde{\mathscr{C}}(v)\setminus \{v\} ) \nonumber \\
& = {\odelta}^{2^{z-y}} \pr(X_{v'} \in \{1, 2\}, \forall v' \in \ot_0 \mid E^{(0)}), \label{eq:ub-step21}
\end{align}
where the first equality uses the Bayes' theorem, the first inequality uses the upper bound on $\pr(X_v = 1 \mid E)$ from Eq.~\eqref{eq:up-xvone} and the fact that $X_v = 1$ implies $X_{v'} = 0, \forall v' \in \tilde{\mathscr{C}}(v)\setminus \{v\}$. In the second equality, we remove $X_v=1$ from the conditioning. Indeed, conditioned on the ancestors of $v$ having value zero, the failure event at $v$ is independent of the variables $X_{v'}$ for $v'\in\overline T^0$. This follows from the independent selection of failures in Eq.~\eqref{eq:pr-f-y} and the fact that, by
Eq.~\eqref{eq:cond-not-child}, no $v'\in\overline T^0$ is a descendant of $v$.
Finally in the last equality, we have
\begin{equation}
E^{(0)} := \{X_{v'} = 0 \mid v' \in S \cup (\tilde{\mathscr{C}}(v) \setminus \{v\})  \}.  
\end{equation}
Recall that $S$ is the original set that is used to define $E$.

\smallskip We will now upper bound the second term on the right hand side of Eq.~\eqref{eq:ub-step-1}. We consider $f(v)$ corresponding to the fixed node $v \in \overline{T}$ defined in Eq.~\eqref{eq:f-v-a} and using it define the following set which corresponds to ancestors of $v$ between $v = (i, j_1, \dots, j_y)$ and $v^{(-f(v))} = (i, j_1, \dots, j_{y-f(v)})$ 
\begin{align}
\tilde{\mathscr{C}}_0(v) &:= \{(i, j_1, \dots, j_{y-f(v) + 1}), \dots, (i, j_1, \dots, j_{y-1})\} \subseteq \tilde{\mathscr{C}}(v) \setminus \{v\}
\end{align}
Note that $X_{v} = 2$ holds if and only if $X_{(v^{(-b)})} = 1$ for a $b$, $1 \leq b \leq y$.  Therefore, using a union bound, we have
\begin{align}
    \pr(X_v = 2, X_{v'} \in \{1, 2\}, \forall v' \in \ot_0 \mid E) &\leq \sum_{\overline{v} \in \tilde{\mathscr{C}}_0(v)}   \pr(X_{\overline{v}} = 1, X_{v'} \in \{1, 2\}, \forall v' \in \ot_0 \mid E)   \nonumber \\
    & \qquad + \pr(X_{v^{(-f(v))}} \in \{1, 2\}, X_{v'} \in \{1, 2\}, \forall v' \in \ot_0 \mid E)  \label{eq:ub-step22}
\end{align}
For any $v^{(-b)} = (i, j_1, \dots, j_{y-b}) \in \tilde{\mathscr{C}}_0(v)$, where $1 \leq b \leq f(v)-1$, note that the condition in Eq.~\eqref{eq:cond-not-child} is satisfied for the set $\ot_0 \cup \{v^{(-b)}\}$. 
Hence, we can upper bound similarly to Eq.~\eqref{eq:ub-step-1} for any $\overline{v}_b = (i, j_1, \dots, j_{y-b}) \in \mathscr{C}_0(v)$ (here $b \leq f(v)-1$)
\begin{equation} \label{eq:ub-step23}
    \pr(X_{\overline{v}} = 1, X_{v'} \in \{1, 2\}, \forall v' \in \ot_0 \mid E) \leq {\odelta}^{2^{z-y + b}} \pr(X_{v'} \in \{1, 2\}, \forall v' \in \ot_0 \mid E^{(b)}),
\end{equation}
where
\begin{equation}
E^{(b)} := \{X_{v'} = 0 \mid v' \in S \cup (\tilde{\mathscr{C}}(v_b) \setminus \{v_b\}) \}.  
\end{equation}
We take the set $S'$ to be $S \cup \tilde{\mathscr{C}}(v_{b_0}) \setminus \{v_{b_0}\}$, for $b_0$, $0 \leq b_0 \leq f(v)-1$, such the probability $\pr(X_{v'} \in \{1, 2\}, \forall v' \in \ot_0 \mid E^{(b)})$, with $0 \leq b \leq f(v)-1$, is largest for $b = b_0$ . Then, we have $E' = E^{(b_0)}$.

\smallskip For $v^{(-f(v))}$, we have that $\mathscr{C}(v^{(-f(v))}) \cap \ot_0 \neq \emptyset$; therefore, the bound in Eq.~\eqref{eq:ub-step23} does not necessarily hold. However, using the fact that $X_{v^{(-f(v))}} \in \{1, 2\}$ implies $X_{v'} = 2$ for all $v' \in \mathscr{C}(v^{(-f(v))}) \setminus \{v^{(-f(v))}\}$, we have 
\begin{align}
    &\pr(X_{v^{(-f(v))}} \in \{1, 2\}, X_{v'} \in \{1, 2\}, \forall v' \in \ot_0 \mid E) \nonumber \\
    & \qquad= \pr(X_{v^{(-f(v))}} \in \{1, 2\}, X_{v'} \in \{1, 2\}, \forall v' \in \ot_0 \setminus (\mathscr{C}(v^{(-f(v))}) \cap \ot_0) \mid E) \nonumber \\
    & \qquad = \pr(X_{v'} \in \{1, 2\}, \forall v' \in \ot_1 \mid E), \label{eq:ub-step24}
\end{align}
where for the last equality note that $\ot_1 = \left(\ot_0 \setminus (\mathscr{C}(v^{(-f(v))}) \cap \ot_0)\right)  \cup \{v^{(-f(v))}\}$.

\smallskip Finally, from Eq.~\eqref{eq:ub-step-1}, Eq.~\eqref{eq:ub-step21}, Eq.~\eqref{eq:ub-step22}, Eq.~\eqref{eq:ub-step23}, and Eq.~\eqref{eq:ub-step24}, and using ${\odelta}^{2^{z-y}} + {\odelta}^{2^{z-y+ 1}}  + {\odelta}^{2^{z-y+2}} + \dots \leq 2 {\odelta}^{2^{z-y}}$ (Recall that $\odelta \leq \frac{1}{2}$), we get Eq.~\eqref{eq:rec-inc}.

\smallskip For the last claim about $|\overline{T}| = 1$, applying a union bound, we have for $v =  (i, j_1, \dots, j_y) \in \overline{T}$ 
\begin{align}
    \pr(X_v \in \{1, 2\} \mid E) &\leq \sum_{\overline{v} \in \tilde{\mathscr{C}}(v)}   \pr(X_{\overline{v}} = 1 \mid E) \nonumber \\
    & \leq {\odelta}^{2^{z-y}} + {\odelta}^{2^{z-y+ 1}}  + \dots +  {\odelta}^{2^{z}}  \nonumber \\
    &\leq 2 {\odelta}^{2^{z-y}}
\end{align}
\end{proof}
We will now apply Lemma~\ref{lem:rec} in order to obtain an upper bound on $\pr( X_{v'} \in \{1, 2\}, \forall v' \in T)$ for any $T$ whose elements satisfy Eq.~\eqref{eq:cond-not-child}. Before we do this, we introduce a quantity, namely the ``node weight", of the set $T$ as follows.
\begin{definition}[Node weight] \label{def:nweight}
  For any set $\overline{T}$ of nodes on the perfect binary tree of depth $z-1$, it's node weight is defined as the number of leaf nodes that have an ancestor in $\overline{T}$, that is,
  \begin{equation}
      W(\ot) = |\{ v := (i, j_1, \dots, j_{z-1})    \mid \tilde{\mathscr{C}}(v) \cap \ot \neq \emptyset \}|
  \end{equation} 
\end{definition}

\begin{remark} \label{rem:node-weight}
   Note that for $\ot$ with only one element $v$, the corresponding node weight $W(\ot)$, which we also denote as $W(v)$, is given by $2^{(z - 1- y(v))}$.  Consider a set $\ot$, satisfying Eq.~\eqref{eq:cond-not-child}. Note that for any two distinct $v_1, v_2 \in \ot$, we have $\mathscr{C}(v_1) \cap \mathscr{C}(v_2) = \emptyset$, since otherwise either $v_1 \in \tilde{\mathscr{C}}(v_2)$ or $v_2 \in \tilde{\mathscr{C}}(v_1)$. Therefore, any leaf node can have only one ancestor in $\ot$. This implies that
\begin{equation}
      W(\ot) = \sum_{v \in \ot} W(v) = \sum_{v \in \ot} 2^{(z - 1- y(v))}
\end{equation}
If $\ot$ is a set of leaf nodes, we simply have $W(\ot) = |\ot|$.
\end{remark}

\smallskip In Lemma~\ref{lem:final-upper-bound} below, we show using Lemma~\ref{lem:rec} that $\pr( X_{v} \in \{1, 2\}, \forall v \in \overline{T})$ decreases exponentially in terms of its node weight $W(\overline{T})$. When the set $\overline{T}$ corresponds to the set of leaf nodes, we have $W(\overline{T}) = |\overline{T}|$,  giving us the desired bound in Eq.~\eqref{eq:cons-err-rate}. 
\begin{lemma} \label{lem:final-upper-bound}
Let $\overline{T}$ be a set of nodes on the perfect binary tree of depth $z-1$, satisfying Eq.~\eqref{eq:cond-not-child}. Then, we have
\begin{equation} \label{eq:inc-prob-1}
   \pr( X_{v} \in \{1, 2\}, \forall v \in \overline{T}) \leq 4^{|\overline{T}|} \:{\odelta}^{2W(\overline{T})},
\end{equation}
In particular, if $\overline{T}$ is a set of leaf nodes, we have 
\begin{equation} \label{eq:inc-prob-2}
   \pr( X_{v} \in \{1, 2\}, \forall v \in \overline{T}) \leq  (2\odelta)^{2|\overline{T}|}.
\end{equation}
\end{lemma}

\begin{proof}
For the sake of clarity, we shall use the notation $\delta' := \odelta^2$.

\smallskip If $|\overline{T}| = 1$, we have from Lemma~\ref{lem:rec} that $\pr(X_v \in \{1, 2\}, v \in \overline{T} \mid E ) \leq 2\odelta^{2^{z - y(v)}} \leq 2 {\delta'}^{z-1-y(v)} = 2{\delta'}^{W(\overline{T})}$. Therefore, Eq.~\eqref{eq:inc-prob-1} is satisfied.

\smallskip For $|\overline{T}| > 1$, we will use Lemma~\ref{lem:rec}. Let $S$ be a set and use it to define an event $E := \{X_{v'} = 0 \mid v' \in S\}$. We may start with $S = \emptyset$. We define $\zeta(\overline{T}, E) := \pr(X_{v'} \in \{1, 2\}, \forall v' \in \overline{T} \mid E)$ and by choosing a fixed $v \in \overline{T}$ and applying Lemma~\ref{lem:rec}, we have
\begin{equation} \label{eq:zeta-rec}
    \zeta(\overline{T}, E) \leq 2 {\delta'}^{W(v)}    \zeta(\ot_0, E') +  \zeta(\ot_1, E),
\end{equation}
where $\delta' = {\odelta}^2$, $W(v)$ is the node weight of $v$. The sets $\ot_0$ and $\ot_1$ are given by (see Eq.~\eqref{eq:T-0} and Eq.~\eqref{eq:T-1}), respectively. 
\begin{align}
 \ot_{0} & = \overline{T} \setminus \{ v \}.\\
 \ot_{1} & = \{ v' \in \overline{T} \mid v' \not \in \mathscr{C}(v^{(-f(v))})\} \cup \{v^{(-f(v))}\}, 
\end{align}
where $f(v)$ is according to Eq.~\eqref{eq:f-v-a}. $E' = \{ X_{v'} = 0, v' \in S' \}$, with $S' \supseteq S$ being a set of nodes such that $\overline{T}_0 \cap S' = \emptyset$.

\smallskip The node $v$ is not included in the set $\ot_0$; however, it has been compensated by the corresponding multiplicative factor $2{\delta'}^{W(v)}$ in the first term on the right hand side of Eq.~\eqref{eq:zeta-rec}. Note also that 
\begin{equation} \label{eq:node-w-split}
  W(\ot_0) = W(\ot)  - W(v)  
\end{equation}

\smallskip Although the size of the set $\ot_1$ is smaller than $\ot$, the set $\ot_1$ effectively contains all the nodes in $\ot$: namely, if $v' \in \ot$ but $v' \not\in \ot_1$, $v'$ must be a descendant of $v^{(-f(v))}$, which is contained in $\ot_1$. Therefore, the node $v'$ has been absorbed in $v^{(-f(v))}$ rather than being completely removed. Moreover, it can be seen as follows that the node weight of the set $\ot_1$ is at least the node weight of $\ot$:

\smallskip Using the fact that both sets $\ot$ and $\ot_1$ satisfy Eq.~\eqref{eq:cond-not-child}, and by Remark~\ref{rem:node-weight}, we have
\begin{align}
    W(\ot_1) &=  2^{z - 1- y(v^{(-f(v))})}   + \sum_{v' \in \overline{T}, v' \not\in \mathscr{C}(v^{(-f(v))})} 2^{(z -1 - y(v'))} \nonumber\\
    &=  2^{z - 1 - y(v^{(-f(v))})} 
     + \sum_{v' \in \overline{T}} 2^{(z - 1 - y(v'))} -   \sum_{v' \in \ot \cap \mathscr{C}(v^{(-f(v))})} 2^{( z - 1 - y(v'))} \nonumber\\
     & \geq  W(\overline{T}), \label{eq:node-w-merge},
\end{align}
where the inequality in the last line uses that $2^{z - 1 - y(v^{(-f(v))})} \geq  \sum_{v' \in \ot \cap \mathscr{C}(v^{(-f(v))})} 2^{( z - 1 - y(v'))}$ using the fact that $\overline{T}$ satisfies Eq.~\eqref{eq:cond-not-child}.

\smallskip We may again apply Lemma~\ref{lem:rec} to upper bound $\zeta(\ot_{u_1}, E')$ and $\zeta(\ot_{u_1}, E)$, where $u_1 \in \{0, 1\}$, by choosing some $v \in \ot_{u_1}$.

\smallskip If $|\ot_{u_1}| > 1$, we apply Lemma~\ref{lem:rec} in order to split it into smaller sets $\ot_{u_1 0}$ and $\ot_{u_1 1}$. 
We continue this procedure recursively on every set having more than one element, and stop a branch once the corresponding set has size one. Since each application of Lemma~\ref{lem:rec} strictly decreases the size of the set, every branch terminates after at most $|T| -1$ splittings. For $1 \leq s\leq |T|-1$, let $B_s \subseteq \{0, 1\}^s$ denote the set of binary strings corresponding to branches that terminate after exactly $s$ splittings.

\smallskip Let $a \equiv a(u_{ 1}, u_{ 2}, \dots, u_{ s})$ denote the number of zeros in $(u_{ 1}, u_{ 2}, \dots, u_{ s})$. Then, we have the following upper bound after $t = |T| -1$ iterations of  Eq.~\eqref{eq:zeta-rec},
\begin{align}
    \zeta(\overline{T}, E) &\leq \sum_{s \leq t}  \: \sum_{(u_{ 1}, u_{ 2}, \dots, u_{ s}) \in B_s} 2^{a(u_{ 1}, u_{ 2}, \dots, u_{ s})} {\delta'}^{\sum_{s' = 0}^{s-1} (1 \oplus u_{ s'+1} ) W(v_{s'})} \zeta(\ot_{u_{ 1}, u_{ 2}, \dots, u_{ s}}, E^{\prime\prime \stackrel{\: a \times }{\cdots}\prime}) \label{eq:zeta-rec-1}
\end{align}
where $v_{s'}$ corresponds to the chosen node $v_{s'} \in \ot_{u_{ 1}, u_{ 2}, \dots, u_{s'}}$ with respect to which Lemma~\ref{lem:rec} is applied. Note that in the exponent of $\delta'$, the node weights of $v_{s'}$ are counted only when the value of $u_{ s'+1}$ is zero. Here, $E^{\prime\prime \stackrel{\: a \times }{\cdots}\prime} = \{ X_{v'} = 0 \mid v' \in S^{\prime\prime \stackrel{\: a \times }{\cdots}\prime} \}$ for some set $S^{\prime\prime \stackrel{\: a \times }{\cdots}\prime}$ such that $\ot_{u_{ 1}, u_{ 2}, \dots, u_{ s-1}} \cap S^{\prime\prime \stackrel{\: a \times }{\cdots}\prime} = \emptyset$. 
Applying Eq.~\eqref{eq:node-w-split} and Eq.~\eqref{eq:node-w-merge} iteratively, we have
\begin{equation} \label{eq:martingale}
   W(\ot_{u_{1}, u_{2}, \dots, u_{s}})  +  \sum_{s' = 0}^{s-1} (1 \oplus u_{s'+1} ) W(v_{s'}) \geq  W(\overline{T}), \: \forall (u_{1}, u_{2}, \dots, u_{s}) \in B_s
\end{equation}
Using $|\ot_{u_{1}, u_{2}, \dots, u_{s}}| = 1$ and Lemma~\ref{lem:rec}, we can directly bound 
$$\zeta(\ot_{u_{1}, u_{2}, \dots, u_{s}}, E^{\prime\prime \stackrel{\: a \times }{\cdots}\prime}) \leq 2 {\delta'}^{W(\ot_{u_{1}, u_{2}, \dots, u_{s}})}  \text{ for any } (u_{1}, u_{2}, \dots, u_{s}) \in B_s.$$ Therefore, after $t = |\overline{T}|$ iterations, we have
\begin{align}
   \zeta(\overline{T}, E) &\leq \sum_{s \leq |\overline{T}|-1} \: \sum_{(u_{1}, u_{2}, \dots, u_{s}) \in B_s} 2^{a(u_{1}, u_{2}, \dots, u_{s}) + 1} {\delta'}^{ W(\ot_{u_{1}, u_{2}, \dots, u_{s}}) + \sum_{s' = 0}^{s-1} (1 \oplus u_{s'+1} ) W(v_{s'}) }  \nonumber \\
&   \leq  \sum_{s \leq |\overline{T}| -1} 2^{2s + 1} {\delta'}^{W(\overline{T})} \nonumber \\
& \leq  4^{|\overline{T}|}  {\delta'}^{W(\overline{T})},
\end{align}
where the second inequality uses Eq.~\eqref{eq:martingale}, and $a(u_{1}, u_{2}, \dots, u_{s}) \leq s$ and $|B_s| \leq 2^s$, and the last inequality uses a bound on sum of geometric series.

\smallskip Finally, if $\overline{T}$ is a set of leaf nodes, we have $W(\overline{T}) = |\overline{T}|$; therefore,
\begin{equation}
    \zeta(\overline{T}) \leq  (4\delta')^{|\overline{T}|} = (2 \odelta)^{2 |\overline{T}|},
\end{equation}
where the last equality uses $\delta' = \odelta^2$.
\end{proof}

\section{Fault-tolerant quantum state preparation with constant qubit overhead}\label{sec:FT-state-prep}
In this section, we will show how to fault-tolerantly realize any state preparation circuit, i.e., a circuit with no input and quantum output, with constant qubit overhead. In particular, using earlier work on constant qubit overhead fault-tolerant quantum computing~\cite{gottesman2013fault, fawzi2020constant}, we first realize the state preparation circuit in the code space of a QLDPC code, and then use our decoding interface from Theorem~\ref{thm:ft-cons-int-main} to get out of the code space. Since fault-tolerant quantum computation is realized in multiple code blocks of a QLDPC code, we apply our decoding interface circuit for multiple code blocks to decode them, ensuring that the overhead in state preparation remains constant (see also Fig.~\ref{fig:ft-state_prep}). 
\begin{figure}[!b]
    \centering
\begin{tikzpicture}[x=1cm,y=1cm, line cap=round, line join=round]

  \def\xL{0}
  \def\xA{1.2}   
  \def\xB{3.2}   
  \def\xC{5.2}   
  \def\xD{7.2}   
  \def\xR{8.6}

  \def\yOne{0.90}
  \def\yTwo{0.50}
  \def\yThree{-0.75}


  \draw (\xA,-1.10) rectangle (\xB,1.10);
  \node at ({(\xA+\xB)/2},0) {$\Phi^r_{\mathrm{FT}}$};

  \draw (\xB,\yOne)   -- (\xC,\yOne);
  \draw (\xB,\yTwo)   -- (\xC,\yTwo);
  \draw (\xB,\yThree) -- (\xC,\yThree);

\foreach \t in {0.2,0.5,0.8} {
  \fill ({(\xB+\xC)/2}, {(1-\t)*(\yTwo-0.12) + \t*(\yThree+0.12)}) circle (0.6pt);
}
\node at ({(\xB+\xC)/2+0.25}, {(\yTwo+\yThree)/2}) {$h$};

\def\xDot{\xD+0.7} 

\foreach \t in {0.2,0.5,0.8} {
  \fill (\xDot, {(1-\t)*\yTwo + \t*\yThree}) circle (0.6pt);
}

\foreach \yy in {\yOne,\yTwo,\yThree} {
  \draw ({(\xB+\xC)/2-0.03},\yy-0.06) -- ({(\xB+\xC)/2+0.09},\yy+0.14);
}
  \node[above] at ({(\xB+\xC)/2+0.2},\yOne) {$r$};
  \node[above] at ({(\xB+\xC)/2+0.2},\yTwo) {$r$};
\node[above] at ({(\xB+\xC)/2+0.2},\yThree) {$r$};

  \draw (\xC,-1.10) rectangle (\xD,1.10);
  \node at ({(\xC+\xD)/2},0) {$\Xi^{[h]}_r$};

  \draw (\xD,\yOne)   -- (\xR,\yOne);
  \draw (\xD,\yTwo)   -- (\xR,\yTwo);
  \draw (\xD,\yThree) -- (\xR,\yThree);

\end{tikzpicture}
\caption{Fault-tolerant realization of a state-preparation circuit $\Phi$ with constant qubit overhead. The circuit
$\Phi_{\mathrm{FT}}^r$ first realizes $\Phi$ in the code space of $h$ blocks of $\mathcal C_r$. The decoding interface $\Xi_r^{[h]}$ then maps the encoded output to bare physical qubits. The resulting state agrees with the output of $\Phi$ up to a local
stochastic channel.}
\label{fig:ft-state_prep}
\end{figure}

\subsection{Construction and main statement} \label{sec:gottesman-cons}
In Theorem~\ref{thm:gottesman-cons} below, we first state the main result of~\cite{gottesman2013fault, fawzi2020constant}. We note that Theorem~\ref{thm:gottesman-cons} is stated in a slightly different language than in~\cite{gottesman2013fault, fawzi2020constant}. In particular, we state it as a state preparation protocol in the code space of a QLDPC code instead of as computation with classical input and output. However, we emphasize that this is directly implied by the results of~\cite{gottesman2013fault, fawzi2020constant} since their protocol involves preparing quantum states in the code space of a QLDPC code and then measuring it in the computational basis. 
\begin{theorem} \label{thm:gottesman-cons}
There exists a threshold value $\delta_{th} > 0$, such that the following holds:

\smallskip Consider the QLDPC code family $\{ \cC_r \mid r = 1, 2, \dots \}$ with constant rate $\alpha >0$. Let $\Phi$ be a state preparation circuit, with $x$ qubit output and operating on $O(x)$ qubits and of size $|\Phi| = \mathrm{poly}(x)$. Then, there exists an encoding level $r$ (depending on $x$) and a fault-tolerant state preparation circuit $\Phi^r_{\mathrm{FT}}$, which operates on $O(x)$ qubits, such that for the channel $\tilde{\cT}_{\Phi_{\mathrm{FT}}}$ corresponding to the noisy realization of $\Phi_{\mathrm{FT}}$ under  circuit-level stochastic noise with parameter $\delta < \delta_{th}$, we have
\begin{equation} \label{eq:gottesman-cons}
\dnorm{\tilde{\cT}_{\Phi^r_{\mathrm{FT}}} -  \cN \circ \cE_r^{\otimes \frac{x}{m_r}} \circ \cT_{\Phi}} \leq \epsilon(x),
\end{equation}
where $\cN$ is a channel with stabilizer-reduced weight $\mu n_r$ with respect to each block of $\cC_r$, and $\epsilon(x) \to 0$ as $x \to \infty$. Moreover, the value of $r$ may be chosen such that $m_r = \Theta \left( x^{\frac{1}{l}} \right)$, $l > 1$ is an arbitrary fixed integer, i.e., not depending on $r$.\footnote{Here $\Theta(\cdot)$ is the usual big-Theta asymptotic scaling, that is, $f(n) = \Theta (g(n))$ iff $ c_1 g(n) \leq f(n) \leq c_2 g(n)$ for some constants $c_1, c_2 > 0$ and sufficiently large $n$.}
\end{theorem}
We note that an analogous version of Theorem~\ref{thm:gottesman-cons} has been proven under a general noise model that includes non-stochastic noise models such as coherent and amplitude damping noise~\cite{christandl2025fault}.

\begin{remark}\label{rem:number-of-blocks}
For the constructions presented in~\cite{gottesman2013fault,fawzi2020constant}, the parameter $\epsilon(x)$ vanishes polynomially with respect to $x$. From Eq.~\eqref{eq:gottesman-cons} and Fig.~\ref{fig:ft-state_prep}, we note that the noisy realization of the fault-tolerant circuit $\Phi^r_{\mathrm{FT}}$ realizes the state preparation circuit $\Phi$ in multiple code blocks of $\cC_r$ (up to a local stochastic channel applied on them). In particular, the number of blocks is given by $h \leq \frac{x}{m_r}$. Using Theorem~\ref{thm:gottesman-cons}, $x = \Theta \left(m_r^{l}\right)$ for $l  > 1$, we have that $h = \Theta\left(m_r^{l-1}\right)$. By choosing the integer $l$ to be large enough, we can ensure that $h \geq p(\cdot)$, as required in Theorem~\ref{thm:main-stprep} for constant qubit overhead. 
\end{remark}
In Theorem~\ref{thm:main-stprep}, we construct a fault-tolerant circuit $\overline{\Phi}$ for any state preparation circuit $\Phi$, such that $\overline{\Phi}$ has the same input and output systems as $\Phi$, and simulates $\Phi$ up to a local stochastic noise on the output.
\begin{theorem} \label{thm:main-stprep}
There exists a threshold value $\delta_{th} > 0$, constants $\kappa_1, \kappa_2 > 0$, such that the following holds:

\smallskip  Consider a state preparation circuit $\Phi$, with $x$ qubit output and operating on $O(x)$ qubits and  having size $|\Phi| = \mathrm{poly}(x)$. Then, there exists a quantum circuit $\overline{\Phi}$, with the same input and output systems as $\Phi$ and working on $O(x)$ qubits, such that its noisy realization $\tilde{\cT}_{\overline{\Phi}}$, under circuit level stochastic noise with parameter $\delta < \delta_{th}$ satisfies the following 
\begin{equation} \label{eq:main-stprep}
 \dnorm{\tilde{\cT}_{\overline{\Phi}} - \cV \circ \cT_{\Phi}} \leq \epsilon(x), 
\end{equation}
where $\cV$ is a local stochastic channel with parameter $(\kappa_1\delta)^{\kappa_2}$. The error parameter $\epsilon(x) \to 0$ as $x \to \infty$.\footnote{$\epsilon(x)$ is the same as in Theorem~\ref{thm:gottesman-cons}.} 
\end{theorem}

\begin{proof}
The quantum circuit $\overline{\Phi}$ is constructed as follows: we first realize the quantum circuit $\Phi$ in the code space of $h = \frac{x}{m_r}$ copies of $\cC_r$ using Theorem~\ref{thm:gottesman-cons}, and then apply the interface circuit $\Xi^{[h]}_{r}$ from Theorem~\ref{thm:ft-cons-int-main}. Therefore,
\begin{equation}
\overline{\Phi} :=  \Xi^{[h]}_{r} \circ \: \Phi^r_{\mathrm{FT}}
\end{equation}
From Theorem~\ref{thm:gottesman-cons}, $\Phi^r_{\mathrm{FT}}$ operates on $O(x')$ qubits. As discussed in Remark~\ref{rem:number-of-blocks}, we can chose $r$ in a way that $h \geq p(\cdot)$ as required for constant qubit overhead in Theorem~\ref{thm:ft-cons-int-main}; therefore, $\Xi^{[h]}_{r}$ operates on $O(m_r h) = O(x)$ qubits. Hence, $\overline{\Phi}$ operates on $O(x)$ qubits.

\smallskip We will now show that Eq.~\eqref{eq:main-stprep} holds. From Theorem~\ref{thm:gottesman-cons}, we have
\begin{equation} \label{eq:phi-r-FT}
\dnorm{\tilde{\cT}_{\Phi^r_{\mathrm{FT}}} -  \cN \circ \cE_r^{\otimes h} \circ \cT_{\Phi}} \leq \epsilon(x),
\end{equation}

\smallskip For $\Xi^{[h]}_{r}$ from Theorem~\ref{thm:ft-cons-int-main}, we have Eq.~\eqref{eq:interface-xi-r},  
\begin{equation} \label{eq:xi-r}
    \tilde{\cT}_{\Xi^{[h]}_r}  \circ \cN \circ \cE_r^{\otimes h}  =  \cV ,  
\end{equation}
where $\cV$ is a local stochastic channel with parameter $(\kappa_1\delta)^{\kappa_2}$ for some constants $\kappa_1, \kappa_2 > 0$. 

\smallskip Finally, we have 
\begin{align}
    \dnorm{\tilde{\cT}_{\overline{\Phi}} - \cV \circ \cT_{\Phi}} &= \dnorm{\tilde{\cT}_{\Xi^{[h]}_r} \circ \tilde{\cT}_{\Phi^r_{\mathrm{FT}}} - \cV \circ \cT_{\Phi}} \nonumber\\
    & = \dnorm{\tilde{\cT}_{\Xi^{[h]}_r} \circ \tilde{\cT}_{\Phi^r_{\mathrm{FT}}} -\tilde{\cT}_{\Xi^{[h]}_r} \circ \cN \circ \cE_r^{\otimes h} \circ \cT_\Phi  + \tilde{\cT}_{\Xi_r}  \circ \cN \circ \cE_r^{\otimes h} \circ \cT_\Phi -  \cV \circ \cT_{\Phi}} \nonumber \\
    & \leq \dnorm{\tilde{\cT}_{\Xi^{[h]}_r}} \dnorm{\tilde{\cT}_{\Phi^r_{\mathrm{FT}}} -  \cN \circ \cE_r^{\otimes h} \circ \cT_{\Phi}} \nonumber \\
    & \leq \epsilon(x) ,
\end{align}
where in the first inequality uses Eq.~\eqref{eq:xi-r} and the sub-multiplicativity of the diamond norm: $\dnorm{\cW_1 \circ \cW_2} \leq \dnorm{\cW_1} \dnorm{\cW_2}$. For the inequality in the last line, we have used Eq.~\eqref{eq:phi-r-FT}, and the fact that $\dnorm{\cW}= 1$ for any quantum channel $\cW$.
\end{proof}

\subsection{Applications} \label{sec:apps} 
We below remark on applications of state preparation with constant qubit overhead to fault-tolerant quantum computation and communication.

\paragraph{Fault-tolerant quantum computation:} The state preparation method from Theorem~\ref{thm:main-stprep} has the following two applications in fault-tolerant quantum computation: 

\begin{enumerate}
\item[(i)] Gate teleportation is a widely used scheme in fault-tolerant quantum computation to realize universal logic gates in the code space of error correcting codes~\cite{gottesman1999demonstrating, zhou2000methodology, nielsen2003quantum,  raussendorf2001one, bravyi2005universal, gottesman2013fault}. However, gate teleportation requires preparation of suitable resource states. Usually this state preparation is costly and incurs a significant overhead that grows as the target infidelity decreases. Since for larger computation, we will need a state with smaller infidelity, the overhead grows with the size of computation. 

\smallskip For example, magic state distillation is a standard method of preparing high fidelity eigenstates of non-Clifford gates in code space of an error correcting code~\cite{bravyi2005universal}. It takes many copies of a noisy quantum state and distills fewer quantum states with high fidelity. Magic state distillation along with Clifford gate gives a universal set of logic gates. However, magic distillation traditionally incurs a qubit overhead that grows polylogarithmically with respect to the inverse of target infidelity. It was a long standing question whether magic state distillation with constant qubit overhead is possible~\cite{bravyi2012magic, haah2018codes, meier2012magic, campbell2012magic, jones2013multilevel, hastings2018distillation, krishna2018towards}. This has been recently resolved affirmatively in~\cite{wills2025constant}. 

\smallskip Our state preparation method from Theorem~\ref{thm:main-stprep} can be used for preparing resource states for gate teleportation  with low overhead; providing an alternative to magic state distillation. 

\item[(ii)] Our state-preparation method yields a direct route to achieve fault-tolerant quantum computation with constant qubit overhead. Concretely, we can encode all logical qubits throughout the computation into a single QLDPC code block with constant rate. Logical gates are implemented using our state-preparation method together with gate teleportation. This approach is more compact than the protocol of Ref.~\cite{gottesman2013fault}, where the logical qubits are stored across many QLDPC code blocks and each block must be tracked and addressed during the computation.

We note that our state-preparation procedure itself encodes logical qubits into multiple QLDPC blocks. However, because the state-preparation circuits are fixed in advance, this multi-block structure does not need to be adapted according to the target computation that we wish to run. The target computation can be executed on a single large QLDPC block by simply implementing the original logical circuit on that block.
\end{enumerate}

\paragraph{Fault-tolerant quantum communication:}
The theory of fault-tolerant quantum communication is concerned with channel coding for noisy channels, while assuming circuit-level noise in the encoder and decoder circuits of the channel code. The fault-tolerant quantum communication was initiated in Ref.~\cite{CChMH-FT2022}, establishing that it is possible to communicate with non-zero communication rates, under circuit-level noise. Moreover, the communication rate approaches the channel capacity as the circuit-level noise rate $\delta$ approaches to zero. The fault-tolerant communication was later extended to entanglement assisted scenarios in Ref.~\cite{belzig2023fault}, and explicit constructions of fault-tolerant communication codes were provided under a general non-Pauli noise in Ref.~\cite{christandl2024fault}.

\smallskip In these earlier works, a model of fault-tolerant communication is used, where the encoder of the channel code corresponds to preparing a quantum state, which is then transferred using copies of a noisy quantum channel. Albeit this state preparation is realized fault-tolerantly, its qubit overhead grows with respect to the size of its output system, i.e., the number of channel uses. Using our state preparation method from Theorem~\ref{thm:main-stprep}, we can make the qubit overhead constant in the state preparation circuit; therefore, significantly reducing the overhead in the encoder.

\smallskip We note that to make the overall overhead constant in fault-tolerant quantum communication, we need to realize the decoder circuit of the channel code with constant qubit overhead too. To do this, we need an encoding interface with constant qubit overhead that fault-tolerantly encodes an arbitrary quantum state in an arbitrary level $r$ of the QLDPC code. We believe that this can be achieved by running our decoding interface procedure in reverse, i.e., first encoding in some fixed level $\ovr$, and then increasing it one by one to reach $r$. A full analysis of the encoding interface is left for future works.

\section*{Acknowledgements}
MC and AG acknowledge financial support from Villum Fonden (Grant 10059 'QMATH') and the Novo Nordisk Foundation (Grant NNF20OC0059939 ‘Quantum for Life’).  
OF acknowledges financial support from the European Research Council (ERC Grant, Agreement No.~851716) and from a government grant managed by the Agence Nationale de la Recherche under the Plan France 2030 with the reference ANR-22-PETQ-0006. 
\printbibliography
\end{document}